\documentclass[11pt]{article}
\usepackage{amsmath,amssymb,bbm,amsthm,mathtools}
\usepackage{fullpage}
\usepackage{thm-restate,color,xcolor,xspace}
\usepackage{hyperref,cleveref}
\usepackage{graphicx}
\usepackage{algorithm,algorithmic}

\usepackage[shortlabels]{enumitem}

\usepackage{thmtools}

\newtheorem{theorem}{Theorem}[section]
\newtheorem{lemma}[theorem]{Lemma}
\newtheorem{definition}[theorem]{Definition}

\newtheorem{claim}[theorem]{Claim}
\newtheorem{subclaim}[theorem]{Subclaim}
\newtheorem{fact}[theorem]{Fact}

\Crefname{observation}{Observation}{Observations}
\Crefname{claim}{Claim}{Claims}
\Crefname{subclaim}{Subclaim}{Subclaims}
\Crefname{fact}{Fact}{Facts}
\Crefname{assumption}{Assumption}{Assumptions}
\newenvironment{subproof}[1][\proofname]{%
  \begin{proof}[#1]%
}{%
  \end{proof}%
}

\begin{document}

\title{Congestion-Approximators from the Bottom Up}
\author{
Jason Li\thanks{Carnegie Mellon University. email: \texttt{jmli@cs.cmu.edu}}
\and
Satish Rao\thanks{UC Berkeley. email: \texttt{satishr@berkeley.edu}}
\and
Di Wang\thanks{Google Research. email: \texttt{wadi@google.com}}
}
 \date{\today}
\maketitle

\begin{abstract}
We develop a novel algorithm to construct a congestion-approximator with polylogarithmic quality on a capacitated, undirected graph in nearly-linear time. Our approach is the first \emph{bottom-up} hierarchical construction, in contrast to previous \emph{top-down} approaches including that of R\"{a}cke, Shah, and Taubig~(SODA 2014), the only other construction achieving polylogarithmic quality that is implementable in nearly-linear time (Peng, SODA 2016). Similar to R\"{a}cke, Shah, and Taubig, our construction at each hierarchical level requires calls to an approximate max-flow/min-cut subroutine. However, the main advantage to our bottom-up approach is that these max-flow calls can be implemented directly \emph{without recursion}. More precisely, the previously computed levels of the hierarchy can be converted into a \emph{pseudo-congestion-approximator}, which then translates to a max-flow algorithm that is sufficient for the particular max-flow calls used in the construction of the next hierarchical level. As a result, we obtain the first non-recursive algorithms for congestion-approximator and approximate max-flow that run in nearly-linear time, a conceptual improvement to the aforementioned algorithms that recursively alternate between the two problems.

\end{abstract}

\section{Introduction}





The famous max-flow min-cut theorem\footnote{For the sake of this
introduction, a flow problem on a graph is a set of excesses and
deficits on vertices and a solution is a set of paths connecting
excess and deficits where no edge is used in more than one path.}
implies that any set of excesses or deficits in a graph can be
connected by disjoint paths if and only if  every cut in the graph
has more capacity in the edges that cross the cut than the ``demand''
that is required to cross it. One direction is easy to see (if a cut does not have
enough capacity, clearly one cannot route demand paths across it), and
the other can be established by linear programming duality (or
explicitly using an algorithm as was done by Ford and
Fulkerson in 1956 \cite{ford1956maximal}.)

Framed slightly differently, a cut where the ratio of the demand across it
to the capacity of the cut is $c$ implies that any flow satisfying the demands
requires at least $c$ units of flow go through some edge. Again, the
max-flow min-cut theorem implies that there is always such a cut
if the optimal flow has congestion $c$.  Here, the congestion of a flow
is the maximum flow that is routed on any edge in the graph. 

More generally, the set of feasible
demands in a graph is captured by the (exponentially sized) set of all cuts in the graph.  Remarkably, it was shown that polynomially many
cuts could approximately (within logarithmic factors) characterize the
congestion of any set of demands \cite{Raecke,RaeckeTree}.  Indeed, the total size of cuts was nearly-linear\footnote{Nearly-linear means $O(m\log^c n)$ where $c$ is a constant. Almost-linear means $O(m^{1+\epsilon})$ for some $\epsilon = o(1)$. } in \cite{Raecke,RST14} and could be computed in nearly-linear time as announced in 2014 by Peng \cite{Pengarxiv}.

More precisely, an $\alpha$-congestion-approximator for a graph is a set of
cuts where for every set of demands, the minimum ratio (over cuts in the $\alpha$-congestion-approximator) of the demands crossing
the cut and the capacity of the cut determines the minimum congestion of
the {\em optimal} flow (that routes the demands) to within an $\alpha$ factor.

In this paper, we provide an $\alpha$-congestion-approximator (along with an
routing scheme) that can be constructed in nearly-linear time
with an $\alpha$ of $O(\log^{10} n)$.  The number of logarithmic
factors is not optimized but the structure is significantly simpler
than the construction by Peng~\cite{Pengarxiv}.

The top-down frame of the decomposition of \cite{RST14} requires that
one compute approximate maximum flows at the very top level which
entailed a non-linear runtime. A breakthrough result of
Sherman \cite{She13} showed how to use congestion-approximators to
find approximate maximum flows.\footnote{Simultaneously, Kelner
et. al. \cite{kelner2014almost} used a dual version of congestion approximators
called oblivous routers to give efficient algorithms
for approximate maximum flow.}  Peng \cite{Peng} used the result of
Sherman \cite{She13} inside the top-down partitioning method based
on \cite{RST14} for constructing congestion-approximators.  There is a
chicken and egg problem here as the two methods need each other, and a
costly recursion was used to combine them based on ultra-sparsifiers
which we discuss a bit more below.

Our result proceeds in a bottom-up iterative fashion; indeed, the
bottom level is a ``weak expander decomposition'' which can be
computed using trivial congestion-approximators consisting of
singleton vertices.\footnote{We note expander decompositions
themselves have found wide application, including in the breakthrough
near linear time algorithm of \cite{maxflow} for maximum flow. Also,
the high-level idea of our congestion-approximator is similar to  
hierachical expander decomposition \cite{goranci2021expander}, and our techniques may be useful
in getting from almost linear to near linear.}  Then we
proceed level by level, using the lower level clusterings as ``pseudo''-congestion-approximators for the next.  The frame is simple, but
``around the edges'' both literally (the edges of the clusters) and
metaphorically there are technical issues which require some
attention.

Still, given the power that congestion-approximators provide with respect to understanding the structure of graphs, finding more efficient
constructions is important and an effective bottom-up or clustering
approach seems a natural path to follow.  We make the first and a
substantive step on this path in the decade since the
announcement of Peng's also polylogarithmic nearly
linear time algorithm \cite{Pengarxiv}.


Recent progress on maximum flow can be compared to developments in nearly-linear time Laplacian solvers \cite{ST00},
which were initially very complex with many logarithmic factors.  But over
the years, new tools were developed that both improved the running times and allowed for better understanding of graphs as well as having broader application.  See, for example,  \cite{koutis2011nearly,kyng2016approximate,cohen2014solving,jambulapati2021ultrasparse}. Progress has been
much slower for maximum flow in comparison.  
Perhaps one reason is that solutions to 
Laplacian linear systems are $\ell_2$ optimizing flows that involve the Euclidean norm which is a
fair bit simpler than the $\ell_\infty$ norm central to maximum flow or norms based on more general
convex bodies. We believe this result to be a step in the process of understanding maximum flows better.

We proceed with a discussion of previous work which consists of a
remarkable series of developments. 

\subsection{Previous work}
The maximum flow, minimum cut theorem is nice and remarkable in that
one can find a single cut that establishes the optimality of the flow or
``routing'', and the optimal routing establishes a tight lower bound
on the size of any cut.\footnote{A wrinkle is that there could be
many (even an exponential) number of minimum cuts or feasible optimal
flows. But any optimal routing or any optimal cut establishes the
optimality of the ``dual'' object.}

We note that this theorem was extended to multi-commodity flow
in an approximate sense in \cite{LR99}, \cite{AKRR}
and \cite{linial1995geometry} where flows and corresponding cuts are
shown to be related by an $O(\log n)$ factor.  The first paper gives a
method for finding the sparsity or conductance of a graph to within an
$O(\log n)$ factor by using multicommodity flow to embed a complete graph, and the others extend the techniques to give approximations between cuts and solutions for arbitrary
multicommodity flow instances.  In a parallel thread,  the
mixing times of random walks or eigenvalues (which involve flow
problems with an $\ell_2$ objective) are related to cuts via classical
and tremendously impactful results of
Cheeger \cite{Cheeger}. Combining eigenvalues and linear programming
methods through semidefinite programming yields improved approximate
relationships between embeddings and cuts of roughly $O(\sqrt{\log
n})$ \cite{arora2009expander,arora2005euclidean}. Fast versions of these
methods involve the cut-matching game developed in \cite{arora2010logn,KRV09,OSVV,Shermansparse} and used directly
in this work.

As we noted previously, one can view the (exponentially sized) set of feasible sets of demands in a graph as being a very general measure of a
graph's capabilities.  And as also mentioned above, R\"{a}cke in \cite{Raecke}
showed that a polynomial sized set of cuts and
corresponding pre-computed routings approximately model the
congestion required for any of (exponentially many) sets of possible
demands. 




A bit more formally, \cite{Raecke} provides an oblivious routing
scheme, where the scheme can obliviously route any set of
demands with no more than $\alpha$ times as much congestion as the
optimal routing. Here, oblivious roughly means that for any demand
pair the routing is done without considering any other demand
pair.~\cite{Raecke} also provides a decomposition where for any set of
demands the maximum congestion (i.e., ratio of demand crossing to edge
capacity) on any cut in the decomposition is within a factor of
$\alpha$ of the congestion needed to route those demands.

The value for $\alpha$ in \cite{Raecke} is $O(\log^3 n)$ but was
non-constructive: however, constructive and improved schemes were
quickly developed with $\alpha$ of $O(\log^2 n \log\log n)$
in \cite{BKR03,HHR}.

An alternative scheme also by R\"{a}cke \cite{RaeckeTree}, gave a
remarkably simple oblivious routing scheme, that consisted of $O(m)$
trees.  The oblivious routing was simply to route any demand by
splitting the flow among the trees. Of course, each tree implicitly
corresponds to a laminar family of cuts. Still the total size of both
the routings and the implicit total size of the cuts is quadratic or
worse. Moreover, the time complexity for its construction was also at least quadratic.  The
approximation factor $\alpha$ for this scheme was $O(\log n)$. We will
refer to this as R\"{a}cke's tree scheme.



Madry, in \cite{Madry}, gave a nearly-linear time algorithm that
produced a congestion-approximator which has almost-linear size and
running time at the cost of having approximation factor of
$O(n^{\epsilon})$ that is based on R\"{a}cke tree scheme.  A central
idea is the use of ultrasparsifiers which were introduced by Spielman
and Teng \cite{ST00} in their breakthrough results on linear time
solvers for Laplacian linear systems.  An ultrasparsifier is
formed by taking a certain kind of spanning tree (called low-stretch),
making clusters from small connected components of the tree, and
sampling a very small set of non-tree edges between the clusters.
Such a graph approximates the cuts to within a small factor in the
sense that cuts in the original graph and the new one have
approximately the same size.  This allows a (complicated) recursion
where one use several ultrasparsifiers to approximate the cuts in the
graph and recursively produce oblivious routing schemes for each of
the ultrasparsifiers.  Again, this approach is more efficient in time
than R\"{a}cke's tree scheme at the cost of a worse approximation
factor.  The scheme enabled the approximate solution to a host of
problems in almost linear time. This contribution was remarkable.

In sum, Madry \cite{Madry} produced a $\alpha =
O(n^{\epsilon})$-congestion approximator in the almost linear running time of
$O(m^{1+\epsilon})$.\footnote{The $\epsilon$ is subconstant, roughly
$1/\sqrt{\log n}$ which means an $m^{\epsilon}$ factor is larger than
any polylogarithmic factor.} The obstruction to getting to polylogarithmic overheads in Madry's
scheme is the need to recurse on several ultrasparsifiers to keep the
approximation from blowing up during the recursion. 

In another exciting development, Sherman \cite{She13} used
Madry's \cite{Madry} structure to produce an almost-linear time
algorithm for $1+\epsilon$-approximate undirected maximum flow. This
is remarkable in that he reduced Madry's approximation factor of
$O(n^{\epsilon})$ to $(1 + \epsilon)$. The approximation factor in
Madry's construction translates into a running time overhead, which
results in the almost-linear running time we mentioned. We point out
that there is close relationship to the work of Spielman and
Teng \cite{ST00}, who used ultrasparsifiers to spectrally approximate a Laplacian system to provide algorithms
whose running times depend on the approximation quality while
obtaining arbitrarily precise solutions.  In that case, the dependence of running time
on precision was logarithmic whereas Sherman's algorithm suffers a
$O(1/\epsilon^2)$ dependence on the error $\epsilon$.  To be sure, in
solving linear systems the idea of using an ``approximate'' graph in
improving condition numbers was known as preconditioning, but for
maximum flow or $\ell_1$ optimization this was a striking development.

Sherman's method (could be seen) to use the multiplicative weights
framework (see \cite{arora2012multiplicative}) to route flow across cuts so that every cut has near zero
residual demand across it, in particular, $\delta = \epsilon/\alpha$
fraction of its capacity.  The congestion approximator then ensures
that the remaining flow can be routed with $\alpha \cdot \delta
= \epsilon$ congestion and one can recurse with slightly more than
$\epsilon$ extra congestion.

We note that Sherman's result (with Madry's construction) bypasses longstanding barriers to faster maximum flow
algorithms. Since the work of Dinitz \cite{Dinitz} in 1973,
algorithms had to pay either for path length or for number of
paths. Dinitz himself traded this off to get a $O(m^{3/2})$ flow
algorithm \cite{christiano2011electrical,lee2013new}. Using linear time solvers to do electrical flow, one could
eke out a few more paths simultaneously, but one still was pretty
stuck at $O(m^{4/3})$.  The congestion-approximator and the
multiplicative weights optimization method bypasses these
obstructions.

In 2014, R\"{a}cke, Shah, and Taubig \cite{RST14} made progress on
R\"{a}cke's decomposition based approach.  In particular, they showed how
to produce a decomposition with approximation parameter of $O(\log^4
n)$ using maximum flow computations of total size $O(m \log^5 n)$.
They use the same frame as R\"{a}cke \cite{Raecke}, but use single
commodity flows in the context of the cut-matching game \cite{KRV,OSVV} which, as previously mentioned, can replace multicommodity
flows in finding sparse cuts and routings.

Note, that a congestion-approximator can be used to compute flows efficiently
where flows can now be used to compute congestion-approximators. This
suggests recursion, but the construction in \cite{RST14} (and indeed previous ones) were very much top down.
That is, the first level of the decomposition (which itself is a tour
de force) falls victim to the fact that typical paths in the flows that find them can be
long and somewhat abundant. Thus, one critically needs something like a congestion-approximator
right away to even compute the top level of the decomposition.

Still, Peng \cite{Peng} was able to get to a nearly-linear running time
recursive method by using ultrasparsifiers and Sherman's approximate
maximum flow algorithm.  He constructs an ultrasparsifier of size $O(m
/\log^c m)$ with polylogarithmic approximation factor, recursively
computes a cut sparsifier for the resulting graph, and then argues that
the combination of Sherman's algorithm with both the ultrasparsifier
and the congestion-approximator can be used to compute an approximate
maximum flow in $O(m \log^c m)$ time.  Again, the key here is that
Sherman \cite{She13} allows him to use polylogarithmic overhead to
combat the polylogarithmic approximation in both the ultrasparsifer
and in the congestion-approximator of the ultrasparsifier. Still, the
recursion is a bear, resulting in an admittedly unoptimized runtime of $O(m \log^{41} n)$. 

As noted previously, our method is bottom up from the start and avoids the costly recursion required above. We have not computed the exponent of the logarithmic
factors due to, for example, using the fair cuts method of \cite{LNPS23} as a black box. Still, the
clustering approach is more natural and we expect is a useful frame.

We proceed with a technical overview of our result and methods.

\section{Technical Overview}

Our main conceptual and technical contribution is a novel congestion approximator that is constructed bottom-up in a hierarchical fashion. We start with an informal construction in the theorem below, which is later formalized in \Cref{thm:cut-approx}. See \Cref{fig:cut-approx} for an illustration.

\begin{figure}
\begin{center}
\includegraphics[scale=.9]{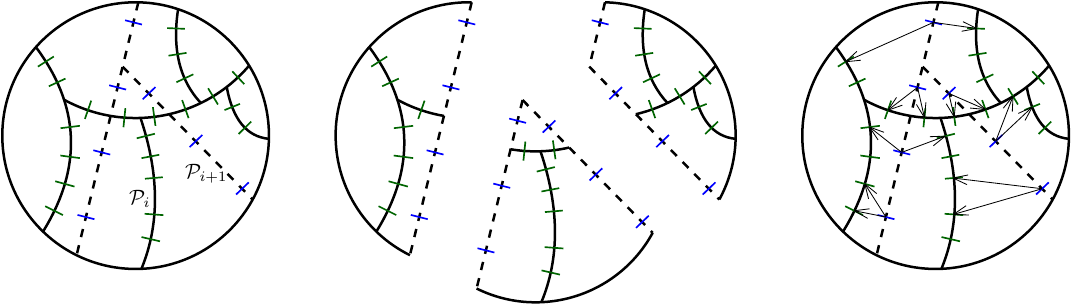}
\caption{On the left, partitions $\mathcal P_i$ (solid) and $\mathcal P_{i+1}$ (dotted) are shown. On the middle, the marked edges for each cluster mix simultaneously in $G$ (property~(\ref{item:cut-approx-2})). On the right, a flow from the inter-cluster edges of $\mathcal P_{i+1}$ to the inter-cluster edges of $\mathcal P_i$ is displayed (property~(\ref{item:cut-approx-3})); assuming edges are unit-weight, each flow path carries a half-unit of flow.}
\label{fig:cut-approx}
\end{center}
\end{figure}

\begin{theorem}[Informal \Cref{thm:cut-approx}]\label{thm:cut-approx-informal}
Consider a capacitated graph $G=(V,E)$. Suppose there exist partitions $\mathcal P_1,\mathcal P_2,\ldots,\mathcal P_L$ of $V$ such that
 \begin{enumerate}
 \item $\mathcal P_1$ is the partition $\{\{v\}:v\in V\}$ of singleton clusters, and $\mathcal P_L$ is the partition $\{V\}$ with a single cluster.\label{item:cut-approx-1-informal}
 \item For each $i\in[L-1]$, for each cluster $C\in\mathcal P_{i+1}$, the inter-cluster edges of $\mathcal P_i$ internal to $C$, together with the boundary edges of $C$, \emph{mix}\footnote{Informally, a set of edges \emph{mixes} if there is a low-congestion multi-commodity flow between the set of edges whose demand pairs form an expander. In other words, there is an expander flow (in the~\cite{arora2009expander} sense) between the set of edges. Our formal definition is in the preliminaries and only considers single-commodity flows.} in the graph $G$. Moreover, the mixings over all clusters $C\in\mathcal P_{i+1}$ have congestion $\alpha$ simultaneously.\label{item:cut-approx-2-informal}
 \item For each $i\in[L-1]$, there is a flow in $G$ with congestion $\beta$ such that each inter-cluster edge in $\mathcal P_{i+1}$ sends its capacity in flow and each inter-cluster edge in $\mathcal P_i$ receives at most half its capacity in flow. \label{item:cut-approx-3-informal}
 \end{enumerate}
For each $i\in[L]$, let partition $\mathcal R_{\ge i}$ be the common refinement of partitions $\mathcal P_i,\mathcal P_{i+1},\ldots,\mathcal P_L$, i.e.,
\[ \mathcal R_{\ge i}=\{C_i\cap C_{i+1}\cap\cdots\cap C_L:C_i\in\mathcal P_i,\,C_{i+1}\in\mathcal P_{i+1},\ldots,\,C_L\in\mathcal P_L,\,C_i\cap\cdots\cap C_L\ne\emptyset\}. \]
Then, their union $\mathcal C=\bigcup_{i\in[L]}\mathcal R_{\ge i}$ is a congestion-approximator with quality $16L^2\alpha\beta$.
\end{theorem}

To understand the construction, first consider the case $i=1$. By property~(\ref{item:cut-approx-1-informal}), $\mathcal P_1$ is the partition of singleton clusters. Since the inter-cluster edges of $\mathcal P_1$ are precisely all the edges, property~(\ref{item:cut-approx-2}) says that for each cluster $C\in\mathcal P_{i+1}$, the edges internal to $C$, together with the boundary edges of $C$, ``mix'' in the graph $G$. In other words, the set of edges with at least one endpoint in $C$ mixes in $G$. In the extreme case $C=V$, the entire edge set mixes in $G$, so the graph $G$ is an expander.\footnote{We do not define expanders in this paper since we do not need their precise definition. The connection to expanders is only stated as motivation for readers familiar with the concept.} In general, we can informally say that each cluster $C$ is a sort of weak-expander, and the partition $\mathcal P_2$ is a weak-expander decomposition of graph $G$.\footnote{A key difference (from an actual expander decomposition) is that the (routing to certify the) mixing of each cluster is not required to be fully inside its \emph{induced} subgraph, although the full graph $G$ still needs to have the capacity to support the mixing of all the clusters simultaneously.}

Now consider property~(\ref{item:cut-approx-3}), which establishes a flow starting from the inter-cluster edges of $\mathcal P_2$ such that each edge in $G$ receives at most half its capacity in flow. We claim that this statement is very natural and follows almost immediately from property~(\ref{item:cut-approx-2}) with one mild assumption: for each cluster in $\mathcal P_{i+1}$, the total capacity of boundary edges is much smaller than the total capacity of internal edges.\footnote{From the perspective of expander decompositions, this assumption is very natural. Expander decompositions require that the total capacity of inter-cluster edges is small relative to the total capacity of all edges. We are simply extending this property to hold for each cluster, not just globally over all clusters.} Indeed, from the mixing of the internal and boundary edges of $C$, we can spread a flow from the boundary edges of $C$ such that each internal edge receives flow proportional to its capacity. As long as the boundary edges have small total capacity, the total flow source is also small, so each internal edge receives a small amount of flow relative to its capacity. Finally, since the clusters mix simultaneously, we can compose the corresponding flows for each cluster and still ensure small congestion.

Now consider a general level $i\ge2$. Recall from property~(\ref{item:cut-approx-2}) that for each cluster $C\in\mathcal P_{i+1}$, the inter-cluster edges of $\mathcal P_i$ inside $C$, together with the boundary edges of $C$, ``mix'' in the graph $G$ (see \Cref{fig:cut-approx}, middle). We can interpret this statement (again) as a weak-expander decomposition where expansion is measured with respect to a subset of edges, namely the inter-cluster edges of partition $\mathcal P_i$ together with the boundary edges of a cluster in $\mathcal P_{i+1}$. Property~(\ref{item:cut-approx-3}) says that we can spread flow from the inter-cluster edges of $\mathcal P_{i+1}$ to the inter-cluster edges of $\mathcal P_i$ such that each edge in $\mathcal P_i$ receives a small amount of flow (see \Cref{fig:cut-approx}, right). Once again, we can show that it follows from property~(\ref{item:cut-approx-2}) with the following mild assumption: for each cluster in $\mathcal P_{i+1}$, the total capacity of boundary edges is much smaller than the total capacity of inter-cluster edges of $\mathcal P_i$ that are internal to this cluster.

Overall, the partitions $\mathcal P_1,\mathcal P_2,\ldots,\mathcal P_L$ can be viewed as hierarchical expander decompositions where each partition $\mathcal P_{i+1}$ is a weak-expander decomposition relative to the inter-cluster edges of partition $\mathcal P$. It is instructive to compare our partitioning to the \emph{expander hierarchy} construction of~\cite{GRST21}, where clusters of the previous partition are \emph{contracted} before building the next partition. While~\cite{GRST21} can also extract a congestion approximator from their expander hierarchy construction, their quality is $n^{o(1)}$ because they lose a multiplicative factor from contracting vertices on each level of the hierarchy. To avoid this multiplicative blow-up per level, we do not contract vertices, so our partitioning is not truly hierarchical: a cluster in partition $\mathcal P_i$ may be ``cut'' into many components by the next partition $\mathcal P_{i+1}$ (see \Cref{fig:cut-approx}, left). While a hierarchical construction is not required, it is a useful property to have when analyzing the quality of the final congestion-approximator. To establish a hierarchy, our key insight is to consider the \emph{common refinement} of partitions $\mathcal P_i,\mathcal P_{i+1},\ldots,\mathcal P_L$, which we name $\mathcal R_{\ge i}$. The partitions $\mathcal R_{\ge i}$ over all $i$ are now hierarchical by construction, and we show that the union of all refinements $\mathcal R_{\ge i}$ is a congestion-approximator with good quality.

We emphasize that our conceptual idea of not contracting clusters and looking at common refinements is novel and may have future applications to bottom-up constructions of hierarchical objects, especially for obtaining polylogarithmic approximations, meaning that one cannot lose a multiplicative factor at each level. Given that \cite{GRST21} has popularized bottom-up hierarchical approaches, and that their methods so far can only obtain $n^{o(1)}$-factors due to multiplicative errors across levels, we believe that our ideas are promising for future development in this area.

\subsection{Bottom-Up Construction}

As mentioned previously, the partition $\mathcal P_2$ is a weak-expander decomposition of the graph, so it can be computed in nearly-linear time using off-the-shelf expander decomposition algorithm that avoid any black-box call to max-flow~\cite{SW19}. For partitions $\mathcal P_3$ onwards, expansion is measured with respect to a subset of edges, so simple expander decomposition algorithms no longer suffice. Instead, we have to resort to expander decomposition algorithms that make black-box calls to (approximate) max-flow. Na\"ively, these max-flows can be computed recursively, resulting in a congestion approximator algorithm that makes recursive calls to max-flow, similar to~\cite{RST14}. Our key insight is that these max-flow instances are actually well-structured enough that recursion is unnecessary. In particular, to construct partition $\mathcal P_{i+1}$, the first $i$ partitions $\mathcal P_1,\mathcal P_2,\ldots,\mathcal P_i$---which the algorithm has already computed---can be converted to a \emph{pseudo-congestion approximator}, which then translates to a max-flow algorithm sufficient for these well-structured instances.

\section{Preliminaries}\label{sec:prelim}

We are given an undirected, capacitated graph $G=(V,E)$. The graph has $n$ vertices and $m$ edges, and each edge $e\in E$ has capacity $c(e)$ in the range $[1,W]$. For a set $C\subseteq V$, let $\partial C$ denote the set of edges with exactly one endpoint in $C$, and let $\delta C=\sum_{e\in\partial C}c(e)$ denote the total capacity of edges in $\partial C$. For a vertex $v\in V$, let $\deg(v)$ denote the capacitated degree of vertex $v$, which is also equal to $\delta\{v\}$. We sometimes write $\partial_GC$, $\delta_GC$, and $\deg_G(v)$ to emphasize that the values are with respect to graph $G$.

For a given edge set $F\subseteq E$, let $\partial_FC$, $\delta_FC$, and $\deg_F(v)$ denote the corresponding values on the subgraph of $G$ with edge set $F$. For a different graph $H$, let $\partial_HC$, $\delta_HC$, and $\deg_H(v)$ denote the corresponding values on graph $H$. We never remove the subscripts $F$ and $H$ to avoid confusion with the original graph $G$.

For a partition $\mathcal P$ of the vertex set $V$, we call each part of the partition a \emph{cluster}. Let $\partial\mathcal P$ denote the set of edges whose endpoints belong to different clusters; we can also write $\partial\mathcal P=\bigcup_{C\in\mathcal P}\partial C$. We also define $\delta\mathcal P$ as the total capacity of edges in $\partial\mathcal P$.

Let $U$ be an arbitrary set of vertices. For a vector $\mathbf x\in\mathbb R^U$, let $\mathbf x(v)$ denote its value on entry $v\in U$. For a vertex subset $S\subseteq U$, define $\mathbf x(S)=\sum_{v\in S}\mathbf x(v)$, and define $\mathbf x|_S\in\mathbb R^U$ as the vector $\mathbf x$ restricted to $S$, i.e., $\mathbf x|_S(v)=\mathbf x(v)$ for all $v\in S$ and $\mathbf x|_S(v)=0$ for all $v\notin S$. For two vectors $\mathbf s,\mathbf t\in\mathbb R^U$, by $\mathbf s\le \mathbf t$ we mean entry-wise inequality, i.e., $\mathbf s(v)\le\mathbf t(v)$ for all $v\in U$. For a vector $\mathbf s\in\mathbb R^U$, let $|\mathbf s|\in\mathbb R^U$ be the vector with entry-wise absolute values, i.e., $|\mathbf s|(v)=|\mathbf s(v)|$.

A \emph{demand vector} is a vector $\mathbf b\in\mathbb R^U$ whose entries sum to $0$, i.e., $\mathbf b(U)=0$. A flow $f$ \emph{routes} demand $\mathbf b$ if each vertex $v\in U$ receives a net flow of $\mathbf b(v)$ in the flow $f$. A flow $f$ has \emph{congestion} $\alpha$ if the amount of flow sent along each (undirected) edge is at most $\alpha$ times the capacity of that edge. Given a flow $f$, a \emph{path decomposition} of flow $f$ is a collection of directed, capacitated paths such that for any two vertices $u,v\in U$ connected by an edge $e$, the amount of flow that $f$ sends from $u$ to $v$ equals the total capacity of (directed) paths that contain edge $e$ in the direction from $u$ to $v$.

A \emph{vertex weighting} is a vector $\mathbf d\in\mathbb R^U_{\ge0}$, i.e., all entries in $\mathbf d$ are nonnegative. The vertex weighting $\mathbf d\in\mathbb R^U_{\ge0}$ \emph{mixes} in graph $G$ with \emph{congestion} $\alpha$ if for any demand $\mathbf b\in\mathbb R^U$ satisfying $|\mathbf b|\le \mathbf d$, there is a flow routing $\mathbf b$ with congestion $\alpha$.\footnote{There is a close connection between the concept of mixing and \emph{expander graphs}, though we do not need the definition of expanders in this paper.} A collection $\{\mathbf d_1,\mathbf d_2,\ldots,\mathbf d_\ell\}$ of vertex weightings mixes \emph{simultaneously} with congestion $\alpha$ if for any demands $\mathbf b_1,\mathbf b_2,\ldots,\mathbf b_\ell$ with $|\mathbf b_i|\le\mathbf d_i$ for each $i\in[\ell]$, there is a flow routing demand $\mathbf b_1+\mathbf b_2+\cdots+\mathbf b_\ell$ with congestion $\alpha$.

Throughout the paper, we view vectors in $\mathbb R^U$ and functions from $U$ to $\mathbb R$ as interchangeable. In particular, we sometimes treat the degree function $\deg:V\to\mathbb R_{\ge0}$ as a vector $\deg\in\mathbb R^V_{\ge0}$. In particular, $\deg$ is a vertex weighting.

\paragraph{Congestion-approximators and approximate flow.}

Given a graph $H=(U,F)$, a \emph{congestion-approximator} $\mathcal C$ of \emph{quality} $\alpha$ is a collection of subsets of $U$ such that for any demand $\mathbf b$ satisfying $|\mathbf b(C)|\le\delta_HC$ for all $C\in\mathcal C$, there is a flow routing demand $\mathbf b$ with congestion $\alpha$. Through Sherman's framework~\cite{She13,She17}, a congestion-approximator of quality $\alpha$ translates to a $(1+\epsilon)$-approximate maximum flow algorithm with running time $\tilde O(\epsilon^{-1}\alpha m)$.\footnote{The $\tilde O(\cdot)$ notation hides polylogarithmic factors in $n$.} In our paper, it is most convenient to work with a stronger variant of approximate min-cut/max-flow called \emph{fair cut/flow}~\cite{LNPS23}, which is formally defined in \Cref{sec:cut-flow-algo}.

\section{Bottom-Up Congestion-Approximator}

In this section, we formally state our congestion-approximator construction and prove its quality guarantee. See \Cref{fig:cut-approx} for an illustration. The most important distinction is that we define and analyze mixing on \emph{vertex weightings}, not edges. This simplifies the notation since we can avoid working with the subdivision graph in~\cite{RST14}. For example, in property~(\ref{item:cut-approx-2}) below, the mixing of the vertex weighting $\mathrm{deg}_{\partial\mathcal P_i\cup\partial C}|_C$ is conceptually equivalent to the mixing of the edges of $\partial\mathcal P_i$ internal to $C$ together with the boundary edges $\partial C$.

\begin{restatable}{theorem}{CutApprox}\label{thm:cut-approx}
Consider a capacitated graph $G=(V,E)$, and let $\alpha\ge1$ and $\beta\ge1$ be parameters. Suppose there exist partitions $\mathcal P_1,\mathcal P_2,\ldots,\mathcal P_L$ of $V$ such that
 \begin{enumerate}
 \item $\mathcal P_1$ is the partition $\{\{v\}:v\in V\}$ of singleton clusters, and $\mathcal P_L$ is the partition $\{V\}$ with a single cluster.\label{item:cut-approx-1}
 \item For each $i\in[L-1]$, the collection of vertex weightings $\{\mathrm{deg}_{\partial\mathcal P_i\cup\partial C}|_C\in\mathbb R^V_{\ge0}:C\in\mathcal P_{i+1}\}$ mixes simultaneously in $G$ with congestion $\alpha$.\label{item:cut-approx-2}
 \item For each $i\in[L-1]$, there is a flow in $G$ with congestion $\beta$ such that each vertex $v\in V$ sends $\deg_{\partial\mathcal P_{i+1}}(v)$ flow and receives at most $\frac12\deg_{\partial\mathcal P_i}(v)$ flow. \label{item:cut-approx-3}
 \end{enumerate}
For each $i\in[L]$, let partition $\mathcal R_{\ge i}$ be the common refinement of partitions $\mathcal P_i,\mathcal P_{i+1},\ldots,\mathcal P_L$, i.e.,
\[ \mathcal R_{\ge i}=\{C_i\cap C_{i+1}\cap\cdots\cap C_L:C_i\in\mathcal P_i,\,C_{i+1}\in\mathcal P_{i+1},\ldots,\,C_L\in\mathcal P_L,\,C_i\cap\cdots\cap C_L\ne\emptyset\}. \]
Then, their union $\mathcal C=\bigcup_{i\in[L]}\mathcal R_{\ge i}$ is a congestion-approximator with quality $16L^2\alpha\beta$.
\end{restatable}

In \Cref{sec:construct-next-partition}, we develop an efficient algorithm to construct partitions $\mathcal P_1,\mathcal P_2,\ldots,\mathcal P_L$ with $L=O(\log n)$. This algorithm builds the partitions iteratively in the order $\mathcal P_1,\mathcal P_2,\ldots,\mathcal P_L$. For technical reasons explained in \Cref{sec:construct-next-partition}, we require the following analogue of \Cref{thm:cut-approx} where $\mathcal P_L$ is not necessarily the partition $\{V\}$. Note that assumptions~(\ref{item:cut-approx2-2}) and~(\ref{item:cut-approx2-3}) remain unchanged below. The key difference is that $\mathcal C$ is no longer a congestion-approximator, but a \emph{pseudo-congestion-approximator} whose precise guarantee is stated below.

\begin{lemma}\label{lem:cut-approx}
Consider a capacitated graph $G=(V,E)$, and let $\alpha\ge1$ and $\beta\ge1$ be parameters. Suppose there exist partitions $\mathcal P_1,\mathcal P_2,\ldots,\mathcal P_L$ such that
 \begin{enumerate}
 \item $\mathcal P_1$ is the partition $\{\{v\}:v\in V\}$ of singleton clusters.
 \item For each $i\in[L-1]$, the collection of vertex weightings $\{\mathrm{deg}_{\partial\mathcal P_i\cup\partial C}|_C\in\mathbb R^V_{\ge0}:C\in\mathcal P_{i+1}\}$ mixes simultaneously in $G$ with congestion $\alpha$.\label{item:cut-approx2-2}
 \item For each $i\in[L-1]$, there is a flow in $G$ with congestion $\beta$ such that each vertex $v\in V$ sends $\deg_{\partial\mathcal P_{i+1}}(v)$ flow and receives at most $\frac12\deg_{\partial\mathcal P_i}(v)$ flow. \label{item:cut-approx2-3}
 \end{enumerate}
For each $i\in[L]$, let partition $\mathcal R_{\ge i}$ be the common refinement of partitions $\mathcal P_i,\mathcal P_{i+1},\ldots,\mathcal P_L$, i.e.,
\[ \mathcal R_{\ge i}=\{C_i\cap C_{i+1}\cap\cdots\cap C_L:C_i\in\mathcal P_i,\,C_{i+1}\in\mathcal P_{i+1},\ldots,\,C_L\in\mathcal P_L,\,C_i\cap\cdots\cap C_L\ne\emptyset\}. \]
Consider their union $\mathcal C=\bigcup_{i\in[L]}\mathcal R_{\ge i}$. Then, for any demand $\mathbf b\in\mathbb R^V$ satisfying $|\mathbf b(C)|\le\delta C$ for all $C\in\mathcal C$, there exists a demand $\mathbf b'\in\mathbb R^V$ satisfying $|\mathbf b'|\le \deg_{\partial\mathcal P_L}$ and a flow routing demand $\mathbf b-\mathbf b'$ with congestion $16L^2\alpha\beta$. 
\end{lemma}

Instead of proving \Cref{thm:cut-approx} directly, we prove \Cref{lem:cut-approx} which is needed for the algorithm. Before we do so, we first establish that \Cref{lem:cut-approx} indeed implies \Cref{thm:cut-approx}.
\begin{proof}[Proof (\Cref{lem:cut-approx}$\implies$\Cref{thm:cut-approx})]
Consider partitions $\mathcal P_1,\mathcal P_2,\ldots,\mathcal P_L$ that satisfy the assumptions of \Cref{thm:cut-approx}. For a given demand $\mathbf b\in\mathbb R^V$ satisfying $|\mathbf b(C)|\le\delta C$ for all $C\in\mathcal C$, we want to establish a flow routing demand $\mathbf b$ with congestion $16L^2\alpha\beta$. \Cref{thm:cut-approx} then follows from the definition of congestion-approximator.

Apply \Cref{lem:cut-approx} to the partitions $\mathcal P_1,\mathcal P_2,\ldots,\mathcal P_L$ and the demand $\mathbf b$. We obtain a demand $\mathbf b'\in\mathbb R^V$ satisfying $|\mathbf b'|\le \deg_{\partial\mathcal P_L}$ and a flow $f$ routing demand $\mathbf b-\mathbf b'$ with congestion $16L^2\alpha\beta$. By assumption~(\ref{item:cut-approx-1}) of \Cref{thm:cut-approx}, we have $\mathcal P_L=\{V\}$, which implies that $\partial\mathcal P_L=\emptyset$. Since $|\mathbf b'|\le \deg_{\partial\mathcal P_L}=\mathbf0$, we must have $\mathbf b'=\mathbf0$. It follows that the flow $f$ routes demand $\mathbf b$ with congestion $16L^2\alpha\beta$, finishing the proof.
\end{proof}

For the rest of this section, we prove \Cref{lem:cut-approx}. We begin with two helper claims that establish structure on the sets $\mathcal R_{\ge i}$.
\begin{claim}\label{clm:refinement}
For all $i,j\in[L]$ with $i<j$, the partition $\mathcal R_{\ge i}$ of $V$ is a refinement of the partition $\mathcal R_{\ge j}$, in the sense that each set in $\mathcal R_{\ge j}$ is the disjoint union of some sets in $\mathcal R_{\ge i}$. In particular, $\partial\mathcal R_{\ge i}\supseteq\partial\mathcal R_{\ge j}$.
\begin{proof}
Consider a set $C=C_j\cap C_{j+1}\cap\cdots\cap C_L\in\mathcal R_{\ge j}$ for some $C_j\in\mathcal P_j,C_{j+1}\in\mathcal P_{j+1},\ldots,C_L\in\mathcal P_L$. Since $\mathcal P_i,\mathcal P_{i+1},\ldots,\mathcal P_{j-1}$ are all partitions of $V$, the set $C$ is the disjoint union of all nonempty sets of the form $C_i\cap C_{i+1}\cap\cdots\cap C_{j-1}\cap C\in\mathcal R_{\ge i}$ for $C_i\in\mathcal P_i,C_{i+1}\in\mathcal P_{i+1},\ldots,C_{j-1}\in\mathcal P_{j-1}$. Therefore, $\mathcal R_{\ge i}$ is a refinement of $\mathcal R_{\ge j}$, and since refinements can only increase the boundary set, the second statement $\partial\mathcal R_{\ge i}\supseteq\partial\mathcal R_{\ge j}$ follows.
\end{proof}
\end{claim}
\begin{claim}\label{clm:difference}
For all $i\in[L-1]$, we have $\partial\mathcal R_{\ge i}\setminus \partial\mathcal R_{\ge i+1}\subseteq\partial\mathcal P_i$.
\end{claim}
\begin{proof}
Consider an edge $(u,v)\in\partial\mathcal R_{\ge i}\setminus \partial\mathcal R_{\ge i+1}$. Since $(u,v)\notin\partial\mathcal R_{\ge i+1}$, there is a set $C\in\mathcal R_{\ge i+1}$ containing both vertices $u$ and $v$. As in the proof of \Cref{clm:refinement}, write $C=C_{i+1}\cap C_{i+2}\cap\cdots\cap C_L\in\mathcal R_{\ge i+1}$ for some $C_{i+1}\in\mathcal P_{i+1},C_{i+2}\in\mathcal P_{i+2},\ldots,C_L\in\mathcal P_L$. The set $C$ is the disjoint union of all nonempty sets of the form $C_i\cap C\in\mathcal R_{\ge i}$ for some $C_i\in\mathcal P_i$. Since $u,v\in C$, both $u$ and $v$ belong to sets of this form. Since $(u,v)\in\partial\mathcal R_{\ge i}$, the sets containing $u$ and $v$ must be different. They can only differ in the set $C_i\in\mathcal P_i$, so $u$ and $v$ belong to different sets in $\mathcal P_i$, and we obtain $(u,v)\in\partial\mathcal P_i$ as promised.
\end{proof}

Let $\mathbf b\in\mathbb R^V$ be a flow demand satisfying $|\mathbf b(C)|\le\delta C$ for all $C\in\mathcal C$. We need to construct a demand $\mathbf b'\in\mathbb R^V$ satisfying $|\mathbf b'|\le \deg_{\partial\mathcal P_L}$ and a flow routing demand $\mathbf b-\mathbf b'$ with congestion $16L^2\alpha\beta$. 

The construction of the flow has $L-1$ iterations. On iteration $i\in[L-1]$, we construct a flow $f_i$ and a demand $\mathbf b_i$ such that
 \begin{enumerate}
 \item The flow $f_i$ routes demand $\mathbf b_{i-1}-\mathbf b_i$, where we initialize $\mathbf b_0=\mathbf b$ on iteration $i=1$.\label{item:flow-1}
 \item The flow $f_i$ has congestion $16L\alpha\beta$.\label{item:flow-2}
 \item For all $C\in\mathcal R_{\ge i+1}$, we have $(\mathbf b_{i-1}-\mathbf b_i)(C)=0$.\label{item:flow-3}
 \item The demand $\mathbf b_i$ satisfies $|\mathbf b_i|\le\deg_{\partial\mathcal R_{\ge i+1}}$.\label{item:flow-4}
 \end{enumerate}
The lemma below shows that properties~(\ref{item:flow-1}), (\ref{item:flow-2}), and~(\ref{item:flow-4}) alone are sufficient to prove \Cref{lem:cut-approx} with demand $\mathbf b'=\mathbf b_{L-1}$ and flow $f_1+f_2+\cdots+f_{L-1}$.
\begin{lemma}
Suppose that properties~(\ref{item:flow-1}), (\ref{item:flow-2}), and~(\ref{item:flow-4}) hold for each $i\in[L-1]$. Then, the demand $\mathbf b_{L-1}$ satisfies $|\mathbf b_{L-1}|\le \deg_{\partial\mathcal P_L}$, and the flow $f_1+f_2+\cdots+f_{L-1}$ routes demand $\mathbf b-\mathbf b_{L-1}$ with congestion $16L^2\alpha\beta$.
\end{lemma}
\begin{proof}
Observe that $\partial\mathcal R_{\ge L} = \partial\mathcal P_L$ by definition.  The demand $\mathbf b_{L-1}$ satisfies $|\mathbf b_{L-1}|\le \deg_{\partial\mathcal R_{\ge L}}=\deg_{\partial\mathcal P_L}$ by property~(\ref{item:flow-4}) on iteration $i=L-1$. By property~(\ref{item:flow-1}) over all $i\in[L-1]$, the flow $f_1+f_2+\cdots+f_{L-1}$ routes demand $(\mathbf b_0-\mathbf b_1)+(\mathbf b_1-\mathbf b_2)+\cdots+(\mathbf b_{L-2}-\mathbf b_{L-1})=\mathbf b-\mathbf b_{L-1}$. The congestion is at most the sum of the congestions of each $f_i$, which is $(L-1)\cdot 16L\alpha\beta\le16L^2\alpha\beta$ by property~(\ref{item:flow-2}).
\end{proof}
In order to establish the conditions above, we will use the following technical lemma:
\begin{restatable}{lemma}{ConstructFlow}\label{lem:construct-flow}
Consider an iteration $i\in[L-1]$ and a vector $\mathbf s\in\mathbb R^V$ such that
 \begin{enumerate}[(a)]
 \item $|\mathbf s|\le\deg_{\partial\mathcal R_{\ge i}}$.\label{item:construct-flow-1}
 \item $|\mathbf s(C)|\le\delta C$ for all $C\in\mathcal R_{\ge i+1}$.\label{item:construct-flow-2}
 \end{enumerate}
Then, we can construct a flow $f$ such that
 \begin{enumerate}[(i)]
 \item Flow $f$ routes demand $\mathbf s-\mathbf t$ for a vector $\mathbf t\in\mathbb R^V$ with $|\mathbf t|\le\deg_{\partial\mathcal R_{\ge i+1}}$.\label{item:technical-1}
 \item The flow $f$ has congestion $16L\alpha\beta$.
 \item For all $C\in\mathcal R_{\ge i+1}$, we have $(\mathbf s-\mathbf t)(C)=0$.\label{item:technical-3}
 \end{enumerate}
\end{restatable}
Before we prove this lemma, we first establish that it implies properties~(\ref{item:flow-1}) to~(\ref{item:flow-4}) above for appropriate $f_i$ and $\mathbf b_i$.
\begin{lemma}\label{lem:construct-f-b}
Assuming \Cref{lem:construct-flow}, we can construct $f_i$ and $\mathbf b_i$ satisfying properties~(\ref{item:flow-1}) to~(\ref{item:flow-4}) for each $i\in[L-1]$.
\end{lemma}
\begin{proof}
We induct on $i\in[L-1]$, where the base case is just property~(\ref{item:flow-4}) for $i=0$ (and $\mathbf b_0=\mathbf b$). For this base case, since the singleton sets $\{v\}$ are in $\mathcal P_1$, they are also in $\mathcal C$, so $|\mathbf b(\{v\})|\le\deg(v)$ for all $v\in V$, which implies $|\mathbf b_0(v)|=|\mathbf b(v)|\le\deg(v)=\deg_{\partial\mathcal R_{\ge 1}}(v)$, as desired.

For the inductive step, we apply \Cref{lem:construct-flow} on iteration $i\ge1$ and the vector $\mathbf s=\mathbf b_{i-1}$.
We first verify the conditions on $\mathbf s$ required by \Cref{lem:construct-flow}. 
\begin{enumerate}[(a)]
 \item Condition~\ref{item:construct-flow-1} follows from property~(\ref{item:flow-4}) for iteration $i-1$, which is assumed inductively.
 \item To establish condition~\ref{item:construct-flow-2}, fix a set $C\in\mathcal R_{\ge i+1}$. We first prove that $\mathbf b_0(C)=\mathbf b_{i-1}(C)$. This is trivial for $i=1$, so assume that $i>1$. For a given $j\in[i-1]$, the set $C$ is a disjoint union of some sets $C_1,\ldots,C_\ell\in\mathcal R_{\ge j+1}$ by \Cref{clm:refinement}. Apply property~(\ref{item:flow-3}) for iteration $j$ to obtain $(\mathbf b_{j-1}-\mathbf b_j)(C_k)=0$ for all $k\in[\ell]$. Summing over all $k\in[\ell]$ gives $(\mathbf b_{j-1}-\mathbf b_j)(C)=\sum_{k\in[\ell]}(\mathbf b_{j-1}-\mathbf b_j)(C_k)=0$, so $\mathbf b_{j-1}(C)=\mathbf b_j(C)$. Over all iterations $j\in[i-1]$, we obtain $\mathbf b_0(C)=\mathbf b_1(C)=\cdots=\mathbf b_{i-1}(C)$.

Since $\mathbf s=\mathbf b_{i-1}$, we have $|\mathbf s(C)|=|\mathbf b_{i-1}(C)|=|\mathbf b_0(C)|=|\mathbf b(C)|$.
Finally, since the initial flow demand $\mathbf b\in\mathbb R^V$ satisfies $|\mathbf b(C)|\le\delta C$, we have $|\mathbf s(C)|=|\mathbf b(C)|\le\delta C$, establishing condition~\ref{item:construct-flow-2}.
\end{enumerate}
With the conditions fulfilled, \Cref{lem:construct-flow} outputs a flow $f$ which we set as $f_i$, immediately satisfying property~(\ref{item:flow-2}). We set $\mathbf b_i=\mathbf t$ so that flow $f_i$ routes demand $\mathbf s-\mathbf t=\mathbf b_{i-1}-\mathbf b_i$ and property~(\ref{item:flow-1}) holds. Properties~(\ref{item:flow-3}) and (\ref{item:flow-4}) follow from properties~\ref{item:technical-3} and \ref{item:technical-1} of \Cref{lem:construct-flow}, respectively, completing the induction and hence the proof.
\end{proof}

For the rest of this section, we establish \Cref{lem:construct-flow}, the most technical part of the proof. We first establish a helper claim about constructing certain demands and flows. The proof is quite technical and is split into several subclaims (and their proofs).

\begin{claim}\label{clm:technical-helper-main}
For any $i\in[L-1]$, consider any vector $\mathbf x\in\mathbb R^V$ with $|\mathbf x|\le\deg_{\partial\mathcal R_{\ge i}}$. There exists a vector $\mathbf y\in\mathbb R^V$ such that
\begin{enumerate}
\item $|\mathbf y|\le2\deg_{\partial\mathcal P_i}+2L\beta\deg_{\partial\mathcal P_{i+1}}$.
\item For all clusters $C\in\mathcal P_{i+1}$, we have $(\mathbf x-\mathbf y)(C)=0$.
\item There is a flow routing demand $\mathbf x-\mathbf y$ with congestion $4L\beta$.
\end{enumerate}
\end{claim}

\begin{proof}
We begin with a few subclaims.
\begin{subclaim}\label{clm:assumption-3-new}
For any vector $\mathbf s\in\mathbb R^V_{\ge0}$ with $\mathbf s\le\deg_{\partial\mathcal P_{i+1}}$, there is a vector $\mathbf t\in\mathbb R^V_{\ge0}$ with $\mathbf t\le\mathrm{deg}_{\partial\mathcal P_i}/2$ and a flow routing demand $\mathbf s-\mathbf t$ with congestion $\beta$.
\end{subclaim}
\begin{subproof}
By assumption~(\ref{item:cut-approx2-3}) of \Cref{lem:cut-approx}, there is a flow in $G$ with congestion $\beta$ such that each vertex $v\in V$ sends $\deg_{\partial\mathcal P_{i+1}}(v)$ flow and receives at most $\frac12\deg_{\partial\mathcal P_i}(v)$ flow. In other words, there is a vector $\mathbf t\in\mathbb R^V_{\ge0}$ with $\mathbf t\le\mathrm{deg}_{\partial\mathcal P_i}/2$ and a flow routing demand $\textup{deg}_{\partial\mathcal P_{i+1}}-\mathbf t$ with congestion $\beta$. Take a path decomposition of the flow where each vertex is the start of $\deg_{\partial\mathcal P_{i+1}}(v)$ total capacity of (potentially empty) paths and the end of $\mathbf t(v)$ total capacity of (potentially empty) paths. Since $\mathbf s\le\deg_{\partial\mathcal P_{i+1}}$, we can remove or decrease the capacity of paths until each vertex is the start of $\mathbf s(v)$ total paths. The resulting flow routes demand $\mathbf s-\mathbf t$ with congestion $\beta$.
\end{subproof}
\begin{subclaim}\label{clm:technical-helper-1}
For any $i\in[L-1]$, consider any vector $\mathbf x\in\mathbb R^V_{\ge0}$ with $\mathbf x\le\deg_{\partial\mathcal R_{\ge i}}$. There exists a vector $\mathbf y\in\mathbb R^V_{\ge0}$ with $\mathbf y\le2\deg_{\partial\mathcal P_i}$ and a flow routing demand $\mathbf x-\mathbf y$ with congestion $(2L-2i)\beta$.
\end{subclaim}
\begin{subproof}
We prove the statement by induction from $i=L$ down to $i=1$. For the base case $i=L$, since $\mathcal R_{\ge L}=\mathcal P_L$, we can simply set $\mathbf y=\mathbf x$ and take the empty flow.

For the inductive step, consider a vector $\mathbf x\in\mathbb R^V_{\ge0}$ with $\mathbf x\le\deg_{\partial\mathcal R_{\ge i}}$. Define $\mathbf x'\in\mathbb R^V_{\ge0}$ as
\[ \mathbf x'(v)=\begin{cases}
\displaystyle\frac{\deg_{\partial\mathcal R_{\ge i+1}}(v)}{\deg_{\partial\mathcal R_{\ge i}}(v)}\cdot\mathbf x(v) &\text{if }\deg_{\partial\mathcal R_{\ge i}}(v)>0,
\\0 &\text{otherwise},
\end{cases} \]
which satisfies $\mathbf x'\le\deg_{\partial\mathcal R_{\ge i+1}}$. By induction, there exists a vector $\mathbf y'\in\mathbb R^V_{\ge0}$ with $\mathbf y'\le2\deg_{\partial\mathcal P_{i+1}}$ and a flow $f_1$ routing demand $\mathbf x'-\mathbf y'$ with congestion $(2L-2(i+1))\beta$. By \Cref{clm:assumption-3-new} on vector $\mathbf s=\mathbf y'/2$, there exists a vector $\mathbf t\in\mathbb R^V_{\ge0}$ with $\mathbf t\le\mathrm{deg}_{\partial\mathcal P_i}/2$ and a flow routing demand $\mathbf y'/2-\mathbf t$ with congestion $\beta$. Scaling this flow by factor $2$, we obtain a flow $f_2$ routing demand $\mathbf y'-2\mathbf t$ with congestion $2\beta$.

The final flow is the sum of flows $f_1$ and $f_2$, which routes demand $(\mathbf x'-\mathbf y')+(\mathbf y'-2\mathbf t)$ and has congestion $(2L-2(i+1))\beta+2\beta=(2L-2i)\beta$. We set $\mathbf y=\mathbf x-\mathbf x'+2\mathbf t$ so that the demand routed is exactly $\mathbf x-\mathbf y$. Note that $\mathbf y\ge0$ since $\mathbf x'\le\mathbf x$ holds by \Cref{clm:refinement}.

To complete the induction, it remains to establish $\mathbf y\le2\deg_{\partial\mathcal P_i}$. Since $\mathbf y=\mathbf x-\mathbf x'+2\mathbf t$ and $\mathbf t\le\mathrm{deg}_{\partial\mathcal P_i}/2$, it suffices to show that $\mathbf x-\mathbf x'\le\deg_{\partial\mathcal P_i}$. If $\deg_{\partial\mathcal R_{\ge i}}(v)=0$, then $\mathbf x(v)=0$ and $\mathbf x'(v)=0$, so $(\mathbf x-\mathbf x')(v)=0\le\deg_{\partial\mathcal P_i}(v)$. Otherwise, if $\deg_{\partial\mathcal R_{\ge i}}(v)>0$, then
\begin{align*}
(\mathbf x-\mathbf x')(v)&=\bigg(\frac{\deg_{\partial\mathcal R_{\ge i}}(v)-\deg_{\partial\mathcal R_{\ge i+1}}(v)}{\deg_{\partial\mathcal R_{\ge i}}(v)}\bigg)\mathbf x(v)\\&\le\frac{\deg_{\partial\mathcal P_i}(v)}{\deg_{\partial\mathcal R_{\ge i}}(v)}\mathbf x(v)
\\&\le\deg_{\partial\mathcal P_i}(v),
\end{align*}
where the first inequality holds by \Cref{clm:difference}.
This completes the induction and hence the proof.
\end{subproof}
The next subclaim almost proves the desired statement, except that $\mathbf x$ and $\mathbf y$ are restricted to be nonnegative.
\begin{subclaim}\label{clm:technical-helper-2}
For any $i\in[L-1]$, consider any vector $\mathbf x\in\mathbb R^V_{\ge0}$ with $\mathbf x\le\deg_{\partial\mathcal R_{\ge i}}$. There exists a vector $\mathbf y\in\mathbb R^V_{\ge0}$ such that
\begin{enumerate}
\item $\mathbf y\le2\deg_{\partial\mathcal P_i}+2L\beta\deg_{\partial\mathcal P_{i+1}}$.
\item For all clusters $C\in\mathcal P_{i+1}$, we have $(\mathbf x-\mathbf y)(C)=0$.
\item There is a flow routing demand $\mathbf x-\mathbf y$ with congestion $2L\beta$.
\end{enumerate}
\end{subclaim}
\begin{subproof}
Apply \Cref{clm:technical-helper-1} on vector $\mathbf x$, obtaining a vector $\mathbf y\in\mathbb R^V_{\ge0}$ which we relabel as $\mathbf y'$, and a flow $f$ routing demand $\mathbf x-\mathbf y'$ with congestion $(2L-2i)\beta\le2L\beta$. Take a path decomposition of flow $f$ where each vertex is the start of $\mathbf x(v)$ total capacity of (potentially empty) paths and the end of $\mathbf y'(v)$ total capacity of (potentially empty) paths. For each path starting at a vertex $v$ in some cluster $C\in\mathcal P_{i+1}$, perform the following operation. If the path contains an edge $(u,v)$ in $\partial C$ with $u\in C$, then replace the path with its prefix ending at $u$; otherwise, do nothing to the path. Note that the modified path ends in the same cluster $C\in\mathcal P_{i+1}$ as its starting point. These modified paths form a new flow $f'$, which also has congestion $2L\beta$.

We now bound the difference in the demands routed by $f$ and $f'$. To do so, we consider the difference in endpoints in the old and new path decompositions. Each vertex $u\in V$ was initially the endpoint of $\mathbf y'(u)$ total capacity of paths. We now claim that for each cluster $C\in\mathcal P_{i+1}$, each vertex $u\in C$ becomes the new endpoint of at most $2L\beta\deg_{\partial C}(u)=2L\beta\deg_{\partial\mathcal P_{i+1}}(u)$ total capacity of paths. This is because each new endpoint is a result of an edge $(u,v)\in \partial C$ in some path, and the total capacity of such paths is at most $2L\beta\deg_{\partial C}(u)$ since the flow $f$ has congestion $2L\beta$. It follows that each vertex $u\in V$ is the (new or old) endpoint of at most $\mathbf y'(u)+2L\beta\deg_{\partial\mathcal P_{i+1}}(u)$ total capacity of paths in the new flow $f'$.

Define vector $\mathbf y\in\mathbb R^V_{\ge0}$ such that each vertex $u\in V$ is the endpoint of $\mathbf y(u)$ total capacity of paths in the new flow $f'$. In other words, the flow $f'$ routes demand $\mathbf x-\mathbf y$. We have shown that $\mathbf y\le\mathbf y'+2L\beta\deg_{\partial\mathcal P_{i+1}}$, and combined with the guarantee $\mathbf y'\le2\deg_{\partial\mathcal P_i}$ from \Cref{clm:technical-helper-1}, we conclude that $\mathbf y\le2\deg_{\partial\mathcal P_i}+2L\beta\deg_{\partial\mathcal P_{i+1}}$. Finally, recall that each path of $f'$ starts and ends in the same cluster of $\mathcal P_{i+1}$, so $(\mathbf x-\mathbf y)(C)=0$ for all clusters $C\in\mathcal P_{i+1}$.
\end{subproof}
Finally, we prove \Cref{clm:technical-helper-main} using \Cref{clm:technical-helper-2}. Given vector $\mathbf x\in\mathbb R^V$ with $|\mathbf x|\le\deg_{\partial\mathcal R_{\ge i}}$, let $\mathbf x^+,\mathbf x^-\in\mathbb R^V_{\ge0}$ be the positive and negative parts of $\mathbf x$, so that $\mathbf x=\mathbf x^+-\mathbf x^-$. Apply \Cref{clm:technical-helper-2} to $\mathbf x^+$ and $\mathbf x^-$ separately to obtain $\mathbf y^+$ and $\mathbf y^-$, respectively, and set $\mathbf y=\mathbf y^+-\mathbf y^-$. The three properties are immediately satisfied; note that the congestion is now $4L\beta$ since we take the difference of the two flows routing demands $\mathbf x^+-\mathbf y^+$ and $\mathbf x^--\mathbf y^-$. This concludes the proof of \Cref{clm:technical-helper-main}.
\end{proof}

We now prove \Cref{lem:construct-flow}, restated below.
\ConstructFlow*
\begin{proof}
We first construct vector $\mathbf t\in\mathbb R^V$ as follows. For each set $C\in\mathcal R_{\ge i+1}$, define\linebreak $\mathbf t(v)=\mathbf s(C)\cdot\deg_{\partial\mathcal R_{\ge i+1}}(v)/\delta C$ for all $v\in C$, which satisfies $|\mathbf t(v)|\le\deg_{\partial\mathcal R_{\ge i+1}}(v)$ by condition~\ref{item:construct-flow-2}. Since $\mathcal R_{\ge i+1}$ is a partition of $V$, this fully defines demand $\mathbf t$, which satisfies the bound required by property~\ref{item:technical-1}. Also, since $C\in\mathcal R_{\ge i+1}$, we have $\sum_{v\in C}\deg_{\partial\mathcal R_{\ge i+1}}(v)=\delta C$, so
\[ \mathbf t(C)=\sum_{v\in C}\mathbf t(v)=\sum_{v\in C}\mathbf s(C)\cdot\frac{\deg_{\partial\mathcal R_{\ge i+1}}(v)}{\delta C}=\mathbf s(C) ,\]
satisfying property~\ref{item:technical-3}. In particular, $\mathbf t(V)=\mathbf s(V)$ and $\mathbf s-\mathbf t$ is a valid demand.

Since $|\mathbf s-\mathbf t|\le|\mathbf s|+|\mathbf t|\le\deg_{\partial\mathcal R_{\ge i}}+\deg_{\partial\mathcal R_{\ge i+1}}\le2\deg_{\partial\mathcal R_{\ge i}}$, we apply \Cref{clm:technical-helper-main} with $\mathbf x=\frac12(\mathbf s-\mathbf t)$, obtaining a vector $\mathbf y\in\mathbb R^V_{\ge0}$ with $|\mathbf y|\le2\deg_{\partial\mathcal P_i}+2L\beta\deg_{\partial\mathcal P_{i+1}}$ and $(\frac12(\mathbf s-\mathbf t)-\mathbf y)(C)=0$ for all clusters $C\in\mathcal P_{i+1}$. Let $f_1$ be the flow scaled by factor $2$, which routes demand $\mathbf s-\mathbf t-2\mathbf y$ with congestion $8L\beta$.

Consider a cluster $C\in\mathcal P_{i+1}$. Since $(\frac12(\mathbf s-\mathbf t)-\mathbf y)(C)=0$ and $(\mathbf s-\mathbf t)(C)=0$, we have $\mathbf y(C)=0$ as well. Moreover, for all vertices $v\in C$, we have
\begin{align*}
|\mathbf y(v)|&\le2\deg_{\partial\mathcal P_i}(v)+2L\beta\deg_{\partial C}(v)
\le4 L\beta\deg_{\partial\mathcal P_i\cup\partial C}(v).
\end{align*}
We conclude that the scaled-down and restricted vector $\frac1{4L\beta}\mathbf y|_C$ is a demand satisfying $\big|\frac1{4L\beta}\mathbf y|_C\big|\le\mathrm{deg}_{\partial\mathcal P_i\cup\partial C}|_C$. By assumption~(\ref{item:cut-approx2-2}) of \Cref{lem:cut-approx}, the collection of vertex weightings $\{\mathrm{deg}_{\partial\mathcal P_i\cup\partial C}|_C\in\mathbb R^V_{\ge0}:C\in\mathcal P_{i+1}\}$ mixes simultaneously in $G$ with congestion $\alpha$, so there is a flow in $G$ routing demand $\sum_{C\in\mathcal P_{i+1}}\frac1{4L\beta}\mathbf y|_C=\frac1{4L\beta}\mathbf y$ with congestion $\alpha$. Scaling this flow by factor $8L\beta$, we obtain a flow $f_2$ routing demand $2\mathbf y$ with congestion $8L\alpha\beta$.

The final flow $f$ is $f_1+f_2$, which routes demand $(\mathbf s-\mathbf t-2\mathbf y)+2\mathbf y=\mathbf s-\mathbf t$ and has congestion $8L\beta+8L\alpha\beta\le16L\alpha\beta$, concluding the proof of \Cref{lem:construct-flow}.
\end{proof}

\section{Partitioning Algorithm}\label{sec:construct-next-partition}

The partitioning algorithm starts with the partition $\mathcal P_1=\{\{v\}:v\in V\}$ of singleton clusters. The algorithm then iteratively constructs partition $\mathcal P_{i+1}$ given the current partitions $\mathcal P_1,\mathcal P_2,\ldots,\mathcal P_i$. The lemma below establishes this iterative algorithm, where we substitute $L$ for $i$.

\begin{theorem}\label{thm:cut-approx-next-level}
Consider a capacitated graph $G=(V,E)$, and let $\alpha\ge1$ be a parameter. Suppose there exist partitions $\mathcal P_1,\mathcal P_2,\ldots,\mathcal P_L$ that satisfy the three properties in \Cref{lem:cut-approx}, i.e.,
 \begin{enumerate}
 \item $\mathcal P_1$ is the partition $\{\{v\}:v\in V\}$ of singleton clusters.
 \item For each $i\in[L-1]$, the collection of vertex weightings $\{\mathrm{deg}_{\partial\mathcal P_i\cup\partial C}|_C\in\mathbb R^V_{\ge0}:C\in\mathcal P_{i+1}\}$ mixes simultaneously in $G$ with congestion $\alpha=O(\log^5(nW))$.\label{item:algo-2}
 \item For each $i\in[L-1]$, there is a flow in $G$ with congestion $\beta=O(\log^3(nW))$ such that each vertex $v\in V$ sends $\deg_{\partial\mathcal P_{i+1}}(v)$ flow and receives at most $\frac12\deg_{\partial\mathcal P_i}(v)$ flow. \label{item:algo-3}
 \end{enumerate}
Then, there is an algorithm that constructs partition $\mathcal P_{L+1}$ such that properties~(\ref{item:algo-2}) and~(\ref{item:algo-3}) hold for $i=L$ as well. The algorithm runs in $\tilde O(m)$ time.
\end{theorem}

We remark that the parameters $\alpha=O(\log^5(nW))$ and $\beta=O(\log^3(nW))$, which result in an overall quality of $O(\log^{10}(nW))$, are far from optimized. We aim for a clean and modular exposition over the optimization of logarithmic factors, leaving the latter open for future work.

Before we describe the algorithm, we first show that $O(\log(nW))$ iterations suffice to obtain a congestion-approximator.

\begin{claim}\label{clm:L}
After $L=O(\log(nW))$ iterations, we have $\mathcal P_L=\{V\}$, and $\mathcal C$ is a congestion-approximator with quality $O(\log^{10}(nW))$.
\end{claim}
\begin{proof}
We first claim that $\delta\mathcal P_{i+1}\le\delta\mathcal P_i/2$ for all $i\in[L]$. By property~(\ref{item:algo-3}), there is a flow that sends a total of $\deg_{\partial\mathcal P_{i+1}}(V)$ flow among the vertices, and receives a total of at most $\frac12\deg_{\partial\mathcal P_i}(V)$ flow among the vertices. Since the total flow sent equals the total flow received, we have $\deg_{\partial\mathcal P_{i+1}}(V)\le\frac12\deg_{\partial\mathcal P_i}(V)$, or equivalently, $\delta\mathcal P_{i+1}\le\frac12\delta\mathcal P_i$.

The guarantee $\delta\mathcal P_{i+1}\le\frac12\delta\mathcal P_i$ ensures that for $L=O(\log(nW))$, we must have $\delta\mathcal P_L<1$. Since all edge capacities are assumed to be at least $1$, we conclude that $\delta\mathcal P_L=0$ and $\mathcal P_L=\{V\}$, fulfilling property~(\ref{item:cut-approx-1}) of \Cref{thm:cut-approx}. By \Cref{thm:cut-approx}, we obtain a congestion-approximator with quality $16L^2\alpha\beta=O(\log^{10}(nW))$.
\end{proof}

The rest of this section is dedicated to proving \Cref{thm:cut-approx-next-level}.
Throughout the section, we fix the input graph $G=(V,E)$ as well as the current partitions $\mathcal P_1,\mathcal P_2,\ldots,\mathcal P_L$. We also define $\mathcal R_{\ge 1},\ldots,\mathcal R_{\ge L}$ and $\mathcal C=\bigcup_{i\in[L]}\mathcal R_{\ge i}$ according to \Cref{lem:cut-approx}.

At a high level, our algorithm proceeds similarly to expander decomposition algorithms like \cite{SW19}, where expanders are defined with respect to the vertex weighting $\deg_{\partial\mathcal P_L}$ of partition $\mathcal P_L$. We iteratively decompose the graph using the cut-matching game of \cite{KRV09}: for each cluster of the decomposition, we either compute a low-conductance cut or certify that the current cluster is mixing (or in \cite{SW19} terms, a nearly expander). The matching step of the cut-matching game requires a call to approximate min-cut/max-flow, but recall that the partitions $\mathcal P_1,\ldots,\mathcal P_L$ only form a pseudo-congestion-approximator. Luckily, we can modify the graph in a way that the pseudo-congestion-approximator can be adapted to an actual congestion-approximator for the new graph. We then black-box a cut/flow algorithm on this new graph, and then modify the cut and flow to fit the old graph in a way that suffices for the cut-matching game.

In more detail, we break down this section as follows:
\begin{enumerate}
\item In \Cref{sec:cut-approximator}, we show that $\mathcal C$ can be used to construct a congestion-approximator for slightly modified graphs.
\item In \Cref{sec:cut-flow-algo}, we cite the (approximate) \emph{fair cut/flow} algorithm of~\cite{LNPS23}, which computes a cut/flow pair with desirable properties given a congestion-approximator.
\item In \Cref{sec:CMG-trimming}, we introduce the cut-matching game as well as a trimming procedure similar to~\cite{SW19}. Both the cut-matching game and the trimming step use the fair cut/flow algorithm (\Cref{sec:cut-flow-algo}) on the modified graphs for which we have a congestion-approximator (\Cref{sec:cut-approximator}).
\item Finally, in \Cref{sec:clustering-algo}, we establish \Cref{thm:cut-approx-next-level}. We present the recursive clustering algorithm that computes the next partition $\mathcal P_{L+1}$ given the current partitions $\mathcal P_1,\mathcal P_2,\ldots,\mathcal P_L$. It uses the cut-matching game and trimming procedures of \Cref{sec:CMG-trimming}.
\end{enumerate}

\subsection{Congestion-approximator}\label{sec:cut-approximator}

We first show how to build a congestion-approximator on certain graph instances that show up in our algorithm. We define these graph instances below.
\begin{definition}[{$G[A,\gamma,\mathbf s,\mathbf t]$}]\label{def:H}
For given vertex set $A\subseteq V$, parameter $\gamma\in(0,1]$, and vertex weightings $\mathbf s,\mathbf t\in\mathbb R^A_{\ge0}$, define the graph $G[A,\gamma,\mathbf s,\mathbf t]$ as follows:
 \begin{itemize}
 \item Start with the graph $G[A]$.
 \item Add new vertices $x$, $s$, and $t$.
 \item For each vertex $v\in A$, add an edge $(x,v)$ with capacity $\gamma\deg_{\partial_G\mathcal P_L\cup\partial_GA}(v)$.
 \item For each vertex $v\in A$, add an edge $(s,v)$ with capacity $\mathbf s(v)$.
 \item For each vertex $v\in A$, add an edge $(t,v)$ with capacity $\mathbf t(v)$.
 \end{itemize}
\end{definition}

To understand these instances, recall a fact about $\mathcal C$ that follows from \Cref{lem:cut-approx}.

\begin{fact}\label{fact:C}
For any demand $\mathbf b\in\mathbb R^V$ satisfying $|\mathbf b(C)|\le\delta C$ for all $C\in\mathcal C$, there exists a demand $\mathbf b'\in\mathbb R^V$ satisfying $|\mathbf b'|\le\deg_{\partial\mathcal P_L}$ and a flow in $G$ routing demand $\mathbf b-\mathbf b'$ with congestion $\kappa$, where we define $\kappa=16L^2\alpha\beta$.
\end{fact}

Suppose we start with the entire graph $G$, and then add a vertex $x$ connected to each vertex $v\in V$ with an edge of capacity $\deg_{\partial\mathcal P_L}(v)$. In the setting of \Cref{fact:C}, suppose we wish to route the demand $\mathbf b$. We start with the flow routing demand $\mathbf b-\mathbf b'$ as promised by \Cref{fact:C}. To route the remaining demand $\mathbf b'$, we simply use the new edges incident to $x$: for each vertex $v\in V$, send $|\mathbf b'(v)|\le\deg_{\partial\mathcal P_L}$ flow along the edge $(x,v)$ in the proper direction, which is a flow with congestion $1$. Overall, we obtain a flow routing demand $\mathbf b$ with congestion $\kappa+1$, and with some more work, we can show that $\mathcal C$ is a congestion approximator of the new graph.

For our algorithm, we actually work with graphs as described in \Cref{def:H}. In particular, the base graphs are induced subgraphs, and there are additional vertices $s$ and $t$. One issue is that the newly added edges may also contribute to the values of $\delta C$ for $C\in\mathcal C$.\footnote{This issue is also present in the example with the entire graph, but can be resolved by investigating the structure of $\mathcal C$.} Nevertheless, we show in the lemma below that as long as $A,\mathbf s,\mathbf t$ are ``well-behaved'', we can modify $\mathcal C$ into a congestion approximator for $G[A,\gamma,\mathbf s,\mathbf t]$.

\begin{lemma}\label{lem:cut-approx-A}
Consider partitions $\mathcal P_1,\mathcal P_2,\ldots,\mathcal P_L$ for a graph $G=(V,E)$ that satisfy the three properties in \Cref{lem:cut-approx}, and define the partitions $\mathcal R_{\ge 1},\ldots,\mathcal R_{\ge L}$ and $\mathcal C=\bigcup_{i\in[L]}\mathcal R_{\ge i}$ according to \Cref{lem:cut-approx}. Fix vertex set $A\subseteq V$, parameter $\gamma\in(0,1]$, and vertex weightings $\mathbf s,\mathbf t\in\mathbb R^A_{\ge0}$ on $A$, and denote the graph $G[A,\gamma,\mathbf s,\mathbf t]$ by $H=(V_H,E_H)$. Consider a parameter $\beta\ge1$ such that the following assumption holds:
\begin{enumerate}[\ensuremath{(\star)}]
\item $\mathbf s(C\cap A)+\mathbf t(C\cap A)+\gamma\cdot \delta_G(C\cap A)\le\beta\cdot \delta_GC$ for all $C\in\mathcal C$. \label{item:property-Ast1}
\end{enumerate}
Let $\mathcal C|_A$ be the collection $\{C\cap A:C\in\mathcal C\}$.
Then, $\mathcal C|_A\cup\{\{x\},\{s\},\{t\}\}$ is a congestion-approximator of $H$ with quality $O(\beta\gamma^{-1}\kappa)$.
\end{lemma}

For the rest of \Cref{sec:cut-approximator}, we prove \Cref{lem:cut-approx-A}. Consider a demand $\mathbf b\in\mathbb R^{A\cup\{x,s,t\}}$ satisfying $|\mathbf b(C\cap A)|\le\delta_H(C\cap A)$ for all $C\in\mathcal C$ as well as $|\mathbf b(v)|\le\delta_H\{v\}$ for $v\in\{x,s,t\}$. We want to establish a flow on $H$ routing demand $\mathbf b$ with congestion $O(\beta\gamma^{-1}\kappa)$.

We first handle the demand at $x$, $s$, and $t$. For each edge $(x,v)$ where $v\in A$, route $\mathbf b(x)/\deg_H(x)\cdot c_H(x,v)$ flow from $x$ to $v$ (or $-\mathbf b(x)/\deg_H(x)\cdot c_H(x,v)$ flow from $v$ to $x$, whichever is nonnegative). Since $|\mathbf b(x)|=|\mathbf b(\{x\})|\le\delta_H\{x\}=\deg_H(x)$, we route at most $c_H(x,v)$ flow along each edge $(x,v)$. Analogously, for each edge $(s,v)$ where $v\in A$, route $\mathbf b(s)/\deg_H(s)\cdot c_H(s,v)$ flow from $s$ to $v$, and for each edge $(t,v)$ where $v\in A$, route $\mathbf b(t)/\deg_H(t)\cdot c_H(t,v)$ flow from $t$ to $v$. By the same argument, we route at most $c_H(s,v)$ and $c_H(t,v)$ flow along each edge $(s,v)$ and $(t,v)$, respectively. In other words, the routing so far has congestion $1$.

After this initial routing, vertices $x$, $s$, and $t$ no longer have any demand, and each vertex $v\in S$ receives at most $c_H(x,v)+c_H(s,v)+c_H(t,v)=c(\{x,s,t\},v)$ additional demand in absolute value. In other words, if $\widetilde{\mathbf b}$ is the new demand that must be routed, we have $\widetilde{\mathbf b}(x)=\widetilde{\mathbf b}(s)=\widetilde{\mathbf b}(t)=0$ and $|\mathbf b(v)-\tilde{\mathbf  b}(v)|\le c_H(\{x,s,t\},v)$.

In order to invoke \Cref{fact:C}, our next goal is to show the following.
\begin{claim}\label{clm:claim-b}
$|\widetilde{\mathbf b}(C\cap A)|\le (1+2\beta+2\gamma)\delta_GC$ for all $C\in\mathcal C$.
\end{claim}
\begin{proof}
We first bound $|\mathbf b(C\cap A)-\widetilde{\mathbf b}(C\cap A)|$ as
\begin{align*}
|\mathbf b(C\cap A)-\widetilde{\mathbf b}(C\cap A)|\le\sum_{v\in C\cap A}|\mathbf b(v)-\widetilde{\mathbf b}(v)|\le\sum_{v\in C\cap A}c_H(\{x,s,t\},v)=c_H(\{x,s,t\},C\cap A).
\end{align*}
Next, we bound $|\mathbf b(C\cap A)|$ as follows. By assumption, we have $|\mathbf b(C\cap A)|\le\delta_H(C\cap A)$. By construction of $H$, we have
\begin{align*}
\delta_H(C\cap A)=c_G(C\cap A,A\setminus C)+c_H(\{x,s,t\},C\cap A)\le\delta_GC+c_H(\{x,s,t\},C\cap A).
\end{align*}
Putting everything together, we obtain
\begin{gather}
|\widetilde{\mathbf b}(C\cap A)|\le|\mathbf b(C\cap A)-\widetilde{\mathbf b}(C\cap A)|+|\mathbf b(C\cap A)|\le\delta_GC+2c_H(\{x,s,t\},C\cap A).\label{eq:tilde-b}
\end{gather}
It remains to bound $c_H(\{x,s,t\},C\cap A)$, which we split into $c_H(\{s,t\},C\cap A)+c_H(\{x\},C\cap A)$. By construction of $H=G[A,\gamma,\mathbf s,\mathbf t]$, we have
\[ c_H(\{s,t\},C\cap A)=\mathbf s(C\cap A)+\mathbf t(C\cap A) \]
and
\[ c_H(\{x\},C\cap A)=\gamma\cdot\deg_{\partial_G\mathcal P_L\cup\partial_G A}(C\cap A)\le\gamma\cdot\deg_{\partial_G\mathcal P_L}(C\cap A)+\gamma\cdot\deg_{\partial_G A}(C\cap A) .\]
We now bound the individual terms $\deg_{\partial_G\mathcal P_L}(C\cap A)$ and $\deg_{\partial_G A}(C\cap A)$ above. For $\deg_{\partial_G A}(C\cap A)$, we claim the bound $\deg_{\partial_G A}(C\cap A)\le\delta_G(C\cap A)$: any edge in $\partial_G A$ with an endpoint in $C\cap A$ has its other endpoint outside $C\cap A$, so the edge must be in $\partial_G(C\cap A)$, and the claimed bound holds. For $\deg_{\partial_G\mathcal P_L}(C\cap A)$, we claim the bound $\deg_{\partial_G\mathcal P_L}(C\cap A)\le\deg_{\partial_G\mathcal P_L}(C)\le\delta_GC$. The first inequality is trivial, and for the second inequality, observe that by construction of $\mathcal C$, each set $C\in\mathcal C$ is a subset of some cluster in the partition $\mathcal P_L$. It follows that any edge in $\partial_G\mathcal P_L$ with an endpoint in $C$ has its other endpoint outside $C$, so the edge must be in $\partial_GC$, and we conclude that $\deg_{\partial_G\mathcal P_L}(C)\le\delta_GC$.

Continuing from (\ref{eq:tilde-b}), we conclude that
\begin{align*}
|\widetilde{\mathbf b}(C\cap A)|&\le\delta_GC+2c_H(\{x,s,t\},C\cap A)
\\&=\delta_GC+2c_H(\{s,t\},C\cap A)+2c_H(\{x\},C\cap A)
\\&\le\delta_GC+2(\mathbf s(C\cap A)+\mathbf t(C\cap A))+2(\gamma\cdot\deg_{\partial_G\mathcal P_L}(C\cap A)+\gamma\cdot\deg_{\partial_G A}(C\cap A))
\\&\le\delta_GC+2(\mathbf s(C\cap A)+\mathbf t(C\cap A))+2(\gamma\cdot\delta_GC+\gamma\cdot\delta_G(C\cap A))
\\&\stackrel{\mathclap{\ref{item:property-Ast1}}}\le\delta_GC+2\beta\cdot\delta_GC+2\gamma\cdot\delta_GC,
\end{align*}
finishing the proof.
\end{proof}

 Let $\widetilde{\mathbf b}'\in\mathbb R^V$ be the vector $\widetilde{\mathbf b}\in\mathbb R^{A\cup\{x,s,t\}}$ without entries $\widetilde{\mathbf b}(x),\widetilde{\mathbf b}(s),\widetilde{\mathbf b}(t)$ and with new entries $\widetilde{\mathbf b}'(v)=0$ for all $v\in V\setminus A$. Since $\widetilde{\mathbf b}$ is a demand with $\widetilde{\mathbf b}(x)=\widetilde{\mathbf b}(s)=\widetilde{\mathbf b}(t)=0$, we have that $\widetilde{\mathbf b}'$ is also a demand, i.e., the coordinates sum to $0$. By \Cref{clm:claim-b}, we have
\[ |\widetilde{\mathbf b}'(C)|=|\widetilde{\mathbf b}(C\cap A)|\le(1+2\beta+2\gamma)\delta_GC ,\]
so we can apply \Cref{fact:C} on demand $\widetilde{\mathbf b}'/(1+2\beta+2\gamma)$ to obtain a demand $\mathbf b'\in\mathbb R^V$ satisfying $|\mathbf b'|\le\deg_{\partial_G\mathcal P_L}$ and a flow on $G$ routing demand $\widetilde{\mathbf b}'/(1+2\beta+2\gamma)-\mathbf b'$ with congestion $\kappa$. Scaling this flow by factor $(1+2\beta+2\gamma)$, we obtain a flow $f'$ on $G$ routing demand $\widetilde{\mathbf b}'-(1+2\beta+2\gamma) \mathbf b'$ with congestion $(1+2\beta+2\gamma)\kappa$.

Next, imagine contracting $V\setminus A$ into a single vertex labeled $x$, so that each edge $(v,x)$ has capacity $\deg_{\partial_G A}(v)$. Consider the corresponding flow $f'$ on this contracted graph, which sends at most $(1+2\beta+2\gamma)\kappa\deg_{\partial_G A}(v)$ flow on each edge $(v,x)$. Now consider the exact same flow on $H$, whose edges $(v,x)$ have capacities $\gamma\deg_{\partial_G\mathcal P_L\cup\partial_G A}(v)$ instead of $\deg_{\partial_G A}(v)$. These capacities are at least $\gamma$ times the capacities of the contracted graph, so the corresponding flow has congestion at most a factor $1/\gamma$ larger. We have established a flow on $H$ routing demand $\widetilde{\mathbf b}'-(1+2\beta+2\gamma) \mathbf b'$ with congestion $\gamma^{-1}(1+2\beta+2\gamma)\kappa$.

Finally, since demand $(1+2\beta+2\gamma)\mathbf b'$ satisfies $|(1+2\beta+2\gamma)\mathbf b'|\le(1+2\beta+2\gamma)\deg_{\partial_G\mathcal P_L}$, we can directly route the demand along the edges $(v,x)$ of capacity $\gamma\deg_{\partial_G\mathcal P_L\cup\partial_G A}(v)$, which is a routing with congestion $\gamma^{-1}(1+2\beta+2\gamma)$.

Adding up all three routings, we have routed the initial demand $\mathbf b$ with congestion $1+\gamma^{-1}(1+2\beta+2\gamma)\kappa+\gamma^{-1}(1+2\beta+2\gamma)=O(\beta\gamma^{-1}\kappa)$, concluding the proof.

\subsection{Fair Cut/Flow Algorithm}\label{sec:cut-flow-algo}

Given a congestion-approximator, the most convenient min-cut/max-flow algorithm is the \emph{fair cut/flow} algorithm of~\cite{LNPS23}. We state the definition of a fair cut/flow and then cite the main result of~\cite{LNPS23}.

\begin{definition}[Fair cut/flow]
Let $G=(V,E)$ be an undirected graph with edge capacities $c\in\mathbb{R}_{>0}^{E}$.
Let $s,t$ be two vertices in $V$. For any parameter $\alpha\ge 1$,
we say that an $(s,t)$ cut $S\subseteq V$ and a feasible flow $f$ is an \emph{$\alpha$-fair $(s,t)$-cut/flow pair}
if for each edge $(u,v)\in\partial S$ with $u\in S$ and $v\in T$, the flow $f$ sends at least $\frac{1}{\alpha}\cdot c(u,v)$ flow along the edge in the direction from $u$ to $v$.
\end{definition}

\begin{fact}\label{fact:fair-cut}
For any $\alpha$-fair $(s,t)$-cut/flow pair $(S,f)$, the cut $\partial S$ is an $\alpha$-approximate minimum $(s,t)$-cut.
\end{fact}

\begin{theorem}[Fair cut/flow algorithm~\cite{LNPS23}]\label{thm:fair-cut-LNPS}
Consider a graph $G=(V,E)$, two vertices $s,t\in V$, and error parameter $\epsilon\in(0,1]$. Given a congestion-approximator $\mathcal C$ with quality $\kappa$, there is an algorithm that outputs a $(1+\epsilon)$-fair $(s,t)$-cut/flow pair in $\tilde{O}((\kappa/\epsilon)^{O(1)}(K+m))$ time where $K=\sum_{C\in\mathcal C}|C|$.
\end{theorem}


We remark that the fair cut/flow algorithm above is not the fastest available algorithm. However, it is conceptually the easiest for our purposes, and we believe that future work may improve the running time of fair cut/flow algorithms to approach those of standard approximate cut/flow algorithms. Hence, we decide to black-box a fair cut/flow algorithm rather than starting with a standard cut/flow algorithm and massaging it to work in our setting.

To apply \Cref{thm:fair-cut-LNPS} to the graph $H=G[A,\gamma,\mathbf s,\mathbf t]$ with the congestion-approximator of \Cref{lem:cut-approx-A}, we need to bound $K$ for the congestion-approximator $\mathcal C|_A\cup\{\{x\},\{s\},\{t\}\}$. 
Recall that $\mathcal C=\bigcup_{i\in[L]}\mathcal R_{\ge i}$ is the union of $L$ partitions of $V$, so $\mathcal C|_A$ is the union of $L$ partitions of $A$. So $K=L|A|+3$, where the $+3$ comes from the singletons in $\{\{x\},\{s\},\{t\}\}$. 
It follows that \Cref{thm:fair-cut-LNPS} runs in time $\tilde O((\kappa/\epsilon)^{O(1)}(K+|E(H)|))=\tilde O((\kappa/\epsilon)^{O(1)}(L|A|+m'))$ where $m'$ is the number of edges in $G$ incident to vertices in $A$.

We conclude this section with the main subroutine that we use to construct partition $\mathcal P_{L+1}$. Note that the assumption~\ref{item:property-Ast2} below remains unchanged.

\begin{theorem}[Flow/cut subroutine]\label{thm:fair-cut}
Consider partitions $\mathcal P_1,\mathcal P_2,\ldots,\mathcal P_L$ for a graph $G=(V,E)$ that satisfy the three properties in \Cref{lem:cut-approx}, and define the partitions $\mathcal R_{\ge 1},\ldots,\mathcal R_{\ge L}$ and $\mathcal C=\bigcup_{i\in[L]}\mathcal R_{\ge i}$ according to \Cref{lem:cut-approx}. Fix vertex set $A\subseteq V$, parameter $\gamma\in(0,1]$, and vertex weightings $\mathbf s,\mathbf t\in\mathbb R^A_{\ge0}$ on $A$, and denote the graph $G[A,\gamma,\mathbf s,\mathbf t]$ by $H=(V_H,E_H)$. Consider a parameter $\beta\ge1$ such that the following assumption holds:
\begin{enumerate}[\ensuremath{(\star)}]
\item $\mathbf s(C\cap A)+\mathbf t(C\cap A)+\gamma\cdot \delta_G(C\cap A)\le\beta\cdot \delta_GC$ for all $C\in\mathcal C$. \label{item:property-Ast2}
\end{enumerate}
Let $\mathcal C|_A$ be the collection $\{C\cap A:C\in\mathcal C\}$.
Then, given two vertices $s,t\in V$ and error parameter $\epsilon\in(0,1]$, there is an algorithm that outputs a $(1+\epsilon)$-fair $(s,t)$-cut/flow pair in $\tilde{O}((L\alpha\beta/\epsilon)^{O(1)}(|A|+m'))$ time, where $m'$ is the number of edges in $G$ incident to vertices in $A$.
\end{theorem}

\subsection{Cut-Matching Game and Trimming}\label{sec:CMG-trimming}

We follow the cut-matching game treatment in \cite{SW19}: either find a ``balanced'' cut of small capacity, or ensure that a ``large'' part of the graph mixes with low congestion. The following lemma is similar to Theorem~2.2 of~\cite{SW19} with one key difference: there is no built-in flow subroutine, so the algorithm makes black-box calls to the fair cut/flow algorithm of~\Cref{thm:fair-cut}.

In past work~\cite{RST14,AKL24}, the analysis of the cut-matching game for capacitated graphs has only been sketched, referencing the fact that a capacitated graph can be modelled by an uncapacitated graph with parallel edges (at a cost). For completeness, we provide a full proof of this capacitated case in \Cref{sec:CMG}.

\begin{restatable}[Cut-Matching]{theorem}{CMG}\label{thm:CMG}
Consider a graph $G=(V,E)$ with integral edge capacities in the range $[1,W]$. Let $A\subseteq V$ be a vertex subset, let $\phi,\eta>0$ be parameters, and define $\mathbf d\in\mathbb R^A_{\ge0}$ as $\mathbf d=\textup{deg}_{\partial\mathcal P_L\cup\partial A}|_A$. Suppose that the following assumption holds:
 \begin{enumerate}[\ensuremath{(\diamond)}]
 \item There is a flow on $G$ with congestion $\kappa$ such that each vertex $v\in A$ is the source of $c_G(\{v\},V\setminus A)$ flow and each vertex $v\in V$ is the sink of at most $\deg_{\partial\mathcal P_L}(v)$ flow.\label{item:assumption-CMG}
 \end{enumerate}

There exists parameter $T=O(\log^2(nW))$ and a randomized, Monte Carlo algorithm that outputs a (potentially empty) set $R\subseteq V$ such that
 \begin{enumerate}
 \item $\delta_{G[A]}R\le\phi\mathbf d(R)+\frac\phi{6 T}\mathbf d(A)$,\label{item:CMG-property-1}
 \item $\mathbf d(R)\le\frac23\mathbf d(A)$, and\label{item:CMG-property-2}
 \item Either $\mathbf d(R)\ge\mathbf d(A)/(6 T)$, or the vertex weighting $\mathbf d|_{A\setminus R}$ mixes in $G[A]$ with congestion $5 T/\phi$ with high probability.\label{item:CMG-property-3}
 \end{enumerate}
The algorithm makes at most $T$ calls to \Cref{thm:fair-cut} with parameters \[A\gets A,\,\epsilon\gets\frac1{18 T^2},\,\gamma\gets\frac{\epsilon\phi}2,\,\beta\gets\max\{1,(24\phi+\epsilon\gamma)(\kappa+2)\} .\]
Outside these calls, the algorithm takes an additional $O((|A|+m')\log^4(nW))$ time, where $m'$ is the number of edges in $G$ incident to vertices in $A$.
\end{restatable}

Property~(\ref{item:CMG-property-3}) asserts that either an approximately balanced cut is found, or the vertex weighting $\mathbf d|_{A\setminus R}$ mixes with low congestion (with high probability). In our algorithm for \Cref{thm:cut-approx-next-level}, we actually want the weighting $\mathbf d_{A\setminus R}+\deg_{\partial R}$ to mix in the second case. To guarantee this stronger property, we augment the set $R$ into $R\cup B$ through one additional call to the fair cut/flow algorithm of~\Cref{thm:fair-cut}. The algorithm is similar to the flow-based \emph{expander trimming} procedure in~\cite{SW19}. For completeness, we defer the algorithm and proof to \Cref{sec:trimming}. Note that the setting, including assumption~\ref{item:assumption-trimming}, is the same as \Cref{thm:CMG}.
\begin{restatable}[Trimming]{theorem}{Trimming}\label{thm:trimming}
Consider a graph $G=(V,E)$ with integral edge capacities in the range $[1,W]$. Let $A\subseteq V$ be a vertex subset, let $\phi,\kappa>0$ be parameters, and define $\mathbf d\in\mathbb R^A_{\ge0}$ as $\mathbf d=\textup{deg}_{\partial\mathcal P_L\cup\partial A}|_A$. Suppose that the following assumption holds:
 \begin{enumerate}[\ensuremath{(\diamond)}]
 \item There is a flow on $G$ with congestion $\kappa$ such that each vertex $v\in A$ is the source of $c_G(\{v\},V\setminus A)$ flow and each vertex $v\in V$ is the sink of at most $\deg_{\partial\mathcal P_L}(v)$ flow.\label{item:assumption-trimming}
 \end{enumerate}
There is a deterministic algorithm that inputs a subset $R\subseteq A$ and a parameter $\epsilon>0$, and outputs a (potentially empty) set $B\subseteq A$ such that
 \begin{enumerate}
 \item $\delta_{G[A]}B\le2\delta_{G[A]}R+2\epsilon\phi\mathbf d(A)$,\label{item:trimming-1}
 \item $\mathbf d(B\setminus R)\le\frac1{6\phi}\,\delta_{G[A]}R+\frac\epsilon6\mathbf d(A)$,\label{item:trimming-2}
 \item If the vertex weighting $\mathbf d|_{A\setminus R}$ mixes in $G[A]$ with congestion $c$, then the vertex weighting\linebreak $(\mathbf d+\deg_{\partial_{G[A]}(R\cup B)})|_{A\setminus(R\cup B)}$ mixes in $G[A]$ with congestion $2+(1+24\phi)c$, and\label{item:trimming-3}
 \item There exists a vector $\mathbf t\in\mathbb R^A_{\ge0}$ with $\mathbf t\le24\phi\mathbf d|_{A\setminus(R\cup B)}$ and a flow $g$ on $G[A\setminus(R\cup B)]$ routing demand $\textup{deg}_{\partial_{G[A]}(R\cup B)}|_{A\setminus(R\cup B)}-\mathbf t$ with congestion $2$.\label{item:trimming-4}
 \end{enumerate}
The algorithm makes one call to \Cref{thm:fair-cut} with parameters
\[ A\gets A,\,\epsilon\gets\epsilon,\,\gamma\gets\frac{\epsilon\phi}2,\,\beta\gets\max\{1,(12\phi+\epsilon\gamma)(\kappa+2)\} .\]
Outside of this call, the algorithm takes an additional $O(|A|+m')$ time, where $m'$ is the number of edges in $G$ incident to vertices in $A$.
\end{restatable}

\subsection{Clustering Algorithm}\label{sec:clustering-algo}
With the necessary primitives established, we now describe how to construct partition $\mathcal P_{L+1}$ given the partitions $\mathcal P_1,\ldots,\mathcal P_L$. The algorithm is recursive, taking as input a vertex subset $A\subseteq V$ that is initially $V$. Throughout, we maintain the invariant that each input subset satisfies assumption~\ref{item:assumption-CMG}, which is the same for \Cref{thm:CMG,thm:trimming}.

On input $A\subseteq V$, the algorithm calls \Cref{thm:CMG} with parameters $\phi\gets\frac1{C\log^3(nW)}$ and $\kappa\gets C\log^3(nW)$ for a large enough constant $C>0$. The algorithm obtains an output set $R\subseteq A$ and then calls \Cref{thm:trimming} on inputs $R\gets R$ and $\epsilon\gets1/(4T)$ with the same parameters $\phi,\kappa$, obtaining a set $B\subseteq A$. There are now two cases:
 \begin{enumerate}
 \item If $\mathbf d(R)\ge\mathbf d(A)/(6T)$, then recursively call the algorithm on inputs $R\cup B$ and $A\setminus(R\cup B)$ if they are nonempty.\label{item:partition-case-1}
 \item Otherwise, make a single recursive call on input $R\cup B$ if it is nonempty, and add the set $A\setminus(R\cup B)$ to the final partition $\mathcal P_{L+1}$.\label{item:partition-case-2}
 \end{enumerate}
\begin{claim}
Property~(\ref{item:algo-2}) of \Cref{thm:cut-approx-next-level} holds for $i=L$, i.e., the collection of vertex weightings\linebreak $\{\mathrm{deg}_{\partial\mathcal P_L\cup\partial C}|_C\in\mathbb R^V_{\ge0}:C\in\mathcal P_{L+1}\}$ mixes simultaneously in $G$ with congestion $O(\log^5(nW))$.
\end{claim}
\begin{proof}
By property~(\ref{item:CMG-property-3}) of \Cref{thm:CMG} and property~(\ref{item:trimming-3}) of \Cref{thm:trimming}, for each set $A\setminus(R\cup B)$ added to the final partition $\mathcal P_{L+1}$, the vertex weighting $(\mathbf d+\deg_{\partial_{G[A]}(R\cup B)})|_{A\setminus(R\cup B)}$ mixes in $G[A]$ with congestion $2+(1+24\phi)\cdot5T/\phi$. Since $\partial_G(A\setminus(R\cup B))\subseteq\partial_GA\cup\partial_{G[A]}(A\setminus(R\cup B))$, we have
\begin{align*}
\mathbf d+\deg_{\partial_{G[A]}(R\cup B)}&=\deg_{\partial_G\mathcal P_L\cup\partial_GA}+\deg_{\partial_{G[A]}(A\setminus(R\cup B))}
\\&\ge\deg_{\partial_G\mathcal P_L\cup\partial_GA\cup\partial_{G[A]}(A\setminus(R\cup B))}
\\&\ge\deg_{\partial_G\mathcal P_L\cup\partial_G(A\setminus(R\cup B))} ,
\end{align*}
so in particular, the vertex weighting $\textup{deg}_{\partial_G\mathcal P_L\cup\partial_G(A\setminus(R\cup B))}|_{A\setminus(R\cup B)}\le(\mathbf d+\deg_{\partial_{G[A]}(R\cup B)})|_{A\setminus(R\cup B)}$ also mixes in $G[A]$ with congestion $2+(1+24\phi)\cdot5T/\phi$. The recursive instances $A$ that add a set $A\setminus(R\cup B)$ to $\mathcal P_{L+1}$ are disjoint, so the vertex weightings $\textup{deg}_{\partial_G\mathcal P_L\cup\partial_G(R\cup B)}|_{A\setminus(R\cup B)}$ mix simultaneously in $G$ with the same congestion. We bound the congestion by $2+(1+24\phi)\cdot5T/\phi=O(T/\phi)=O(\log^5(nW))$, concluding the proof.
\end{proof}

It remains to establish condition~(\ref{item:algo-3}) of \Cref{thm:cut-approx-next-level} for $i=L$ as well as the assumption~\ref{item:assumption-CMG} of \Cref{thm:CMG,thm:trimming}. To do so, we first prove a few guarantees of the algorithm.

\begin{claim}\label{clm:recursive-d}
For any recursive call $A'\subseteq A$, we have $\mathbf d'(A')\le(1-\frac1{36T})\mathbf d(A)$ where $\mathbf d'=\textup{deg}_{\partial\mathcal P_L\cup\partial A'}|_{A'}$.
\end{claim}
\begin{proof}
We first claim that $\partial_GA'\subseteq\partial_GA\cup\partial_{G[A]}A'$. For any edge in $\partial_GA'$, consider its endpoint in $V\setminus A'$. Either it is in $A\setminus A'$, in which case the edge belongs to $\partial_{G[A]}A'$, or it is in $V\setminus A$, in which case the edge belongs to $\partial_GA$. It follows that $\partial_GA'\subseteq\partial_GA\cup\partial_{G[A]}A'$, and we can bound $\mathbf d'(A')$ as follows:
\begin{align*}
\mathbf d'(A')&=\deg_{\partial_G\mathcal P_L\cup\partial_GA'}(A')
\\&\le\deg_{\partial_G\mathcal P_L\cup\partial_GA\cup\partial_{G[A]}A'}(A')
\\&\le\deg_{\partial_G\mathcal P_L\cup\partial_GA}(A')+\deg_{\partial_{G[A]}A'}(A')
\\&=\mathbf d(A')+\delta_{G[A]}A'.
\end{align*}

By properties~(\ref{item:CMG-property-1}) and~(\ref{item:CMG-property-2}) of \Cref{thm:CMG}, we have $\delta_{G[A]}R\le\phi\mathbf d(R)+\frac\phi{6 T}\mathbf d(A)$ and $\mathbf d(R)\le\mathbf d(A)/2$. By properties~(\ref{item:trimming-1}) and~(\ref{item:trimming-2}) of \Cref{thm:trimming}, we have $\delta_{G[A]}B\le2\delta_{G[A]}R+2\epsilon\phi\mathbf d(A)$ and $\mathbf d(B\setminus R)\le\frac1{6\phi}\,\delta_{G[A]}R+\frac\epsilon6\mathbf d(A)$. The only two options for the recursive instance $A'$ are $A'=R\cup B$ and $A'=A\setminus(R\cup B)$, and in both cases, we have
\begin{align*}
\delta_{G[A]}A'=\delta_{G[A]}(R\cup B)&\le\delta_{G[A]}R+\delta_{G[A]}B
\\&\le\delta_{G[A]}R+2\delta_{G[A]}R+2\epsilon\phi\mathbf d(A)
\\&=3\delta_{G[A]}R+\frac\phi{2T}\mathbf d(A)
\\&\le3\left(\phi\mathbf d(R)+\frac\phi{6 T}\mathbf d(A)\right)+\frac\phi{2T}\mathbf d(A)
\\&=3\phi\mathbf d(R)+\frac\phi T\mathbf d(A).
\end{align*}
Combining the two bounds so far, we obtain
\[ \mathbf d'(A')\le\mathbf d(A')+3\phi\mathbf d(R)+\frac\phi T\mathbf d(A) .\]

To bound $\mathbf d(A')$, we case on whether $A'=R\cup B$ or $A'=A\setminus(R\cup B)$. If $A'=A\setminus(R\cup B)$, then we must be in case~(\ref{item:partition-case-1}) of the algorithm, which means $\mathbf d(R)\ge\mathbf d(A)/(6T)$. In this case, we bound $\mathbf d(A')=\mathbf d(A\setminus(R\cup B))\le\mathbf d(A\setminus R)=\mathbf d(A)-\mathbf d(R)$. Together with the bound $\phi\le1/24$, we obtain
\begin{align*}
\mathbf d'(A')&\le\mathbf d(A')+3\phi\mathbf d(R)+\frac\phi T\mathbf d(A)
\\&\le\mathbf d(A)-\mathbf d(R)+3\phi\mathbf d(R)+\frac\phi T\mathbf d(A)
\\&\le\mathbf d(A)-\frac12\mathbf d(R)+\frac1{24T}\mathbf d(A)
\\&\le\mathbf d(A)-\frac12\cdot\frac{\mathbf d(A)}{6T}+\frac1{24T}\mathbf d(A)
\\&=\left(1-\frac1{24T}\right)\mathbf d(A),
\end{align*}
as promised. Otherwise, suppose that $A'=R\cup B$. We have
\begin{align*}
\mathbf d(R\cup B)&=\mathbf d(R)+\mathbf d(B\setminus R)
\\&\le\mathbf d(R)+\frac1{6\phi}\delta_{G[A]}R+\frac\epsilon6\mathbf d(A)
\\&\le\mathbf d(R)+\frac1{6\phi}\left(\phi\mathbf d(R)+\frac\phi{6 T}\mathbf d(A)\right)+\frac\epsilon6\mathbf d(A)
\\&\le\frac76\mathbf d(R)+\frac1{36T}\mathbf d(A)+\frac\epsilon6\mathbf d(A)
\\&\le\frac76\cdot\frac23\mathbf d(A)+\frac1{36}\mathbf d(A)+\frac16\mathbf d(A)
\\&=\frac{35}{36}\mathbf d(A)
\\&\le\left(1-\frac1{36T}\right)\mathbf d(A),
\end{align*}
as promised. With both cases established, this concludes the proof.
\end{proof}

For a given recursive call $A$, define its recursion depth inductively as follows: the initial call $A\gets V$ has depth $0$, and given a recursive call $A$ of depth $d$, all of its recursive calls have depth $d+1$. By \Cref{clm:recursive-d}, the value of $\mathbf d(A)$ decreases multiplicatively by factor $1/(36T)$ on each recursive call, so the maximum recursion depth is $O(T\log(nW))$.

For a given recursion depth $d$, let $E_d\subseteq E$ denote the union of edges $\partial_{G[A]}(R\cup B)$  over all instances $A$ of depth $d$. By construction of the algorithm, the (disjoint) union of $E_d$ over all recursion depths $d$ is exactly $\partial\mathcal P_{L+1}$. To avoid clutter, we also define $E_{<d}=E_1\cup\cdots\cup E_{d-1}$.

\begin{claim}\label{clm:flow-recursion-depth-d}
For any recursion depth $d\ge0$, there is a flow on $G$ with congestion $4$ such that each vertex $v\in V$ sends $\deg_{E_d}(v)$ flow and receives at most $48\phi\deg_{\partial\mathcal P_L\cup E_{<d}}(v)$ flow.
\end{claim}
\begin{proof}
We prove the statement by induction on $d\ge0$. The base case $d=0$ is satisfied with the empty flow: since $\mathcal A_0$ is the partition $\{V\}$ with a single part, each vertex $v\in A$ indeed sends $\deg_{\partial A_0}(v)=0$ flow. Now assume by induction that there is a flow on $G$ with congestion $4$ such that each vertex $v\in A$ sends $\deg_{E_d}(v)$ flow and receives at most $48\phi\deg_{\partial\mathcal P_L}(v)$ flow.

For each instance $A$ of depth $d$, the algorithm calls \Cref{thm:trimming} which defines $\mathbf d=\textup{deg}_{\partial\mathcal P_L\cup\partial A}|_A$. By property~(\ref{item:trimming-4}) of \Cref{thm:trimming}, there exists a vector $\mathbf t\in\mathbb R^A_{\ge0}$ with $\mathbf t\le24\phi\mathbf d|_{A\setminus(R\cup B)}$ and a flow $g$ on $G[A\setminus(R\cup B)]$ routing demand $\textup{deg}_{\partial_{G[A]}(R\cup B)}|_{A\setminus(R\cup B)}-\mathbf t$ with congestion $2$. We now construct a flow in $G[A]$ with congestion $4$ such that each vertex $v\in A$ sends $\deg_{E_d}(v)=\deg_{\partial_{G[A]}(R\cup B)}(v)$ flow and receives at most $48\phi\deg_{\partial\mathcal P_L\cup\partial E_{<d}}(v)$ flow. First, for each edge in $\partial_{G[A]}(R\cup B)$, send flow to full capacity in the direction from $R\cup B$ to $A\setminus(R\cup B)$. In this initial flow, each vertex $v\in R\cup B$ sends exactly $\deg_{\partial_{G[A]}(R\cup B)}(v)$ flow, and each vertex $v\in A\setminus(R\cup B)$ receives exactly $\deg_{\partial_{G[A]}(R\cup B)}(v)$ flow. Next, we send the flow $g$ scaled by $2$, so that each vertex $v\in A\setminus(R\cup B)$ sends exactly $2\deg_{\partial_{G[A]}(R\cup B)}(v)$ flow and each vertex $v\in A\setminus(R\cup B)$ receives at most $48\phi\mathbf d(v)$ flow. Note that $\mathbf d(v)=\deg_{\partial\mathcal P_L\cup\partial A}(v)\le\deg_{\partial\mathcal P_L\cup\partial E_{<d}}(v)$ since $\partial A\subseteq\partial E_{<d}$. Summing the two flows, we obtain a flow in $G[A]$ such that each vertex $v\in A$ sends $\deg_{E_d}(v)=\deg_{\partial_{G[A]}(R\cup B)}(v)$ flow and receives at most $48\phi\deg_{\partial\mathcal P_L\cup\partial E_{<d}}(v)$ flow. The congestion of the flow is $4$, since edges in $\partial_{G[A]}(R\cup B)$ have congestion $1$ in the initial flow, and edges in $G[A\setminus(R\cup B)]$ have congestion $2$ in the flow $g$ scaled by $2$.

To complete the induction, our final flow is the union of the constructed flow over all recursive instances $A$ of depth $d$. Since the flow for instance $A$ is in $G[A]$, and since the instances $A$ are disjoint, the flows are also disjoint over all $A$. It follows that their union is a flow on $G$ with congestion $4$ such that each vertex $v\in A$ sends $\deg_{E_d}(v)$ flow and each vertex $v\in V$ receives $48\phi\deg_{\partial\mathcal P_L\cup E_{<d}}(v)$ flow.
\end{proof}

Finally, the two claims below establish property~(\ref{item:algo-3}) of \Cref{thm:cut-approx-next-level} and assumption~\ref{item:assumption-CMG}, respectively.
\begin{claim}\label{clm:algo-property-3}
Property~(\ref{item:algo-3}) of \Cref{thm:cut-approx-next-level} holds for $i=L$, i.e., there is a flow in $G$ with congestion $O(T\log(nW))$ such that each vertex $v\in V$ sends $\deg_{\partial\mathcal P_{L+1}}(v)$ flow and receives at most $\frac12\deg_{\partial\mathcal P_L}(v)$ flow.
\end{claim}
\begin{proof}
Let $D=O(T\log(nW))$ be the maximum recursion depth.
Summing the flows from \Cref{clm:flow-recursion-depth-d} over all recursion depths $d$, and using that $E_1\cup\cdots\cup E_d=\partial\mathcal P_{L+1}$, we obtain a flow with congestion $O(T\log(nW))$ such that each vertex $v\in V$ sends $\deg_{\partial\mathcal P_{L+1}}(v)$ flow and receives at most $48D\phi\deg_{\partial\mathcal P_L\cup\partial\mathcal P_{L+1}}(v)$ flow. Recall that we set $\phi\gets\frac1{C\log^3(nW)}$ for large enough constant $C>0$. We choose $C$ large enough that $48D\phi\le1/4$, so that each vertex $v\in V$ receives at most $\frac13\deg_{\partial\mathcal P_L\cup\partial\mathcal P_{L+1}}(v)\le\frac13\deg_{\partial\mathcal P_L}(v)+\frac13\deg_{\partial\mathcal P_{L+1}}(v)$ flow. We can cancel out at most $\frac13\deg_{\partial\mathcal P_{L+1}}(v)$ flow received at each vertex $v\in V$ from the $\deg_{\partial\mathcal P_{L+1}}(v)$ flow sent. After cancellation, we obtain a flow with congestion $O(T\log(nW))$ such that each vertex $v\in V$ sends at least $\frac23\deg_{\partial\mathcal P_{L+1}}(v)$ flow and receives at most $\frac13\deg_{\partial\mathcal P_L}(v)$ flow. Scaling the flow by factor $3/2$, taking a path decomposition, and removing enough paths until each vertex is the start of exactly $\deg_{\partial\mathcal P_{L+1}}(v)$ paths, we obtain the desired flow with congestion $O(T\log(nW))$.
\end{proof}
\begin{claim}
For each recursive instance $A$, the assumption~\ref{item:assumption-CMG} of \Cref{thm:CMG,thm:trimming} hold, i.e., there is a flow on $G$ with congestion $O(\log^3(nW))$ such that each vertex $v\in A$ is the source of $c_G(\{v\},V\setminus A)$ flow and each vertex $v\in V$ is the sink of at most $\deg_{\partial\mathcal P_L}(v)$ flow.
\end{claim}
\begin{proof}
By \Cref{clm:algo-property-3}, there is a flow in $G$ with congestion $O(\log^3(nW))$ such that each vertex $v\in V$ sends $\deg_{\partial\mathcal P_{L+1}}(v)$ flow and receives at most $\frac12\deg_{\partial\mathcal P_L}(v)$ flow. Observe that for any recursive instance $A$, we have $\partial A\subseteq\partial\mathcal P_{L+1}$ since the recursive algorithm starting at instance $A$ adds a partition of $A$ into $\mathcal P_{L+1}$. In particular, each vertex $v\in A$ sends at least $\deg_{\partial A}(v)=c_G(\{v\},V\setminus A)$ flow. Take a path decomposition of the flow and remove enough paths until each vertex is the start of exactly $c_G(\{v\},V\setminus A)$ paths. The resulting flow satisfies assumption~\ref{item:assumption-CMG}, concluding the proof.
\end{proof}

It remains to bound the running time of the algorithm for \Cref{thm:cut-approx-next-level}. For each instance $A$, \Cref{thm:CMG,thm:trimming} run in $\tilde O(|A|+m')$ time plus $O(\log^2(nW))$ calls to \Cref{thm:fair-cut}, which takes $\tilde O(|A|+m')$ time per call, for a total time of $\tilde O(|A|+m')$. The instances $A$ on a given recursion depth are disjoint, so the sum of $|A|+m'$ over all such instances $A$ is $O(m)$. The maximum recursion depth is $O(T\log(nW))=O(\log^3(nW))$, so the sum of $|A|+m'$ over all instances of the algorithm is $O(m\log^3(nW))$. It follows that the algorithm of \Cref{thm:cut-approx-next-level} runs in $\tilde O(m)$ time.

\section{Approximate Maximum Flow}

From \Cref{thm:cut-approx-next-level} and \Cref{clm:L}, we obtain an algorithm that constructs a congestion-approximator of quality $O(\log^{10}(nW))$ in $\tilde O(m)$ time. Recall that Sherman's framework~\cite{She13,She17} translates a congestion-approximator of quality $\alpha$ to a $(1+\epsilon)$-approximate maximum flow algorithm with running time $\tilde O(\epsilon^{-1}\alpha m)$. Thus, for any parameter $\epsilon>0$, we obtain a $(1+\epsilon)$-approximate maximum flow algorithm with running time $\tilde O(\epsilon^{-1}m)$.


\section*{Acknowledgement}
JL would like to thank Evangelos Kosinas for pointing out an error in the initial version of \Cref{lem:construct-f-b}.

\bibliographystyle{alpha}
\bibliography{ref}

\appendix

\section{Cut-Matching Game}\label{sec:CMG}

In this section, we prove \Cref{thm:CMG}, restated below.
\CMG*

Our setup resembles Appendix~B of~\cite{SW19} with a few minor changes. In particular, since our flow routine is more restrictive, we have to adapt the algorithm to handle our flow outputs.

We begin with notation from~\cite{SW19}. For simplicity, we avoid working with the subdivision graph in~\cite{RST14,SW19}. Define a $A$-commodity flow as a multi-commodity flow where each vertex $v\in A$ is the source of quantity $\mathbf d(v)$ of its distinct flow commodity. Only for analysis, we consider a $A\times A$ \emph{flow-matrix} $\mathbf F\in\mathbb R^{A\times A}_{\ge0}$ which encodes information about a $A$-commodity flow. We say that $\mathbf F$ is \emph{routable with congestion $c$} if there exists a $A$-commodity flow $f$ such that, simultaneously for all $u,v\in A$, we have that $u$ can send quantity $\mathbf F(u,v)$ of its own commodity to $v$, and the amount of flow through each edge is at most $c$.

The algorithm initializes flow-matrix $\mathbf F_0\in\mathbb R^{A\times A}_{\ge0}$ as the diagonal matrix with value $\mathbf d(v)$ on entry $\mathbf F(v,v)$. Trivially, $\mathbf F$ is routable with zero congestion. The algorithm initializes $A_0=A$ and $R_0=\emptyset$, and then proceeds for at most $T=O(\log^2(nW))$ rounds. For each round $t\in[T]$, the algorithm \emph{implicitly} updates $\mathbf F_{t-1}$ to $\mathbf F_t$ such that it is routable with congestion $t/\phi$. The operation for implicitly updating $\mathbf F_{t-1}$ will be described explicitly later on, but we ensure that row sums do not change from $\mathbf F_{t-1}$ to $\mathbf F_t$, i.e., there is always $\mathbf d(v)$ total quantity of each commodity $v\in A$ spread among the vertices. For each round $t\in[T]$, the algorithm (explicitly) finds a partition of $A_{t-1}$ into $A_t^\ell,A_t^r$, and then computes
 \begin{enumerate}[(i)]
 \item A (possibly empty) set $S_t\subseteq A$ satisfying $\delta_{G[A]}S_t\le\phi\mathbf d(S\cap A_{t-1})+\frac\phi{6 T^2}\mathbf d(A)$ and $\mathbf d(S_t)\le0.6\mathbf d(A)$, and\label{item:condition-St}
 \item A (possibly empty) flow $f_t$ from $A_t^\ell\setminus S_t$ to $A_t^r\setminus S_t$ such that each vertex $v\in A_t^\ell\setminus S_t$ is the source of exactly $\mathbf d(v)/24$ flow, and each vertex $v\in A_t^r\setminus S_t$ is the sink of at most $\mathbf d(v)$ flow. The flow $f$ has congestion $1/\phi$.\label{item:condition-ft}
 \end{enumerate}
The algorithm then updates $A_t\gets A_{t-1}\setminus S_t$ and $R_t\gets R_{t-1}\cup S_t$. Note that on each round $t$, the sets $A_t$ and $R_t$ partition $A$. If $\mathbf d(R_t)\ge\mathbf d(A)/(6 T)$ holds, then the algorithm immediately terminates and outputs $R=R_t$. Otherwise, we have $\mathbf d(R_T)<\mathbf d(A)/(6 T)$ at the end, and the algorithm outputs $R=R_T$.
\begin{lemma}\label{lem:set-R}
For any round $t$, we have $\delta_{G[A]}R_t\le\phi\mathbf d(R_t)+\frac\phi{6 T}\mathbf d(A)$ and $\mathbf d(R_t)\le0.6\mathbf d(A)$.
\end{lemma}
\begin{proof}
We start by proving the first statement. Each time we remove a set $S_t$, we are guaranteed that $\delta_{G[A]}S_t\le\phi\mathbf d(S_t\cap A_{t-1})+\frac\phi{6 T^2}\mathbf d(A)$. We charge the $\phi\mathbf d(S_t\cap A_{t-1})$ part to the vertices in $S_t\cap A_{t-1}$ so that each vertex $v\in S_t\cap A_{t-1}$ is charged exactly $\phi\mathbf d(v)$. Since the algorithm updates $A_t\gets A_{t-1}\setminus S_t$ and $R_t\gets R_{t-1}\cup S_t$, each newly charged vertex leaves $A_t$ and joins $R_t$. In total, we charge $\sum_{t=1}^T\phi\mathbf d(S_t\cap A_{t-1})$ to the vertices in $R_t$ so that each vertex $v\in R_t$ is charged once at exactly $\phi\mathbf d(v)$. It follows that $\sum_{t=1}^T\phi\mathbf d(S_t\cap A_{t-1})\le\phi\mathbf d(R_t)$ and
\[ \delta_{G[A]}R_t\le\sum_{t=1}^T\delta_{G[A]}S_t\le\sum_{t=1}^T\left(\phi\mathbf d(S_t\cap A_{t-1})+\frac\phi{6 T^2}\mathbf d(A)\right)\le\phi\mathbf d(R_t)+T\cdot\frac\phi{6 T^2}\mathbf d(A) ,\]
concluding the first statement of the lemma.

For the second statement, consider the round $t$ with $\mathbf d(R_t)\ge\mathbf d(A)/(6 T)$, if it exists, at which point the algorithm terminates. (If there is no such $t$, then we are done.) We have $\mathbf d(R_{t-1})<\mathbf d(A)/(6 T)$ and $\mathbf d(R_{t-1}\cup S_t)=\mathbf d(R_t)\ge\mathbf d(A)/(6 T)$, and the final set $S_t$ satisfies $\mathbf d(S_t)\le0.6\mathbf d(A_{t-1})\le0.6\mathbf d(A)$. It follows that the new set $R_t$ satisfies
\begin{align*}
\mathbf d(R_t)=\mathbf d(R_{t-1}\cup S_t)\le\mathbf d(R_{t-1})+\mathbf d(S_t)\le\frac{\mathbf d(A)}{6 T}+0.6\mathbf d(A)\le\frac23\mathbf d(A)
\end{align*}
for large enough $T$, concluding the second statement and the proof.
\end{proof}

Purely for the analysis, we define a potential function $\psi(t)$ for each round $t\in[T]$ as follows. For each vertex $v\in A$, let $\mathbf F_t(u)\in\mathbb R^A_{\ge0}$ be row $u$ of matrix $\mathbf F_t$; we call $\mathbf F_t(u)$ a \emph{flow-vector} of $u$. We define the potential function
\[ \psi(t)=\sum_{u\in A_t}\mathbf d(u)\left\lVert\frac{\mathbf F_t(u)}{\mathbf d(u)}-\boldsymbol\mu_t\right\rVert_2^2 \]
where
\begin{gather}
\boldsymbol\mu_t=\frac{\sum_{u\in A_t}\mathbf F_t(u)}{\mathbf d(A_t)} = \arg\min_{\boldsymbol\mu(v)\in\mathbb R^A}\sum_{u\in A_t}\mathbf d(u)\left\lVert\frac{\mathbf F_t(u)}{\mathbf d(u)}-\boldsymbol\mu\right\rVert_2^2\label{eq:mu-t}
\end{gather}
is the weighted average of the flow-vectors in $A_t$, which is also the minimizer of $\psi(t)$ when treated as a function of $\boldsymbol\mu$. The latter fact can be verified separately for each coordinate $v\in A$; namely,
\[ \boldsymbol\mu_t(v)=\frac{\sum_{u\in A_t}\mathbf F_t(u,v)}{\mathbf d(A_t)} =\arg\displaystyle\min_{\boldsymbol\mu\in\mathbb R^A}\sum_{u\in A_t}\mathbf d(u)\bigg(\frac{\mathbf F_t(u,v)}{\mathbf d(u)}-\boldsymbol\mu(v)\bigg)^2 \]
follows from setting the derivative to zero.

\subsection{Small Potential Implies Mixing}\label{sec:mixing}

We first show that if $\Phi(t)\le1/\textup{poly}(nW)$ for a sufficiently large polynomial, then the vertex weighting $\mathbf d|_{A_t}$ mixes in $G$.

\begin{lemma}\label{lem:CMG-mix}
For any $t\in T$, if $\mathbf d(R_t)\le\mathbf d(A)/(6 T)$ and $\psi(t)\le1/(nW)^C$ for large enough constant $C>0$, then the vertex weighting $\mathbf d|_{A_t}$ mixes in $G$ with congestion $5 T/\phi$.
\end{lemma}
For the rest of \Cref{sec:mixing}, we prove \Cref{lem:CMG-mix}.
Suppose that $\mathbf d(R_t)\le\mathbf d(A)/(6 T)$ and $\psi(t)\le1/(nW)^C$ for large enough constant $C>0$. We first prove two claims about the flow-matrix $\mathbf F_t$. Let $\mathbf F_t(A_t,A_t)$ denote the sum $\sum_{u,v\in A_t}\mathbf F_t(u,v)$.

\begin{claim}\label{clm:sum-Ft}
$\mathbf F_t(A_t,A_t)\ge\mathbf d(A)/3$.
\end{claim}
\begin{proof}
Since the $A$-commodity flow routing $\mathbf F_t$ has congestion $t/\phi$, and since $\partial_{G[A]}R_t$ is a cut separating $A_t\subseteq A\setminus R_t$ from $R_t$, we have $\sum_{u\in A_t,v\in R_t}\mathbf F_t(u,v)\le t/\phi\cdot\delta_{G[A]}R_t$. Since $\delta_{G[A]}R_t\le\phi\mathbf d(R_t)+\frac\phi6\mathbf d(A)$ by \Cref{lem:set-R}, we have
\[ \sum_{u\in A_t,v\in R_t}\mathbf F_t(u,v)\le\frac{t}\phi\cdot\delta_{G[A]}R_t\le\frac{t}\phi\cdot\left(\phi\mathbf d(R_t)+\frac\phi{6 T}\mathbf d(A)\right)\le T\cdot\mathbf d(R_t)+\frac16\mathbf d(A)\le\frac13\mathbf d(A) ,\]
where the last inequality holds by the assumption $\mathbf d(R_t)\le\mathbf d(A)/(6 T)$.
Since the sum of row $v\in A$ in $\mathbf F$ is always $\mathbf d(v)$, we have
\[ \sum_{u\in A_t,v\in A}\mathbf F_t(u,v)=\mathbf d(A_t)=\mathbf d(A)-\mathbf d(R_t)\ge\mathbf d(A)-\frac{\mathbf d(A)}{6 T}\ge\frac23\mathbf d(A) .\]
Subtracting the two inequalities above, we conclude that
\[ \mathbf F_t(A_t,A_t)=\sum_{u\in A_t,v\in A}\mathbf F_t(u,v)-\sum_{u\in A_t,v\in R_t}\mathbf F_t(u,v)\ge\frac23\mathbf d(A)-\frac13\mathbf d(A)=\frac13\mathbf d(A) .\qedhere\]
\end{proof}

\begin{claim}\label{clm:sum-Ft-2}
For all $u\in A$, we have
\[ \bigg|\sum_{v\in A_t}\mathbf F_t(u,v)-\mathbf d(u)\cdot\frac{\mathbf F_t(A_t,A_t)}{\mathbf d(A_t)}\bigg|\le\frac1{(nW)^{C/2-3}} .\]
In particular, $\sum_{v\in A_t}\mathbf F_t(u,v)\ge\mathbf d(u)/4$.
\end{claim}
\begin{proof}
Since $\psi(t)\le1/(nW)^C$ by assumption, we have
\[ \bigg|\frac{\mathbf F_t(u,v)}{\mathbf d(u)}-\boldsymbol\mu_t(v)\bigg|^2\le\mathbf d(u)\,\bigg|\frac{\mathbf F_t(u,v)}{\mathbf d(u)}-\boldsymbol\mu_t(v)\bigg|^2\le\sum_{u\in A_t}\mathbf d(u)\left\lVert\frac{\mathbf F_t(u)}{\mathbf d(u)}-\boldsymbol\mu_t\right\rVert_2^2=\psi(t)\le\frac1{(nW)^C}, \]
so
\[ \big|\mathbf F_t(u,v)-\mathbf d(u)\cdot\boldsymbol\mu_t(v)\big|=\mathbf d(u)\,\bigg|\frac{\mathbf F_t(u,v)}{\mathbf d(u)}-\boldsymbol\mu_t(v)\bigg|\le nW\cdot\frac1{(nW)^{C/2}}=\frac1{(nW)^{C/2-1}}. \]
In particular,
\begin{align*}
\bigg|\sum_{v\in A_t}\big(\mathbf F_t(u,v)-\mathbf d(u)\cdot\boldsymbol\mu_t(v)\big)\bigg|&\le\sum_{v\in A_t}\big|\mathbf F_t(u,v)-\mathbf d(u)\cdot\boldsymbol\mu_t(v)\big|
\\&\le\sum_{v\in A_t}\frac{\mathbf d(u)}{(nW)^{C/2-1}}=|A_t|\cdot\frac{\mathbf d(u)}{(nW)^{C/2-1}}\le\frac1{(nW)^{C/2-3}} .
\end{align*}
By the definition of $\boldsymbol\mu_t$,
\[ \sum_{v\in A_t}\boldsymbol\mu_t(v)=\sum_{u,v\in A_t}\frac{\mathbf F_t(u,v)}{\mathbf d(A_t)}=\frac{\mathbf F_t(A_t,A_t)}{\mathbf d(A_t)} ,\]
concluding the first statement of the claim. By \Cref{clm:sum-Ft}, the expression above is at least $\frac{\mathbf d(A)/3}{\mathbf d(A_t)}\ge1/3$, so
\[ \sum_{v\in A_t}\mathbf F_t(u,v)\ge\sum_{v\in A_t}\mathbf d(u)\cdot\boldsymbol\mu_t(v)-\frac1{(nW)^{C/2-3}}\ge\frac13\mathbf d(u)-\frac1{(nW)^{C/2-3}}\ge\frac14\mathbf d(u) \]
for $C>0$ large enough, concluding the second statement.
\end{proof}

With \Cref{clm:sum-Ft,clm:sum-Ft-2} established, we now prove \Cref{lem:CMG-mix}.
Recall that $\mathbf F_t$ is routable with congestion $ T/\phi$. Decompose this $A$-commodity flow into single-commodity flows $f_{u,v}:u,v\in A$ that send quantity $\mathbf F_t(u,v)$ of commodity $u$ from $u$ to $v$.

Consider any demand $\mathbf b\in\mathbb R^A$ satisfying $|\mathbf b|\le\mathbf d|_{A_t}$. In particular, $\sum_{u\in A_t}\mathbf b(u)=0$. We want to construct a single-commodity flow routing demand $\mathbf b$ with congestion $5 T/\phi$. For each $u,v\in A_t$, we first route the flow
\[ f'_{u,v}=\frac{f_{u,v}}{\sum_{v'\in A_t}\mathbf F_t(u,v')}\cdot \mathbf b(u) .\]
Summing over all $v\in A_t$, we observe that each vertex $u\in A_t$ sends a total of $\mathbf b(u)$ demand. Also, by \Cref{clm:sum-Ft-2}, each flow $f'_{u,v}$ has (absolute) value
\[ \frac{\mathbf F_t(u,v)}{\sum_{v'\in A_t}\mathbf F_t(u,v')}\cdot|\mathbf b(u)|\le \frac{\mathbf F_t(u,v)}{\mathbf d(u)/4}\cdot\mathbf d(u)\le4\mathbf F_t(u,v). \]
In particular, these flows can be routed simultaneously with congestion $4$ times the $A$-commodity flow, which is congestion $4 T/\phi$.

After routing this flow, each vertex $v\in A_t$ receives demand
\[ \sum_{u\in A_t}\frac{\mathbf F_t(u,v)}{\sum_{v'\in A_t}\mathbf F_t(u,v')}\cdot \mathbf b(u) ,\]
Since $\psi(t)\le1/(nW)^C$, we can approximate $\mathbf F_t(u,v)\approx\mathbf d(u)\cdot\boldsymbol\mu_t(v)$ for all $u\in A$. Together with \Cref{clm:sum-Ft-2}, we can approximate the numerator and denominator of the fraction above to
\[ \sum_{u\in A_t}\frac{\mathbf d(u)\cdot\boldsymbol\mu_t(v)}{\mathbf d(u)\cdot \mathbf F_t(A_t,A_t)/\mathbf d(A_t)}\cdot \mathbf b(u) = \frac{\boldsymbol\mu_t(v)}{\mathbf F_t(A_t,A_t)/\mathbf d(A_t)}\sum_{u\in A_t}\mathbf b(u), \]
which equals $0$ since $\sum_{u\in A_t}\mathbf b(u)=0$. When $C>0$ is large enough, the approximated value of $0$ is within an additive $1/n^2$ of the true value. In particular, each vertex $v\in A_t$ receives a demand whose absolute value is at most $1/n^2$. Since the minimum edge capacity is $1$, the remaining demand can trivially be routed with congestion $1$. The final congestion is $4 T/\phi+1\le5 T/\phi$, concluding the proof of \Cref{lem:CMG-mix}.

\subsection{Computing the Set and Flow}
In this subsection, we describe a round of the algorithm in detail, following Lemma~B.3 of \cite{SW19} with some minor changes. Recall that on each iteration $t\in[T]$, the algorithm first computes a partition of $A_{t-1}$ into $A_t^\ell,A_t^r$, and then computes a set $S_t$ and flow $f_t$ satisfying certain properties. The algorithm then updates $A_t\gets A_{t-1}\setminus S_t$ and $R_t\gets R_{t-1}\cup S_t$. The choice of partition $A_t^\ell,A_t^r$ will depend on an implicitly represented flow-matrix $\mathbf F_t$ that is useful for the analysis.

\subsubsection{Constructing the partition}\label{sec:construct-A}

We start with the construction of $A_t^\ell$ and $A_t^r$. We first list some variables key to the algorithm and analysis.
\begin{enumerate}
\item Let $\mathbf r\in\mathbb R^A$ be a random unit vector orthogonal to the all-ones vector.
\item For each $v\in A$, let $\mathbf p(v)=\langle \mathbf F_{t-1}(v)/\mathbf d(v),\,\mathbf r\rangle$ be the projection of normalized flow-vector $\mathbf F_{t-1}(v)/\mathbf d(v)$ onto the vector $r$. We later show in \Cref{clm:compute-projection} that the values $\mathbf p(v)$ can be computed in total time $O(mT)$ without explicitly maintaining the flow matrix $F$.
\item Let $\bar\mu_{t-1}=\langle\boldsymbol\mu_{t-1},\mathbf r\rangle$ be the projection of the weighted average $\boldsymbol\mu_{t-1}=\sum_{u\in A_{t-1}}\mathbf F_{t-1}(u)/\mathbf d(A_{t-1})$ onto the vector $\mathbf r$. It is only used in the analysis.
\item Let $A_t^\ell$ and $A_t^r$ be constructed by \Cref{lem:RST-CMG} below.
\end{enumerate}
We first cite a lemma of~\cite{RST14} whose proof translates directly to the setting with vertex weighting $\mathbf d$.
\begin{lemma}[Lemma~3.3 of~\cite{RST14}]\label{lem:RST-CMG}
Given the values $\mathbf p(v)$ for all $v\in A_{t-1}$, we can find in time $O(|A_{t-1}|\log|A_{t-1}|)$ a partition of $A_{t-1}$ into two sets $A_t^\ell,A_t^r$ and a separation value $\eta\in\mathbb R$ such that
 \begin{enumerate}[(a)]
 \item $\eta$ separates the projections of $A_t^\ell,A_t^r$, i.e., either $\displaystyle\max_{u\in A_t^\ell}\mathbf p(u)\le\eta\le\min_{v\in A_t^r}\mathbf p(v)$ or $\displaystyle\min_{u\in A_t^\ell}\mathbf p(u)\ge\eta\ge\max_{v\in A_t^r}\mathbf p(v)$,\label{item:CMG-separator-1}
 \item $\mathbf d(A_t^\ell)\le\mathbf d(A_{t-1})/8$ and $\mathbf d(A_t^r)\ge\mathbf d(A_{t-1})/2$,\label{item:CMG-separator-2}
 \item $(\mathbf p(v)-\eta)^2\ge\frac19(\mathbf p(v)-\bar\mu_{t-1})^2$ for each vertex $v\in A^\ell_t$, and\label{item:CMG-separator-3}
 \item $\sum_{v\in A_t^\ell}\mathbf d(v)\cdot(\mathbf p(v)-\bar\mu_{t-1})^2\ge\frac1{80}\sum_{v\in A_{t-1}}\mathbf d(v)\cdot(\mathbf p(v)-\bar\mu_{t-1})^2$.\label{item:CMG-separator-4}
 \end{enumerate}
\end{lemma}

We conclude this section with a lemma about the behavior of projections onto a random unit vector. 
\begin{lemma}[Lemma~3.4 of~\cite{KRV09}]\label{lem:KRV-CMG}
For all vertices $v\in A$, we have
\[ \mathbb E[(\mathbf p(v)-\bar\mu_{t-1})^2]=\frac1n\left\lVert\frac{\mathbf F_{t-1}(v)}{\mathbf d(v)}-\boldsymbol\mu_{t-1}\right\rVert_2^2, \]
and for any pair $u,v\in A$ and constant $c>0$, we have
\[ (\mathbf p(u)-\mathbf p(v))^2\le\frac{c\log n}{n}\left\lVert\frac{\mathbf F_{t-1}(u)}{\mathbf d(u)}-\frac{\mathbf F_{t-1}(v)}{\mathbf d(v)}\right\rVert_2^2 \]
and
\[ (\mathbf p(u)-\bar\mu_{t-1})^2\le\frac{c\log n}{n}\left\lVert\frac{\mathbf F_{t-1}(u)}{\mathbf d(u)}-\boldsymbol\mu_{t-1}\right\rVert_2^2 ,\]
each with probability at least $1-n^{-c/4}$. 
\end{lemma}

\subsubsection{Max-flow/min-cut call}\label{sec:max-flow-min-cut}

Given $A_t^\ell$ and $A_t^r$, we call \Cref{thm:fair-cut} with parameters
\begin{gather} A\gets A,\,\epsilon\gets\frac1{18 T^2},\,\gamma\gets\frac{\epsilon\phi}2,\,\beta\gets\max\{1,(12\phi+\epsilon\gamma)(\kappa+2)\},\,\,\mathbf s\gets\phi\mathbf d|_{A_{t-1}\setminus A_t^r}+\epsilon\phi\mathbf d,\,\mathbf t\gets12\phi\mathbf d|_{A_t^r} ,\label{eq:parameters}\end{gather}
and we denote the graph $G[A,\gamma,\mathbf s,\mathbf t]$ by $H=(V_H,E_H)$. We will later show in \Cref{lem:stronger-beta} that Assumption~\ref{item:property-Ast2} of \Cref{thm:fair-cut} is satisfied with our parameter $\beta$.

From \Cref{thm:fair-cut}, we obtain an $(1+\epsilon)$-approximate fair cut/flow pair $(S,f)$. The algorithm sets $S_t\gets S\setminus\{s,x\}$, which satisfies the condition below as required by step~\ref{item:condition-St}.

\begin{claim}
$\delta_{G[A]}S_t\le\phi\mathbf d(S\cap A_{t-1})+3\epsilon\phi\mathbf d(A)$ and $\mathbf d(S_t)\le0.6\mathbf d(A)$, 
\end{claim}
\begin{proof}
We begin with the first statement. We begin by bounding $\delta_HS$ as follows. The edges in $\partial_HS$ can be split into three groups:
 \begin{enumerate}
 \item The edges in $G[A]$, which have total capacity $\delta_{G[A]}(S\setminus\{s,x\})=\delta_{G[A]}S_t$,
 \item The edges incident to $s$, which have total capacity
\[ c_H(\{s\},V_H\setminus S)=\mathbf s(A\setminus S_t)\ge\phi\mathbf d((A_{t-1}\setminus A_t^r)\setminus S_t) ,\]
 \item The edges incident to $t$ or $x$, which we ignore.
 \end{enumerate}
It follows that
\begin{gather*} \delta_HS\ge\delta_{G[A]}S_t+\phi\mathbf d((A_{t-1}\setminus A_t^r)\setminus S_t) .\end{gather*}
By \Cref{fact:fair-cut}, the cut value $\delta_HS$ is at most $(1+\epsilon)$ times the minimum $(s,t)$-cut. Since the $(s,t)$-minimum cut is at most $\delta_H\{s\}=\mathbf s(A)$, we have
\[ \delta_HS\le(1+\epsilon)\mathbf s(A)=(1+\epsilon)(\phi\mathbf d(A_{t-1}\setminus A_t^r)+\epsilon\phi\mathbf d(A))\le(1+\epsilon)\phi\mathbf d(A_{t-1}\setminus A_t^r)+2\epsilon\phi\mathbf d(A) .\]
It follows that
\begin{align*}
\delta_{G[A]}S_t&\le\delta_HS-\phi\mathbf d((A_{t-1}\setminus A_t^r)\setminus S_t)
\\&\le(1+\epsilon)\phi\mathbf d(A_{t-1}\setminus A_t^r)+2\epsilon\phi\mathbf d(A)-\phi\mathbf d((A_{t-1}\setminus A_t^r)\setminus S_t)
\\&=\phi\mathbf d(S_t\cap (A_{t-1}\setminus A_t^r))+\epsilon\phi\mathbf d(A_{t-1}\setminus A_t^r)+2\epsilon\phi\mathbf d(A)
\\&\le\phi\mathbf d(S_t\cap A_{t-1})+3\epsilon\phi\mathbf d(A) ,
\end{align*}
concluding the first statement. For the second statement, observe that $\delta_HS\le(1+\epsilon)\mathbf s(A)\le(1+\epsilon)^2\phi\mathbf d(A)$ and $\delta_HS\ge c_H(S,\{t\})=12\phi\mathbf d(S_t\cap A_t^r)$, so
\[ \mathbf d(S_t\cap A_t^r)\le\frac{\delta_HS}{12\phi}\le\frac{(1+\epsilon)^2\phi\mathbf d(A)}{12\phi}\le\frac{\mathbf d(A)}{10} \]
for $\epsilon>0$ small enough.
By property~\ref{item:CMG-separator-2} of \Cref{lem:RST-CMG}, we have $\mathbf d(A_t^r)\ge\mathbf d(A_{t-1})/2$, so $\mathbf d(S_t)\le\mathbf d(S_t\cap A_t^r)+\mathbf d(A_{t-1}\setminus A_t^r)\le\mathbf d(A)/10+\mathbf d(A_{t-1})/2\le0.6\mathbf d(A)$, concluding the second statement and the proof.
\end{proof}
The algorithm defines the flow $f_t$ as follows. First, scale the flow $f$ by factor $1/(12\phi)$ and restrict the flow to edges in $G[A\setminus S_t]$. In other words, the flow on edges incident to $S_t\cup\{s,t,x\}$ are removed. Then, decompose the flow $f_t$ into paths (using, for example, a dynamic tree~\cite{ST83}). For each vertex $v\in A_t^\ell\setminus S_t$, remove enough paths starting at $v$ until it is the start of exactly $\mathbf d(v)/24$ total capacity of paths, and for each vertex $v\notin A_t^\ell\setminus S_t$, remove all paths starting at $v$; let the remaining flow be $f_t$. The claim below shows that this flow removal step is always possible, and that $f_t$ satisfies the condition below as required by step~\ref{item:condition-ft}.
\begin{claim}\label{clm:properties-ft}
$f_t$ is a flow from $A_t^\ell\setminus S_t$ to $A_t^r\setminus S_t$ such that each vertex $v\in A_t^\ell\setminus S_t$ is the source of exactly $\mathbf d(v)/24$ flow, and each vertex $v\in A_t^r\setminus S_t$ is the sink of at most $\mathbf d(v)$ flow. The flow $f$ has congestion $1/(12\phi)$.
\end{claim}
\begin{proof}
It suffices to show that the scaled flow $f/(12\phi)$, once restricted to edges in $G$, is a flow such that
 \begin{enumerate}
 \item Each vertex $v\in A_t^\ell\setminus S_t$ is the source of at least $\mathbf d(v)/24$ flow, and
 \item The only sinks are at vertices $v\in A_t^r\setminus S_t$, and each such vertex $v$ is the sink of at most $\mathbf d(v)$ flow.
 \end{enumerate}
The path-removing step in the algorithm is then possible, and ensures that each vertex $v\in A_t^\ell\setminus S_t$ is the source of exactly $\mathbf d(v)/24$ flow in $f_t$, each vertex $v\in A_t^r\setminus S_t$ is the sink of at most $\mathbf d(v)$ flow, and there are no (nonzero) sources or sinks elsewhere.

We begin with the unscaled flow $f$. Since $(S,f)$ is a $(1+\epsilon)$-fair cut/flow pair, each vertex $v\in V_H\setminus(S\cup\{t\})$ receives at least $\frac1{1+\epsilon}\,c_H(\{v\},S)\ge\frac12 c_H(\{v\},S)$ total flow from vertices in $S$. Since $A\setminus S_t\subseteq V_H\setminus(S\cup\{t\})$, the same holds for all $v\in A\setminus S_t$. By construction of $H$, we have $c_H(\{v\},S)\ge c_H(v,s)=\phi\mathbf d|_{A_{t-1}\setminus A_t^r}(v)+\epsilon\phi\mathbf d(v)$ for all $v\in A\setminus S_t$. It follows that each vertex $v\in A\setminus S_t$ receives at least $\frac\phi2\mathbf d|_{A_{t-1}\setminus A_t^r}(v)+\frac{\epsilon\phi}2\mathbf d(v)$ total flow from vertices in $S$.

We now investigate the effect of restricting the flow $f$ to edges in $G[A\setminus S_t]$, starting with removing all edges incident to $S$. Continuing the argument above, removing these edges causes each vertex $v\in A\setminus S_t$ to be the source of at least $\frac\phi2\mathbf d|_{A_{t-1}\setminus A_t^r}(v)+\frac{\epsilon\phi}2\mathbf d(v)$ flow.

If $x\notin S$, then we now remove the edges incident to $x$. By construction of $H$, each vertex $v\in A$ has an edge to $x$ of capacity $\gamma\mathbf d(v)=\frac{\epsilon\phi}2\mathbf d(v)$. Since the flow $f$ is feasible, there is at most $\frac{\epsilon\phi}2\mathbf d(v)$ flow along the edge $(v,x)$. Removing this edge changes the net flow out of $v$ by at most $\frac{\epsilon\phi}2\mathbf d(v)$. Since each vertex $v\in A\setminus S_t$ is the source of at least $\frac\phi2\mathbf d|_{A_{t-1}\setminus A_t^r}(v)+\frac{\epsilon\phi}2\mathbf d(v)$ flow before this step, it is the source of at least $\frac\phi2\mathbf d|_{A_{t-1}\setminus A_t^r}(v)\ge0$ flow after this step. In particular, it cannot become a sink. Also, if $v\in A_t^\ell\setminus S_t\subseteq A_{t-1}\setminus A_t^r$, then it is the source of at least $\frac\phi2\mathbf d(v)$ flow after restriction.

Finally, we remove the edges incident to $t$. Since the flow $f$ is feasible, each edge $(v,t)$ carries at most $c_H(v,t)$ flow. By construction of $H$, we have $c_H(v,t)=12\phi\mathbf d(v)$ for all $v\in A_t^r$, so each vertex $v\in A_t^r$ receives a net flow of at most $12\phi\mathbf d(v)$ from vertices other than $t$ (and then sends that flow to $t$). We may assume that the flow does not send any flow away from $t$ (i.e., along any edge $(v,t)$ in the direction from $t$ to $v$), since otherwise we can remove such flow using a path decomposition. Under this assumption, removing flow on edges incident to $t$ does not create any sources, only sinks. In particular, each vertex $v\in A_t^r$ is now the sink of at most $12\phi\mathbf d(v)$ flow.

Finally, scaling this restricted flow by $1/(12\phi)$, we conclude that the scaled flow $f$ has congestion $1/(12\phi)$, each vertex $v\in A_t^\ell\setminus S_t$ is the source of at least $\mathbf d(v)/24$ flow, and each vertex $v\in A_t^r\setminus S_t$ is the sink of at most $\mathbf d(v)$ flow.
\end{proof}

\begin{claim}\label{clm:compute-ft}
Given flow $f$, the flow $f_t$ can be constructed in $O(m\log m)$ time.
\end{claim}
\begin{proof}
The flow $f_t$ is on the graph $H$ with at most $m+n$ edges, so scaling and restricting the flow takes $O(m)$ time. Using a dynamic tree, we can decompose the new flow into at most $m$ (implicit) paths in $O(m\log m)$ time. The path removal step also takes $O(m\log m)$ time through dynamic trees. The overall running time is $O(m\log m)$.
\end{proof}
\begin{claim}\label{lem:stronger-beta}
Assumption~\ref{item:property-Ast2} of \Cref{thm:fair-cut} holds for our choice of parameters $\gamma,\beta$. That is, we have $\mathbf t(C\cap A)+\epsilon\gamma\cdot \delta_G(C\cap A)\le\beta\cdot \delta_GC$ for all $C\in\mathcal C$.
\end{claim}
\begin{proof}
Recall that we set $\gamma\gets\frac{\epsilon\phi}2,\,\beta\gets (12\phi+\epsilon\gamma)(\kappa+2)$, and $\mathbf t\gets12\phi\mathbf d|_{A_t^r}$. For the entire proof, fix a set $C\in\mathcal C$. We first bound $\mathbf t(C\cap A)$ as
\[ \mathbf t(C\cap A)=12\phi\mathbf d|_{A_t^r}(C\cap A)\le12\phi\mathbf d(C\cap A)=12\phi\deg_{\partial\mathcal P_L\cap\partial A}(C\cap A)\le12\phi(\deg_{\partial\mathcal P_L}(C\cap A)+\deg_{\partial A}(C\cap A)) .\]
For $\deg_{\partial A}(C\cap A)$, we have $\deg_{\partial A}(C\cap A)=c_G(C\cap A,V\setminus A)$ since any edge in $\partial A$ with an endpoint in $C\cap A$ has its other endpoint outside $A$. For $\deg_{\partial\mathcal P_L}(C\cap A)$, we claim the bound $\deg_{\partial\mathcal P_L}(C\cap A)\le\deg_{\partial\mathcal P_L}(C)\le\delta_GC$. The first inequality is trivial, and for the second inequality, observe that by construction of $\mathcal C$, each set $C\in\mathcal C$ is a subset of some cluster in the partition $\mathcal P_L$. It follows that any edge in $\partial_G\mathcal P_L$ with an endpoint in $C$ has its other endpoint outside $C$, so the edge must be in $\partial_GC$, and we conclude that $\deg_{\partial\mathcal P_L}(C)\le\delta_GC$. In total, we have established the bound $\mathbf t(C\cap A)\le12\phi(c_G(C\cap A,V\setminus A)+\delta_GC)$.

Next, we bound $\delta_G(C\cap A)$ as
\[ \delta_G(C\cap A)=c_G(C\cap A,V\setminus A)+c_G(C\cap A,A\setminus C)\le c_G(C\cap A,V\setminus A)+\delta_GC ,\]
and together with the bound on $\mathbf t(C\cap A)$, we obtain
\begin{align}
\mathbf t(C\cap A)+\epsilon\gamma\cdot \delta_G(C\cap A)&\le12\phi(c_G(C\cap A,V\setminus A)+\delta_GC)+\epsilon\gamma(c_G(C\cap A,V\setminus A)+\delta_GC)\nonumber
\\&=(12\phi+\epsilon\gamma)(c_G(C\cap A,V\setminus A)+\delta_GC).\label{eq:stronger-beta}
\end{align}

We now focus on bounding $c_G(C\cap A,V\setminus A)$. Recall from assumption~\ref{item:assumption-CMG} of \Cref{thm:CMG} that there is a flow on $G$ with congestion $\kappa$ such that each vertex $v\in A$ is the source of $c_G(\{v\},V\setminus A)$ flow and each vertex $v\in V$ is the sink of at most $\deg_{\partial\mathcal P_L}(v)$ flow. In particular, among the vertices $v\in C$, there is a total of $c_G(C\cap A,V\setminus A)$ source and at most $\deg_{\partial\mathcal P_L}(C)$ sink. Since the flow has congestion $\kappa$, the net flow out of $C$ is at most $\kappa\delta_GC$. Putting everything together, we obtain
\[ c_G(C\cap A,V\setminus A)\le\deg_{\partial\mathcal P_L}(C)+\kappa\delta_GC\le\delta_GC+\kappa\delta_GC,\]
where the second inequality follows from our earlier claim $\deg_{\partial\mathcal P_L}(C)\le\delta_GC$. Continuing from (\ref{eq:stronger-beta}), we conclude that
\begin{align*}
\mathbf t(C\cap A)+\epsilon\gamma\cdot \delta_G(C\cap A)&\le(12\phi+\epsilon\gamma)(c_G(C\cap A,V\setminus A)+\delta_GC)
\\&=(12\phi+\epsilon\gamma)(\delta_GC+\kappa\delta_GC+\delta_GC)
\\&=(12\phi+\epsilon\gamma)(\kappa+2)\delta_GC
\\&\le\beta\delta_GC,
\end{align*}
finishing the proof.
\end{proof}

\subsubsection{Constructing the matching and updating the flow-matrix}\label{sec:construct-matching}
Using the flow $f_t$, the algorithm then constructs a (fractional) matching graph $M_t$ on vertex set $A_{t-1}$. Take a path decomposition of flow $f_t$ into paths from $A_t^\ell\setminus S_t$ to $A_t^r\setminus S_t$. For each pair of vertices $u\in A_t^\ell\setminus S_t$ and $v\in A_t^r\setminus S_t$, add an edge $(u,v)$ to $M_t$ whose capacity is the sum of capacities of all $u$--$v$ paths in the decomposition of $f_t$. The following claim is immediate using an (implicit) path decomposition.

\begin{claim}\label{clm:compute-Mt}
Given flow $f_t$, the matching graph $M_t$ can be constructed in $O(m\log m)$ time.
\end{claim}

From the matching graph $M_t$, we implicitly update the flow-matrix $\mathbf F_{t-1}$ to $\mathbf F_t$ as follows. For each $v\in A$, we set
\[ \mathbf F_t(u)=\mathbf F_{t-1}(u)+\sum_{v\in A}\frac{c_{M_t}(u,v)}2\bigg(\frac{\mathbf F_{t-1}(v)}{\mathbf d(v)}-\frac{\mathbf F_{t-1}(u)}{\mathbf d(u)}\bigg) .\]
Note that we are viewing $M_t$ as a graph on vertex set $A$, where vertices outside $(A_t^\ell\setminus S_t)\cup(A_t^r\setminus S_t)$ are isolated.

Recall from \Cref{sec:construct-A} that the algorithm needs to compute the projection $\mathbf p(v)=\langle \mathbf F_{t-1}(v)/\mathbf d(v),\,\mathbf r\rangle$ onto the vector $r$ for each vertex $v\in A$. Now that $\mathbf F_{t-1}(v)$ has been explicitly defined, we show that the projections can be computed in total time $O(mT)$.

\begin{claim}\label{clm:compute-projection}
Given vector $\mathbf r\in\mathbb R^A$, the projection $\mathbf p(v)=\langle \mathbf F_{t-1}(v)/\mathbf d(v),\,\mathbf r\rangle$ for each vertex $v\in A$ can be computed in total time $O(mT)$.
\end{claim}
\begin{proof}
By linearity of inner product, we have
\[ \langle \mathbf F_j(u),\mathbf r\rangle=\langle \mathbf F_{j-1}(u),\mathbf r\rangle+\sum_{v\in A}\frac{c_{M_j}(u,v)}2\bigg(\frac{\langle \mathbf F_{j-1}(v),\mathbf r\rangle}{\mathbf d(v)}-\frac{\langle \mathbf F_{j-1}(u),\mathbf r\rangle}{\mathbf d(u)}\bigg) \]
for all vertices $v\in A$ and iterations $j\in[t-1]$. Since the flow $f_j$ is on $G$, we can assume that the path decomposition has at most $m$ paths, so the matching graph $M_j$ has at most $m$ edges. In other words, there are at most $m$ nonzero values of $c_{M_j}(u,v)$. It follows that given the values $\langle \mathbf F_{j-1}(u),\mathbf r\rangle$ for all $u\in A$, we can compute $\langle \mathbf F_j(u),\mathbf r\rangle$ for all $u\in A$ in $O(m)$ time. Since $\mathbf F_0$ is a diagonal matrix, the initial values $\langle \mathbf F_0(u),\mathbf r\rangle$ can be computed in $O(n)$ time. Over $t-1$ iterations, the total time is $O(mT)$.
\end{proof}

\subsection{Analyzing the Potential Decrease}\label{sec:potential-decrease}
The main goal of \Cref{sec:potential-decrease} is to prove the following lemma, establishing a $(1-\Omega(1/\log n))$ expected decrease in $\psi(t)$ on each iteration.
\begin{restatable}{lemma}{PotentialReduction}\label{lem:potential-reduction}
Over the random choice of the unit vector $r$, we have $\mathbb E[\psi(t-1)-\psi(t)]\ge\Omega(\psi(t-1)/\log n)$.
\end{restatable}

To prove \Cref{lem:potential-reduction}, we begin by listing properties of $M_t$ and $\mathbf F_t$.
\begin{claim}\label{clm:degree-bound}
$\deg_{M_t}(u)=\mathbf d(u)/24$ for all $u\in A_t^\ell\setminus S_t$, and $\deg_{M_t}(v)\le\mathbf d(v)$ for all $u\in A_t^r\setminus S_t$.
\end{claim}
\begin{proof}
For each vertex $u\in A_t^\ell\setminus S_t$, its degree in $M_t$ equals the total capacity of paths starting at $u$ in the decomposition of $f_t$, which is exactly $\mathbf d(u)/24$ by \Cref{clm:properties-ft}. A symmetric argument establishes $\deg_{M_t}(v)\le\mathbf d(v)$ for all $u\in A_t^r\setminus S_t$.
\end{proof}
\begin{claim}\label{clm:flow-vector-sum}
For each vertex $u\in A$, the flow-vector $\mathbf F_t(u)$ sums to $\mathbf d(u)$.
\end{claim}
\begin{proof}
We prove by induction on $t$ that the flow-vector $\mathbf F_t(u)$ sums to $\mathbf d(u)$. The statement is true for $t=0$ since $\mathbf F_0$ is defined as the diagonal matrix with value $\mathbf d(v)$ on entry $F(v,v)$. Assume by induction that the flow-vector $\mathbf F_{t-1}(u)$ sums to $\mathbf d(u)$ for each $u\in A$. For each $u\in A$, we take the definition of vector $\mathbf F_t(u)$ and sum over its coordinates $w\in A$ to obtain
\begin{align*}
\sum_{w\in A}\mathbf F_t(u,w)&=\sum_{w\in A}\left( \mathbf F_{t-1}(u,w)+\sum_{v\in A}\frac{c_{M_t}(u,v)}2\bigg(\frac{\mathbf F_{t-1}(v,w)}{\mathbf d(v)}-\frac{\mathbf F_{t-1}(u,w)}{\mathbf d(u)}\bigg) \right)
\\&=\sum_{w\in A}\mathbf F_{t-1}(u,w)+\sum_{v\in A}\frac{c_{M_t}(u,v)}2\bigg(\frac{\sum_{w\in A}\mathbf F_{t-1}(v,w)}{\mathbf d(v)}-\frac{\sum_{w\in A}\mathbf F_{t-1}(u,w)}{\mathbf d(u)}\bigg)
\\&=\mathbf d(u)+\sum_{v\in A}\frac{c_{M_t}(u,v)}2\bigg(\frac{\mathbf d(v)}{\mathbf d(v)}-\frac{\mathbf d(u)}{\mathbf d(u)}\bigg)
\\&=\mathbf d(u),
\end{align*}
completing the induction and the proof.
\end{proof}
\begin{claim}
If $\mathbf F_{t-1}$ is routable with congestion $t/\phi$, then $\mathbf F_t$ is routable with congestion $(t+1)/\phi$.
\end{claim}
\begin{proof}
Take the $A$-commodity flow routing $\mathbf F_{t-1}$ with congestion $t/\phi$ and reverse it, forming a $A$-commodity flow routing the transpose $\mathbf F_{t-1}^\top$ with the same congestion. In this reversed flow, each vertex $u\in A$ has $\mathbf F_{t-1}^\top(v,u)=\mathbf F_{t-1}(u,v)$ quantity of commodity $v$. In other words, the quantities of each commodity at vertex $u\in A$ is captured by its flow vector $\mathbf F_{t-1}(u)$.

Next, we ``mix'' commodities along the edges of the matching graph $M_t$: for each edge $(u,v)$ in $M_t$ of capacity $c$, send a proportional $\frac c{2\mathbf d(u)}$ fraction of each commodity at $u$ along the corresponding path (in the path decomposition of $f/\phi$) in the direction from $u$ to $v$, and send a proportional $\frac c{2\mathbf d(v)}$ fraction of each commodity at $v$ in the direction from $v$ to $u$. By \Cref{clm:flow-vector-sum}, each vertex $u\in A_{t-1}$ has $\mathbf d(u)$ total quantity of commodities, so we send $c/2$ total quantity of commodities in each direction, or $c$ in total, along this path of capacity $c$ (in the path decomposition of $f_t$). Since flow $f_t$ has congestion $1/\phi$ by \Cref{clm:properties-ft}, the total congestion of this mixing step is $1/\phi$. After this mixing step, the quantities of each commodity at vertex $u\in A$ is exactly $\mathbf F_t(u)$. In other words, we have established a $A$-commodity flow routing $\mathbf F_t^\top$ with congestion $(t+1)/\phi$. Reversing the flow, we obtain a routing of $\mathbf F_t$ with the same congestion.
\end{proof}

We now analyze the potential decrease in the following technical lemma.
\begin{lemma}\label{lem:potential-reduction-before-projection}
\[ \psi(t)-\psi(t-1)\ge\displaystyle\frac12\sum_{u,v\in A}c_{M_t}(u,v)\left\lVert\frac{\mathbf F_t(u)}{\mathbf d(u)}-\frac{\mathbf F_t(v)}{\mathbf d(v)}\right\rVert_2^2+\sum_{u\in A_{t-1}\setminus A_t}\mathbf d(u)\left\lVert\frac{\mathbf F_{t-1}(u)}{\mathbf d(u)}-\boldsymbol\mu_{t-1}\right\rVert_2^2 .\]
\end{lemma}
\begin{proof}
For each vertex $v\in A$, define the normalized flow vector $\widetilde{\mathbf F}_{t-1}(v)=\mathbf d(v)^{-1/2}\mathbf F_{t-1}(v)$, and define the normalized flow matrix $\widetilde{\mathbf F}_{t-1}\in\mathbb R^{A\times A}_{\ge0}$ with vector $\widetilde{\mathbf F}_{t-1}(v)$ for each row $v\in A$. Let $D$ be the diagonal matrix with value $\mathbf d(u)$ on entry $(u,u)$. We can then write $\widetilde{\mathbf F}_{t-1}=D^{-1/2} \mathbf F_{t-1}$. Define the vectors $\widetilde{\mathbf F}_t(v)=\mathbf d(v)^{-1/2}\mathbf F_t(v)$ and matrix $\widetilde{\mathbf F}_t=D^{-1/2} \mathbf F_t$ analogously. For each vertex $u\in A$, by definition of flow vector $\mathbf F_t(u)$, we have
\begin{align*}
\widetilde{\mathbf F}_t(u)=\mathbf d(u)^{-1/2}\mathbf F_t(u)&=\mathbf d(u)^{-1/2}\bigg(\mathbf F_{t-1}(u)+\sum_{v\in A_t}\frac{c_{M_t}(u,v)}2\bigg(\frac{\mathbf F_{t-1}(v)}{\mathbf d(v)}-\frac{\mathbf F_{t-1}(u)}{\mathbf d(u)}\bigg) \bigg)
\\&=\mathbf d(u)^{-1/2}\bigg(\mathbf d(u)^{1/2}\widetilde{\mathbf F}_{t-1}(u)+\sum_{v\in A_{t-1}}\frac{c_{M_t}(u,v)}2\bigg(\frac{\widetilde{\mathbf F}_{t-1}(v)}{\mathbf d(v)^{1/2}}-\frac{\widetilde{\mathbf F}_{t-1}(u)}{\mathbf d(u)^{1/2}}\bigg) \bigg)
\\&=\widetilde F(u)+\sum_{v\in A_{t-1}}\frac{c_{M_t}(u,v)}2\bigg(\frac{\widetilde{\mathbf F}_{t-1}(v)}{\mathbf d(u)^{1/2}\mathbf d(v)^{1/2}}-\frac{\widetilde{\mathbf F}_{t-1}(u)}{\mathbf d(u)}\bigg)
\\&=\widetilde F(u)-\frac12\bigg(\frac{\deg_{M_t}(u)}{\mathbf d(u)}\widetilde F(u)-\sum_{v\in A_{t-1}}\frac{c_{M_t}(u,v)}{\mathbf d(u)^{1/2}\mathbf d(v)^{1/2}}\widetilde F(v)\bigg).
\end{align*}

Let $L\in\mathbb R^{A\times A}_{\ge0}$ be the Laplacian matrix for the matching graph $M_t$ on vertex set $A$ (where vertices outside $A_t^\ell\cup A_t^r$ are isolated). That is, we define $L(u,u)=\deg_{M_t}(u)$ for all $u\in A$ and $L(u,v)=-c_{M_t}(u,v)$ for all distinct $u,v\in A$.

We first prove the following about the matrix $D^{-1/2}LD^{-1/2}$.
\begin{subclaim}\label{clm:laplacian}
$\mathbf0\preceq D^{-1/2}LD^{-1/2}\preceq2I$.
\end{subclaim}
\begin{subproof}
Consider the normalized Laplacian $D_L^{-1/2}LD^{-1/2}_L$ where $D_L$ is the diagonal matrix with value $\deg_{M_t}(u)$ on each entry $(u,u)$. It is well-known that the normalized Laplacian has eigenvalues in the range $[0,2]$, so $\mathbf0\preceq D_L^{-1/2}LD_L^{-1/2}\preceq2I$. Multiplying by $D_L^{1/2}$ on both sides gives $\mathbf0\preceq L\preceq2D_L$. \Cref{clm:degree-bound} implies that $D_L\preceq D$, so we obtain $\mathbf0\preceq L\preceq2D_L\preceq2D$. Finally, multiplying by $D^{-1/2}$ on both sides gives the desired $\mathbf0\preceq D^{-1/2}LD^{-1/2}\preceq2I$.
\end{subproof}

We also need the following fact from linear algebra, whose routine proof we defer to the end of this subsection.
\begin{restatable}{fact}{Trace}\label{lem:linear-algebra-fact}
For any symmetric matrices $A,B\in\mathbb R^{n\times n}$ satisfying $A\succeq\mathbf0$ and $-I\preceq B\preceq I$, we have $\textup{Tr}(B^\top AB)\le\textup{Tr}(A)$.
\end{restatable}

By definition, the matrix $D^{-1/2}LD^{-1/2}$ has value $\frac{\deg_{M_t}(u)}{\mathbf d(u)}$ on entry $(u,u)$ and value $-\frac{c_{M_t}(u,v)}{\mathbf d(u)^{1/2}\mathbf d(v)^{1/2}}$ on entry $(u,v)$ with $u\ne v$. It follows that we can write matrix $\widetilde{\mathbf F}_t$ as
\[ \widetilde{\mathbf F}_t=\bigg(I-\frac12D^{-1/2}LD^{-1/2}\bigg)\widetilde{\mathbf F}_{t-1}=\frac12(I+N)\widetilde{\mathbf F}_{t-1}\quad\text{where}\quad N=I-D^{-1/2}LD^{-1/2}. \]

Let $J\in\mathbb R^{A\times A}$ denote the all-ones matrix, and for two matrices $A,B\in\mathbb R^{A\times A}$, define the Hadamard product $A\bullet B=\sum_{u,v\in A}A(u,v)\cdot B(u,v)=\textup{Tr}(A^\top B)$. We first write $\psi(t-1)$ as
\begin{align*}
\psi(t-1)&=\sum_{u\in A_{t-1}}\mathbf d(u)\left\lVert\frac{\mathbf F_{t-1}(u)}{\mathbf d(u)}-\boldsymbol\mu_{t-1}\right\rVert_2^2
\\&=\sum_{u\in A}\mathbf d(u)\left\lVert\frac{\mathbf F_{t-1}(u)}{\mathbf d(u)}-\boldsymbol\mu_{t-1}\right\rVert_2^2-\sum_{u\in A\setminus A_{t-1}}\mathbf d(u)\left\lVert\frac{\mathbf F_{t-1}(u)}{\mathbf d(u)}-\boldsymbol\mu_{t-1}\right\rVert_2^2
\\&=\sum_{u\in A}\left\lVert\frac{\mathbf F_{t-1}(u)}{\mathbf d(u)^{1/2}}-\mathbf d(u)^{1/2}\boldsymbol\mu_{t-1}\right\rVert_2^2-\sum_{u\in A\setminus A_{t-1}}\mathbf d(u)\left\lVert\frac{\mathbf F_{t-1}(u)}{\mathbf d(u)}-\boldsymbol\mu_{t-1}\right\rVert_2^2
\\&=(\widetilde{\mathbf F}_{t-1}-D^{1/2}J\boldsymbol\mu_{t-1})\bullet(\widetilde{\mathbf F}_{t-1}-D^{1/2}J\boldsymbol\mu_{t-1})-\sum_{u\in A\setminus A_{t-1}}\mathbf d(u)\left\lVert\frac{\mathbf F_{t-1}(u)}{\mathbf d(u)}-\boldsymbol\mu_{t-1}\right\rVert_2^2 .
\end{align*}
To bound $\psi(t)$, we use the fact that $\boldsymbol\mu_t$ minimizes the expression in (\ref{eq:mu-t}) to get
\begin{align*}
\psi(t)&=\sum_{u\in A_t}\mathbf d(u)\left\lVert\frac{\mathbf F_t(u)}{\mathbf d(u)}-\boldsymbol\mu_t\right\rVert_2^2
\\&\le\sum_{u\in A_t}\mathbf d(u)\left\lVert\frac{\mathbf F_t(u)}{\mathbf d(u)}-\boldsymbol\mu_{t-1}\right\rVert_2^2
\\&=\sum_{u\in A}\mathbf d(u)\left\lVert\frac{\mathbf F_t(u)}{\mathbf d(u)}-\boldsymbol\mu_{t-1}\right\rVert_2^2-\sum_{u\in A\setminus A_t}\mathbf d(u)\left\lVert\frac{\mathbf F_t(u)}{\mathbf d(u)}-\boldsymbol\mu_{t-1}\right\rVert_2^2
\\&=\sum_{u\in A}\left\lVert\frac{\mathbf F_t(u)}{\mathbf d(u)^{1/2}}-\mathbf d(u)^{1/2}\boldsymbol\mu_{t-1}\right\rVert_2^2-\sum_{u\in A\setminus A_t}\mathbf d(u)\left\lVert\frac{\mathbf F_t(u)}{\mathbf d(u)}-\boldsymbol\mu_{t-1}\right\rVert_2^2
\\&=(\widetilde{\mathbf F}_t-D^{1/2}J\boldsymbol\mu_{t-1})\bullet(\widetilde{\mathbf F}_t-D^{1/2}J\boldsymbol\mu_{t-1})-\sum_{u\in A\setminus A_t}\mathbf d(u)\left\lVert\frac{\mathbf F_t(u)}{\mathbf d(u)}-\boldsymbol\mu_{t-1}\right\rVert_2^2 .
\end{align*}
Taking the difference and using the fact that $\mathbf F_{t-1}(u)=\mathbf F_t(u)$ for $u\in A\setminus A_t$, and then expanding,
\begin{align*}
&\psi(t-1)-\psi(t)-\sum_{u\in A_{t-1}\setminus A_t}\mathbf d(u)\left\lVert\frac{\mathbf F_{t-1}(u)}{\mathbf d(u)}-\boldsymbol\mu_{t-1}\right\rVert_2^2
\\&\ge(\widetilde{\mathbf F}_{t-1}-D^{1/2}J\boldsymbol\mu_{t-1})\bullet(\widetilde{\mathbf F}_{t-1}-D^{1/2}J\boldsymbol\mu_{t-1})-(\widetilde{\mathbf F}_t-D^{1/2}J\boldsymbol\mu_{t-1})\bullet(\widetilde{\mathbf F}_t-D^{1/2}J\boldsymbol\mu_{t-1})
\\&=(\widetilde{\mathbf F}_{t-1}-D^{1/2}J\boldsymbol\mu_{t-1})\bullet(\widetilde{\mathbf F}_{t-1}-D^{1/2}J\boldsymbol\mu_{t-1})-(\widetilde{\mathbf F}_t-D^{1/2}J\boldsymbol\mu_{t-1})\bullet(\widetilde{\mathbf F}_t-D^{1/2}J\boldsymbol\mu_{t-1})
\\&=\widetilde{\mathbf F}_{t-1}\bullet\widetilde{\mathbf F}_{t-1}-2\widetilde{\mathbf F}_{t-1}\bullet D^{1/2}J\boldsymbol\mu_{t-1}+D^{1/2}J\boldsymbol\mu_{t-1}\bullet D^{1/2}J\boldsymbol\mu_{t-1}
\\&\qquad\quad\, -\widetilde{\mathbf F}_t\bullet\widetilde{\mathbf F}_t+2\widetilde{\mathbf F}_t\bullet D^{1/2}J\boldsymbol\mu_{t-1}-D^{1/2}J\boldsymbol\mu_{t-1}\bullet D^{1/2}J\boldsymbol\mu_{t-1}
\\&=\widetilde{\mathbf F}_{t-1}\bullet\widetilde{\mathbf F}_{t-1}-\widetilde{\mathbf F}_t\bullet\widetilde{\mathbf F}_t+2(\widetilde{\mathbf F}_t-\widetilde{\mathbf F}_{t-1})\bullet D^{1/2}J\boldsymbol\mu_{t-1}
\\&=\widetilde{\mathbf F}_{t-1}\bullet\widetilde{\mathbf F}_{t-1}-\widetilde{\mathbf F}_t\bullet\widetilde{\mathbf F}_t-2\bigg(\frac12D^{-1/2}LD^{-1/2}\bigg)\bullet D^{1/2}J\boldsymbol\mu_{t-1}.
\end{align*}
The third term equals $\textup{Tr}\big(D^{-1/2}LD^{-1/2}D^{1/2}J\boldsymbol\mu_{t-1}\big)$, which equals $0$ since $LD^{-1/2}D^{1/2}J=LJ$ is the zero matrix. Expanding the first and second terms, we obtain
\begin{align}
&\psi(t-1)-\psi(t)-\sum_{u\in A_{t-1}\setminus A_t}\mathbf d(u)\left\lVert\frac{\mathbf F_{t-1}(u)}{\mathbf d(u)}-\boldsymbol\mu_{t-1}\right\rVert_2\nonumber
\\&\ge\widetilde{\mathbf F}_{t-1}\bullet\widetilde{\mathbf F}_{t-1}-\widetilde{\mathbf F}_t\bullet\widetilde{\mathbf F}_t\nonumber
\\&=\widetilde{\mathbf F}_{t-1}\bullet\widetilde{\mathbf F}_{t-1}-\frac12(I+N)\widetilde{\mathbf F}_{t-1}\bullet\frac12(I+N)\widetilde{\mathbf F}_{t-1}\nonumber
\\&=\frac34\widetilde{\mathbf F}_{t-1}\bullet\widetilde{\mathbf F}_{t-1}-\frac12N\widetilde{\mathbf F}_{t-1}\bullet\widetilde{\mathbf F}_{t-1}-\frac14N\widetilde{\mathbf F}_{t-1}\bullet N\widetilde{\mathbf F}_{t-1}\nonumber
\\&=\frac14\big(\widetilde{\mathbf F}_{t-1}\bullet\widetilde{\mathbf F}_{t-1}-N\widetilde{\mathbf F}_{t-1}\bullet N\widetilde{\mathbf F}_{t-1}\big)+\frac12(I-N)\widetilde{\mathbf F}_{t-1}\bullet\widetilde{\mathbf F}_{t-1}\nonumber
\\&=\frac14\big(\widetilde{\mathbf F}_{t-1}\bullet\widetilde{\mathbf F}_{t-1}-N\widetilde{\mathbf F}_{t-1}\bullet N\widetilde{\mathbf F}_{t-1}\big)+\frac12D^{-1/2}LD^{-1/2}\widetilde{\mathbf F}_{t-1}\bullet\widetilde{\mathbf F}_{t-1}.\label{eq:phi-expand}
\end{align}
To bound the first term in (\ref{eq:phi-expand}), recall from \Cref{clm:laplacian} that $\mathbf0\preceq D^{-1/2}LD^{-1/2}\preceq 2I$, which means that $-I\preceq N\preceq I$. Since $\widetilde{\mathbf F}_{t-1}\widetilde{\mathbf F}_{t-1}^\top\succeq\mathbf0$ and $-I\preceq N\preceq I$, we have $\textup{Tr}(N\widetilde{\mathbf F}_{t-1}\widetilde{\mathbf F}_{t-1}^\top N)\le\textup{Tr}(\widetilde{\mathbf F}_{t-1}\widetilde{\mathbf F}_{t-1}^\top)$ by \Cref{lem:linear-algebra-fact}, so
\[ \widetilde{\mathbf F}_{t-1}\bullet\widetilde{\mathbf F}_{t-1}=\textup{Tr}(\widetilde{\mathbf F}_{t-1}\widetilde{\mathbf F}_{t-1}^\top)\ge\textup{Tr}(N\widetilde{\mathbf F}_{t-1}\widetilde{\mathbf F}_{t-1}^\top N)=\textup{Tr}(\widetilde{\mathbf F}_{t-1}^\top NN\widetilde{\mathbf F}_{t-1})=N\widetilde{\mathbf F}_{t-1}\bullet N\widetilde{\mathbf F}_{t-1} ,\]
so the first term $\frac14(\widetilde{\mathbf F}_{t-1}\bullet\widetilde{\mathbf F}_{t-1}-N\widetilde{\mathbf F}_{t-1}\bullet N\widetilde{\mathbf F}_{t-1})$ is at least $0$. 

Continuing from (\ref{eq:phi-expand}), we conclude that
\[ \psi(t-1)-\psi(t)-\sum_{u\in A_{t-1}\setminus A_t}\mathbf d(u)\left\lVert\frac{\mathbf F_{t-1}(u)}{\mathbf d(u)}-\boldsymbol\mu_{t-1}\right\rVert_2\ge\frac12D^{-1/2}LD^{-1/2}\widetilde{\mathbf F}_{t-1}\bullet\widetilde{\mathbf F}_{t-1} .\]

It remains to understand the above expression $D^{-1/2}LD^{-1/2}\widetilde{\mathbf F}_{t-1}\bullet\widetilde{\mathbf F}_{t-1}$. For each vertex $v\in A$, let $\mathbbm 1_v\in\mathbb R^A$ denote the unit vector in direction $v$. The Laplacian can be written as
\[ L=\sum_{u,v\in A}c_{M_t}(u,v)\cdot(\mathbbm 1_u-\mathbbm 1_v)(\mathbbm 1_u-\mathbbm 1_v)^\top ,\]
so that
\begin{align*}
D^{-1/2}LD^{-1/2}\widetilde{\mathbf F}_{t-1}\bullet\widetilde{\mathbf F}_{t-1}&=\sum_{u,v\in A}c_{M_t}(u,v)\cdot D^{-1/2}(\mathbbm 1_u-\mathbbm 1_v)(\mathbbm 1_u-\mathbbm 1_v)^\top D^{-1/2}\widetilde{\mathbf F}_{t-1}\bullet\widetilde{\mathbf F}_{t-1}
\\&=\sum_{u,v\in A}c_{M_t}(u,v)\cdot\textup{Tr}\big(\widetilde{\mathbf F}_{t-1}^\top D^{-1/2}(\mathbbm 1_u-\mathbbm 1_v)(\mathbbm 1_u-\mathbbm 1_v)^\top D^{-1/2}\widetilde{\mathbf F}_{t-1}\big)
\\&=\sum_{u,v\in A}c_{M_t}(u,v)\cdot\textup{Tr}\big((\mathbbm 1_u-\mathbbm 1_v)^\top D^{-1/2}\widetilde{\mathbf F}_{t-1}\widetilde{\mathbf F}_{t-1}^\top D^{-1/2}(\mathbbm 1_u-\mathbbm 1_v)\big)
\\&=\sum_{u,v\in A}c_{M_t}(u,v)\cdot(\mathbbm 1_u-\mathbbm 1_v)^\top D^{-1/2}\widetilde{\mathbf F}_{t-1}\widetilde{\mathbf F}_{t-1}^\top D^{-1/2}(\mathbbm 1_u-\mathbbm 1_v)
\\&=\sum_{u,v\in A}c_{M_t}(u,v)\cdot\left\lVert\widetilde{\mathbf F}_{t-1}^\top D^{-1/2}(\mathbbm 1_u-\mathbbm 1_v)\right\rVert_2^2
\\&=\sum_{u,v\in A}c_{M_t}(u,v)\cdot\left\lVert\mathbf F_{t-1}^\top D^{-1}(\mathbbm 1_u-\mathbbm 1_v)\right\rVert_2^2
\\&=\sum_{u,v\in A}c_{M_t}(u,v)\cdot\left\lVert\frac{\mathbf F_{t-1}(u)}{\mathbf d(u)}-\frac{\mathbf F_{t-1}(v)}{\mathbf d(v)}\right\rVert_2^2 .
\end{align*}
We conclude that
\[ \psi(t-1)-\psi(t)-\sum_{u\in A_{t-1}\setminus A_t}\mathbf d(u)\left\lVert\frac{\mathbf F_{t-1}(u)}{\mathbf d(u)}-\boldsymbol\mu_{t-1}\right\rVert_2\ge\frac12\sum_{u,v\in A}c_{M_t}(u,v)\cdot\left\lVert\frac{\mathbf F_{t-1}(u)}{\mathbf d(u)}-\frac{\mathbf F_{t-1}(v)}{\mathbf d(v)}\right\rVert_2^2.\qedhere\]
\end{proof}

Finally, we prove the main lemma of \Cref{sec:construct-matching}, restated below.
\PotentialReduction*
\begin{proof}
By \Cref{lem:potential-reduction-before-projection}, we have
\begin{gather}
\psi(t-1)-\psi(t)\ge\displaystyle\frac12\sum_{u,v\in A}c_{M_t}(u,v)\left\lVert\frac{\mathbf F_{t-1}(u)}{\mathbf d(u)}-\frac{\mathbf F_{t-1}(v)}{\mathbf d(v)}\right\rVert_2^2+\sum_{u\in A_{t-1}\setminus A_t}\mathbf d(u)\left\lVert\frac{\mathbf F_{t-1}(u)}{\mathbf d(u)}-\boldsymbol\mu_{t-1}\right\rVert_2^2 .\label{eq:pot-red-1}
\end{gather}
By \Cref{lem:KRV-CMG}, we have
\begin{gather}
\mathbb E[(\mathbf p(v)-\bar\mu_{t-1})^2]=\frac1n\left\lVert\frac{\mathbf F_{t-1}(v)}{\mathbf d(v)}-\boldsymbol\mu_{t-1}\right\rVert_2^2\label{eq:pot-red-2}
\end{gather}
for all vertices $v\in A$, and for some constant $C>0$,
\begin{gather}
(\mathbf p(u)-\mathbf p(v))^2\le\frac{C\log n}{n}\left\lVert\frac{\mathbf F_{t-1}(u)}{\mathbf d(u)}-\frac{\mathbf F_{t-1}(v)}{\mathbf d(v)}\right\rVert_2^2 ,\quad (\mathbf p(u)-\bar\mu_{t-1})^2\le\frac{C\log n}{n}\left\lVert\frac{\mathbf F_{t-1}(u)}{\mathbf d(u)}-\boldsymbol\mu_{t-1}\right\rVert_2^2\label{eq:pot-red-3}
\end{gather}
for all vertices $u,v\in A$ with high probability. By \Cref{clm:degree-bound}, we have $\deg_{M_t}(u)=\mathbf d(u)/24$ for all $u\in A_t^\ell\setminus S_t$, and together with statements~\ref{item:CMG-separator-1} to~\ref{item:CMG-separator-4} of \Cref{lem:RST-CMG}, we have
\begin{align}
\sum_{\substack{u\in A_t^\ell\setminus S_t\\v\in A_t^r\setminus S_t}}c_{M_t}(u,v)\cdot(\mathbf p(u)-\mathbf p(v))^2&\stackrel{\mathclap{\textup{\ref{item:CMG-separator-1}}}}\ge\sum_{\substack{u\in A_t^\ell\setminus S_t\\v\in A_t^r\setminus S_t}}c_{M_t}(u,v)\cdot(\mathbf p(u)-\eta)^2\nonumber
\\&=\sum_{u\in A_t^\ell\setminus S_t}\deg_{M_t}(u)\cdot(\mathbf p(u)-\eta)^2\nonumber
\\&=\frac1{24}\sum_{u\in A_t^\ell\setminus S_t}\mathbf d(u)\cdot(\mathbf p(u)-\eta)^2\nonumber
\\&\stackrel{\mathclap{\textup{\ref{item:CMG-separator-3}}}}\ge\frac1{216}\sum_{u\in A_t^\ell\setminus S_t}\mathbf d(u)\cdot(\mathbf p(u)-\bar\mu_{t-1})^2.\label{eq:pot-red-4}
\end{align}
Taking the expectation and putting everything together,
\begin{align*}
&\mathbb E[\psi(t-1)-\psi(t)]
\\&\stackrel{\mathclap{(\ref{eq:pot-red-1})}}\ge\mathbb E\bigg[\displaystyle\frac12\sum_{u,v\in A}c_{M_t}(u,v)\left\lVert\frac{\mathbf F_{t-1}(u)}{\mathbf d(u)}-\frac{\mathbf F_{t-1}(v)}{\mathbf d(v)}\right\rVert_2^2+\sum_{u\in A_{t-1}\setminus A_t}\mathbf d(u)\left\lVert\frac{\mathbf F_{t-1}(u)}{\mathbf d(u)}-\boldsymbol\mu_{t-1}\right\rVert_2^2\bigg]
\\&=\mathbb E\bigg[\displaystyle\sum_{\substack{u\in A_t^\ell\setminus S_t\\v\in A_t^r\setminus S_t}}c_{M_t}(u,v)\left\lVert\frac{\mathbf F_{t-1}(u)}{\mathbf d(u)}-\frac{\mathbf F_{t-1}(v)}{\mathbf d(v)}\right\rVert_2^2+\sum_{u\in S_t\cap A_{t-1}}\mathbf d(u)\left\lVert\frac{\mathbf F_{t-1}(u)}{\mathbf d(u)}-\boldsymbol\mu_{t-1}\right\rVert_2^2\bigg]
\\&\stackrel{\mathclap{(\ref{eq:pot-red-3})}}\ge\mathbb E\bigg[\frac n{C\log n}\sum_{\substack{u\in A_t^\ell\setminus S_t\\v\in A_t^r\setminus S_t}}c_{M_t}(u,v)\cdot(\mathbf p(u)-\mathbf p(v))^2+\frac n{C\log n}\sum_{u\in S_t\cap A_{t-1}}\mathbf d(u)\cdot(\mathbf p(u)-\bar\mu_{t-1})^2\bigg]
\\&\stackrel{\mathclap{(\ref{eq:pot-red-4})}}\ge\mathbb E\bigg[\frac n{216C\log n}\sum_{u\in A_t^\ell\setminus S_t}\mathbf d(u)\cdot(\mathbf p(u)-\bar\mu_{t-1})^2+\frac n{C\log n}\sum_{u\in S_t\cap A_{t-1}}\mathbf d(u)\cdot(\mathbf p(u)-\bar\mu_{t-1})^2\bigg]
\\&\ge\mathbb E\bigg[\frac n{216C\log n}\sum_{u\in A_t^\ell}\mathbf d(u)\cdot(\mathbf p(u)-\bar\mu_{t-1})^2\bigg]
\\&\stackrel{\mathclap{\textup{\ref{item:CMG-separator-4}}}}\ge\mathbb E\bigg[\frac n{80\cdot216C\log n}\sum_{u\in A_{t-1}}\mathbf d(u)\cdot(\mathbf p(u)-\bar\mu_{t-1})^2\bigg]
\\&=\frac n{80\cdot216C\log n}\sum_{u\in A_{t-1}}\mathbf d(u)\cdot\mathbb E[(\mathbf p(u)-\bar\mu_{t-1})^2]
\\&\stackrel{\mathclap{(\ref{eq:pot-red-2})}}=\frac n{80\cdot216C\log n}\sum_{u\in A_{t-1}}\mathbf d(u)\cdot\frac1n\left\lVert\frac{\mathbf F_{t-1}(u)}{\mathbf d(u)}-\boldsymbol\mu_{t-1}\right\rVert_2^2
\\&=\frac1{80\cdot216C\log n}\cdot\psi(t-1),
\end{align*}
which concludes the proof of \Cref{lem:potential-reduction}.
\end{proof}

We conclude this subsection with the proof of \Cref{lem:linear-algebra-fact}, restated below.
\Trace*
\begin{proof}
Let $\lambda_k(X)$ be the $k$th largest eigenvalue of matrix $X$. Since trace is the sum of eigenvalues, it suffices to prove that $\lambda_k(B^\top AB)\le\lambda_k(A)$ for all $k\in[n]$.

By the Courant-Fischer Theorem, $\lambda_k(B^\top AB)$ is the maximum over any $k$-dimensional subspace $U\subseteq\mathbb R^n$ of the quantity
\[ \min_{\mathbf x\in U,\mathbf x\ne\mathbf0}\frac{\mathbf x^\top B^\top AB\mathbf x}{\mathbf x^\top \mathbf x} .\]
Let $U^*$ be the $k$-dimensional subspace attaining this maximum. If there exists $\mathbf x\in U^*$ with $\mathbf x\ne\mathbf0$ and $B\mathbf x=\mathbf0$, then the expression above is $0$, so we obtain the desired $\lambda_k(B^\top AB)=0\le\lambda_k(A)$ since $A\succeq\mathbf0$. Otherwise, let $BU^*=\{B\mathbf x:\mathbf x\in U^*\}$, which has dimension exactly $k$.

Fix a nonzero vector $\mathbf y\in BU^*$. Then, there exists a nonzero vector $\mathbf x\in U^*$ with $\mathbf y=B\mathbf x$. Since $-I\preceq B\preceq I$, we have $B^2\preceq I$ and $\mathbf y^\top\mathbf y=\mathbf x^\top B^2\mathbf x\le \mathbf x^\top \mathbf x$. It follows that
\[ \frac{\mathbf x^\top B^\top AB\mathbf x}{\mathbf x^\top \mathbf x}\le\frac{\mathbf x^\top B^\top AB\mathbf x}{\mathbf y^\top \mathbf y}=\frac{\mathbf y^\top A\mathbf y}{\mathbf y^\top \mathbf y} .\]
Taking the minimum over all nonzero $\mathbf y\in BU^*$,
\[ \min_{\mathbf x\in U^*,\mathbf x\ne\mathbf0}\frac{\mathbf x^\top B^\top AB\mathbf x}{\mathbf x^\top \mathbf x}\le\min_{\mathbf y\in BU^*,\mathbf y\ne\mathbf0}\frac{\mathbf y^\top A\mathbf y}{\mathbf y^\top \mathbf y} .\]
The expression on the left is exactly $\lambda_k(B^\top AB)$ by definition of $U^*$, and the expression on the right is at most $\lambda_k(A)$ since $BU^*$ is a $k$-dimensional subspace and the Courant-Fischer Theorem takes the maximum. We conclude that $\lambda_k(B^\top AB)\le\lambda_k(A)$, finishing the proof.
\end{proof}

\subsection{Putting Everything Together}
Finally, we prove correctness of the algorithm by establishing properties~(\ref{item:CMG-property-1}) to~(\ref{item:CMG-property-3}) in \Cref{thm:CMG}. Properties~(\ref{item:CMG-property-1}) and~(\ref{item:CMG-property-2}) follow directly from \Cref{lem:set-R}. For property~{(\ref{item:CMG-property-3}), if the algorithm terminates early, then $\mathbf d(R)\ge\mathbf d(A)/(6 T)$ by the termination condition and property~{(\ref{item:CMG-property-3}) is satisfied. Otherwise, we have $\mathbf d(R_T)<\mathbf d(A)/(6 T)$. By \Cref{lem:potential-reduction}, the potential $\psi(t)$ for any iteration $t$ is at most $1-\Omega(1/\log n)$ times the previous potential $\psi(t-1)$ in expectation. Over $T=O(\log^2(nW))$ iterations, the expected potential $\psi(T)$ is at most $1/(nW)^{2C}$, where the constant $C>0$ can be made arbitrarily large by setting the constant in $T=O(\log^2(nW))$ appropriately. By Markov's inequality, we have $\psi(T)\le1/(nW)^C$ with probability at least $1-1/n^C$, in which case \Cref{lem:CMG-mix} implies that vertex weighting $\mathrm{deg}_F|_{A_t}$ mixes in $G$ with congestion $5 T/\phi=O(\phi^{-1}\log^2(nW))$, fulfilling property~{(\ref{item:CMG-property-3}).

It remains to bound the running time. For each iteration, the projections $\mathbf p(v)=\langle \mathbf F_{t-1}(v)/\mathbf d(v),\,\mathbf r\rangle$ for each vertex $v\in A$ can be computed in total time $O(mT)$ by \Cref{clm:compute-projection}.
Computing the values $A_t^\ell,A_t^r$ takes additional time $O(|A_{t-1}|\log|A_{t-1}|)=O(m\log m)$ by \Cref{lem:RST-CMG}. The algorithm makes one call to \Cref{thm:fair-cut} with the parameters in~(\ref{eq:parameters}), and then computes flow $f_t$ and the matching graph $M_t$ in $O(m\log m)$ time by \Cref{clm:compute-ft,clm:compute-Mt}. Overall, the algorithm runs in $O(m\log^2(nW))$ time per iteration for $T=O(\log^2(nW))$ many iterations, which is $O(m\log^4(nW))$ time in total. The algorithm also makes $T=O(\log^2(nW))$ many calls to \Cref{thm:fair-cut}.

With both properties~(\ref{item:CMG-property-1}) to~(\ref{item:CMG-property-3}) and the running time established, this concludes the proof of \Cref{thm:CMG}.

\section{Expander Trimming}\label{sec:trimming}

In this section, we prove \Cref{thm:trimming}, restated below.
\Trimming*

Our setup closely resembles Section~3.2 of~\cite{SW19}, except that we replace the exact min-cut/max-flow oracle in their ``slow-trimming'' procedure with an approximate one. We also avoid defining expanders and work exclusively with flow mixing, which simplifies the analysis.

Define the vertex weighting $\mathbf s_0\in\mathbb R^A_{\ge0}$ as $\mathbf s_0(v)=c_G(\{v\},R)$ for all $v\in A\setminus R$, and $\mathbf s_0(v)=0$ for all $v\in R$.
We call \Cref{thm:fair-cut} with parameters
\begin{gather*} A\gets A,\,\epsilon\gets\epsilon,\,\gamma\gets\frac{\epsilon\phi}2,\,\beta\gets\max\{1,(12\phi+\epsilon\gamma)(\kappa+2)\},\,\,\mathbf s\gets\mathbf s_0+\epsilon\phi\mathbf d,\,\mathbf t\gets12\phi\mathbf d|_{A\setminus R} ,\end{gather*}
and we denote the graph $G[A,\gamma,\mathbf s,\mathbf t]$ by $H=(V_H,E_H)$. Note that except for $\epsilon$ and $\mathbf s$, these parameters match the parameters~(\ref{eq:parameters}) in the proof of \Cref{thm:CMG}, where $A_t^\ell$ and $A_t^r$ are replaced by $R$ and $A\setminus R$, respectively. Since assumption~\ref{item:assumption-trimming} is the same as the one in \Cref{thm:CMG}, the statement and proof of \Cref{lem:stronger-beta} (which does not depend on $\epsilon$ or $\mathbf s$) translates directly to this setting. For brevity, we omit the details, and simply restate \Cref{lem:stronger-beta} below.
\begin{claim}
Assumption~\ref{item:property-Ast2} of \Cref{thm:fair-cut} holds for our choice of parameters $\epsilon,\gamma,\beta$. That is, we have $\mathbf t(C\cap A)+\epsilon\gamma\cdot \delta_G(C\cap A)\le\beta\cdot \delta_GC$ for all $C\in\mathcal C$.
\end{claim}

From \Cref{thm:fair-cut}, we obtain an $(1+\epsilon)$-approximate fair cut/flow pair $(S,f)$. The algorithm sets $B\gets S\setminus\{s,x\}$, which trivially takes $O(|A|)$ time outside the call. It remains to show that properties~(\ref{item:trimming-1}) to~(\ref{item:trimming-4}) are satisfied.

To show property~(\ref{item:trimming-1}), we start with $\delta_GB\le\delta_HS$ and proceed by bounding $\delta_HS$. By \Cref{fact:fair-cut}, the cut value $\delta_HS$ is at most $(1+\epsilon)$ times the minimum $(s,t)$-cut. Since the $(s,t)$-minimum cut is at most $\delta_H\{s\}=\mathbf s(A)$, we have
\[ \delta_HS\le(1+\epsilon)\mathbf s(A)=(1+\epsilon)(\mathbf s_0(A)+\epsilon\phi\mathbf d(A))\le2\mathbf s_0(A)+2\epsilon\phi\mathbf d(A) .\]
By definition of $\mathbf s_0$, we have $\mathbf s_0(A)=\sum_{v\in A\setminus R}c_G(\{v\},R)=c_G(A\setminus R,R)=\delta_{G[A]}R$. Putting everything together,
\[ \delta_GB\le\delta_HS\le2\mathbf s_0(A)+2\epsilon\phi\mathbf d(A) =2\delta_{G[A]}R+2\epsilon\phi\mathbf d(A) ,\]
fulfilling property~(\ref{item:trimming-1}).
By construction of $H$, we also have $c_H(S,\{t\})=12\phi\mathbf d|_{A\setminus R}(B)$, so
\[ 12\phi\mathbf d(B\setminus R)=12\phi\mathbf d|_{A\setminus R}(B)=c_H(S,\{t\})\le\delta_HS\le2\delta_{G[A]}R+2\epsilon\phi\mathbf d(A)  ,\]
and dividing by $12\phi$ establishes property~(\ref{item:trimming-2}).

The proofs of properties~(\ref{item:trimming-3}) and~(\ref{item:trimming-4}) are longer, so we package them into the two claims below.

\begin{claim}\label{clm:properties-ft-2}
Property~(\ref{item:trimming-4}) holds, i.e., there exists a vector $\mathbf t\in\mathbb R^A_{\ge0}$ with $\mathbf t\le24\phi\mathbf d|_{A\setminus(R\cup B)}$ and a flow $g$ on $G[A\setminus(R\cup B)]$ routing demand $\textup{deg}_{\partial_{G[A]}(R\cup B)}|_{A\setminus(R\cup B)}-\mathbf t$ with congestion $2$.
\end{claim}
\begin{proof}
Since $(S,f)$ is a $(1+\epsilon)$-fair cut/flow pair, each vertex $v\in V_H\setminus(S\cup\{t\})$ receives at least $\frac1{1+\epsilon}\,c_H(\{v\},S)\ge\frac12 c_H(\{v\},S)$ total flow from vertices in $S$. Since $A\setminus(R\cup B)\subseteq V_H\setminus(S\cup\{t\})$, the same holds for all $v\in A\setminus(R\cup B)$. By construction of $H$, we have the following for all $v\in A\setminus(R\cup B)$:
\begin{align*}
c_H(\{v\},S)&\ge c_H(v,s)+c_H(\{v\},B)
\\&=\mathbf s_0(v)+\epsilon\phi\mathbf d(v)+c_G(\{v\},B)
\\&=c_G(\{v\},R)+\epsilon\phi\mathbf d(v)+c_G(\{v\},B)
\\&\ge c_G(\{v\},R\cup B)+\epsilon\phi\mathbf d(v)
\\&=\deg_{\partial_{G[A]}(R\cup B)}(v)+\epsilon\phi\mathbf d(v) .
\end{align*}
It follows that each vertex $v\in A\setminus(R\cup B)$ receives at least $\frac12(\deg_{\partial_{G[A]}(R\cup B)}(v)+\epsilon\phi\mathbf d(v))$ total flow from vertices in $R\cup B$.

We now investigate the effect of restricting the flow $f$ to edges in $G[A\setminus(R\cup B)]$, starting with removing all edges incident to $R\cup B$. Continuing the argument above, removing these edges causes each vertex $v\in A\setminus(R\cup B)$ to be the source of at least $\deg_{\partial_{G[A]}(R\cup B)}(v)+\epsilon\phi\mathbf d(v)$ flow.

If $x\notin S$, then we now remove the edges incident to $x$. By construction of $H$, each vertex $v\in A$ has an edge to $x$ of capacity $\gamma\mathbf d(v)=\frac{\epsilon\phi}{2\kappa}\mathbf d(v)$. Since the flow has congestion $\kappa$, there is at most $\frac{\epsilon\phi}2\mathbf d(v)$ flow along the edge $(v,x)$. Removing this edge changes the net flow out of $v$ by at most $\frac{\epsilon\phi}2\mathbf d(v)$. Since each vertex $v\in A\setminus(R\cup B)$ is the source of at least $\frac12(\deg_{\partial_{G[A]}(R\cup B)}(v)+\epsilon\phi\mathbf d(v))$ flow before this step, it is the source of at least $\frac12\deg_{\partial_{G[A]}(R\cup B)}(v)$ flow after this step.

At this point, only the vertices $A\setminus(R\cup B)$ and $t$ remain, and each vertex $v\ne t$ is the source of at least $\frac12\deg_{\partial_{G[A]}(R\cup B)}(v)$ flow. Scale the flow by factor $2$, take a path decomposition of the flow, and remove paths until each vertex $v\ne t$ is the source of exactly $\deg_{\partial_{G[A]}(R\cup B)}(v)$ flow. We may assume that the flow does not send any flow away from $t$ (i.e., along any edge $(v,t)$ in the direction from $t$ to $v$), since otherwise we can cancel such flow using the path decomposition.  

Finally, we remove the edges incident to $t$. Since the original flow $f$ is feasible, the scaled flow has congestion $2$, so each edge $(v,t)$ carries at most $2c_H(v,t)$ flow. By construction of $H$, we have $c_H(v,t)=12\phi\mathbf d(v)$ for all $v\in A\setminus R$, so each vertex $v\in A\setminus(R\cup B)$ receives a net flow of at most $24\phi\mathbf d(v)$ from vertices other than $t$ (and then sends that flow to $t$). Let $\mathbf t(v)\le24\phi\mathbf d(v)$ be the net flow received, so that the vector $\mathbf t\in\mathbb R^{A\setminus(R\cup B)}_{\ge0}$ satisfies $\mathbf t\le24\phi\mathbf d|_{A\setminus(R\cup B)}$. Removing the edge $(v,t)$ increases the net flow into $v$ by exactly $\mathbf t(v)$. Since each vertex $v\in A\setminus(R\cup B)$ is the source of exactly $\deg_{\partial_{G[A]}(R\cup B)}(v)$ flow before removing the edge $(v,t)$, the vertex $v$ has net flow $\deg_{\partial_{G[A]}(R\cup B)}(v)-\mathbf t(v)$ after removal. In other words, the new flow routes demand $\deg_{\partial_{G[A]}(R\cup B)}(v)-\mathbf t(v)$ and has congestion $2$. By construction, it is also restricted to $G[A\setminus(R\cup B)]$, concluding the proof.
\end{proof}

\begin{claim}\label{clm:for-any-x-construct-y}
Property~(\ref{item:trimming-3}) holds, i.e., if the vertex weighting $\mathbf d|_{A\setminus R}$ mixes in $G[A]$ with congestion $c$, then the vertex weighting $(\mathbf d+\deg_{\partial_{G[A]}(R\cup B)})|_{A\setminus(R\cup B)}$ mixes in $G[A]$ with congestion $2+(1+24\phi)c$.
\end{claim}
\begin{proof}
Consider any demand $\mathbf b\in\mathbb R^A$ satisfying $|\mathbf b|\le(\mathbf d+\deg_{\partial_G(R\cup B)})|_{A\setminus(R\cup B)}$. In particular, $\sum_{v\in R\cup B}\mathbf b(v)=0$. To fulfill property~(\ref{item:trimming-3}), we want to construct a single-commodity flow routing demand $\mathbf b$ with congestion $2+(1+24\phi)c$.

We first split $\mathbf b$ into demands $\mathbf x$ and $\mathbf b-\mathbf x$, where vector $\mathbf x\in\mathbb R^A_{\ge0}$ is defined as
\[ \mathbf x(v)=
\begin{cases}
\displaystyle\frac{\deg_{\partial_G(R\cup B)}(v)}{\mathbf d(v)+\deg_{\partial_G(R\cup B)}(v)}\cdot \mathbf b(v)&\text{if }v\in A\setminus(R\cup B)\text{ and }\mathbf d(v)+\deg_{\partial_G(R\cup B)}(v)>0,\\
0&\text{otherwise}.
\end{cases}\]
Since $|\mathbf b|\le(\mathbf d+\deg_{\partial_G(R\cup B)})|_{A\setminus(R\cup B)}$, we have $|\mathbf x|\le\text{deg}_{\partial_G(R\cup B)}|_{A\setminus(R\cup B)}$ and $|\mathbf b-\mathbf x|\le\mathbf d|_{A\setminus(R\cup B)}$.

Take the flow $g$ from property~(\ref{item:trimming-3}), and compute a path decomposition of the flow. For each path in the decomposition from vertex $u\in A\setminus(R\cup B)$ to vertex $v\in A\setminus(R\cup B)$ of capacity $c$, route $\frac{\mathbf x(u)}{\deg_{\partial_G(R\cup B)}(u)}\cdot c$ flow along the path from $u$ to $v$ (or $-\frac{\mathbf x(u)}{\deg_{\partial_G(R\cup B)}(u)}\cdot c$ flow in the reversed direction, whichever is nonnegative). Since $|\mathbf x(u)|\le\text{deg}_{\partial_G(R\cup B)}(u)$, we send at most $c$ flow along each such path of capacity $c$, so our new flow $g_1$ has congestion at most that of $g$, which is $2$. Each vertex $v\in A\setminus(R\cup B)$ is the end of at most $\mathbf t(v)\le24\phi\mathbf d(v)$ total capacity of paths in the decomposition, and for each such path of capacity $c$, the vertex $v$ receives a net flow of at most $c$ in absolute value. It follows that each vertex $v\in A\setminus(R\cup B)$ receives at most $24\phi\mathbf d(v)$ net flow in absolute value. Setting $\mathbf y(v)$ as this net flow, we obtain $|\mathbf y|\le24\phi\mathbf d|_{A\setminus(R\cup B)}$ and that flow $g_1$ route demand $\mathbf x-\mathbf y$.

Having routed demand $\mathbf x-\mathbf y$ through flow $g_1$, it remains to route demand $\mathbf b-\mathbf x+\mathbf y$. We bound
\[ |\mathbf b-\mathbf x+\mathbf y|\le|\mathbf b-\mathbf x|+|\mathbf y|\le\mathbf d|_{A\setminus(R\cup B)}+24\phi\mathbf d|_{A\setminus(R\cup B)}\le(1+24\phi)\mathbf d|_{A\setminus(R\cup B)}.\]
By assumption, the vertex weighting $(\mathbf d+\deg_{\partial(R\cup B)})|_{A\setminus R}$ mixes in $G$ with congestion $c$, so there is a flow routing demand $(\mathbf b-\mathbf x+\mathbf y)/(1+24\phi)$ with congestion $c$. Scaling this flow by factor $1+24\phi$, we obtain a flow $g_2$ routing demand $\mathbf b-\mathbf x+\mathbf y$ with congestion $(1+24\phi)c$. Summing the two flows, the final flow $g_1+g_2$ routes demand $\mathbf b$ with congestion $2+(1+24\phi)c$, concluding the proof.
\end{proof}
 and\label{item:condition-St}

\end{document}